\documentclass[preprint,12pt]{elsarticle}

\usepackage{amsmath,amsfonts,bm}

\def\1{\bm{1}}

\DeclareMathAlphabet{\mathsfit}{\encodingdefault}{\sfdefault}{m}{sl}
\SetMathAlphabet{\mathsfit}{bold}{\encodingdefault}{\sfdefault}{bx}{n}

\usepackage{amssymb}
\usepackage{amsmath}
\usepackage{amsthm}
\usepackage{mathrsfs}
\usepackage{dsfont}
\usepackage{bbm}
\usepackage{multirow}
\usepackage{arydshln}
\usepackage{hyperref}

\newtheorem{thm}{Theorem}[section]
\newtheorem{defn}{Definition}[section]
\newtheorem{lemma}{Lemma}[section]
\newtheorem{obs}{Observation}[section]
\newtheorem{prop}{Proposition}[section]

\newtheorem{assumption}{Assumption}[section]
\newtheorem{remark}{Remark}[section]
\newtheorem{corollary}{Corollary}[section]

\journal{Advances in Mathematics}

\begin{document}
\begin{frontmatter}

\title{The Empirical Spectral Distribution of {i.i.d.\!} Random
Matrices with Random Perturbations}

\author[pku]{Kun Chen}
\ead{kchen0415@pku.edu.cn}
\author[pku]{Zhihua Zhang\corref{cor1}}
\cortext[cor1]{Corresponding author}
\affiliation[pku]{organization={School of Mathematical Sciences},
            addressline={Peking University},
            postcode={100871}, 
            city={Beijing},
            country={China}}
\ead{zhzhang@math.pku.edu.cn}

\begin{abstract}
A large i.i.d.\ random matrix with deterministic low-rank perturbation has been extensively studied, particularly in the aspects of the ESD (Empirical Spectral Distribution) and the outliers of eigenvalues. In this work, we investigate the analogous scenario where the perturbation is \textit{random}, and extend the previous results from the deterministic perturbation to the random case. Specifically, we consider an i.i.d.\ matrix with random perturbation, $\bm{M}$. Our results show that: (i) the eigenvalue outliers of $\bm{M}$ converge to the eigenvalues of its perturbation; (ii) the ESD of $\bm{M}$ converges to the circular law; (iii) the eigenvector alignment holds for specific perturbations.

As an application of the above random matrices, we present the first optimal \textit{query complexity} lower bound for approximating the top eigenvector of asymmetric matrices. In the inverse polynomial accuracy regime, the complexity matches the upper bounds that can be obtained via the power method. As far as we know, it is the first lower bound for approximating the eigenvector of an asymmetric matrix. 

\end{abstract}



\end{frontmatter}
\newpage
\tableofcontents
\newpage

\section{Introduction}\label{section: introduction}
The independent and identically distributed (i.i.d.) random matrix (i.e., its entries are i.i.d.) is of significant importance and has been extensively studied in the field of random matrix theory. Among the extensive research conducted on random matrices, previous works primarily revolved around aspects like the Empirical Spectral Distribution (ESD) \cite{tao2008random, tao10random, bai2010spectral, Philip12sparse}, the norm \cite{Bai1986LimitingBO, rudelson2009smallest, Vershynin2018hdp}, and the eigenvectors \cite{ORourke14lowrank, knowles2017anisotropic}. Meanwhile, such a random matrix occupies a pivotal position due to its numerous practical applications that are intricately linked to various fields, such as dynamical systems, high-dimensional statistics, signal processing, and machine learning \cite{johnstone2001distribution, candes2008introduction, pennington2017geometry, coston2020gaussian, erdHos2017dynamical, rajan2006eigenvalue}.

In this paper, we mainly explore two aspects of the i.i.d.\ matrix. In the first aspect, we focus on the eigenvalues and eigenvectors of the i.i.d.\ matrix with low-rank perturbations, which has also been extensively investigated \cite{tao2014outliers, tao10random, rajan2006eigenvalue, ORourke14lowrank}. While the literature mainly considers the deterministic perturbations, we consider the random case. More specifically, we consider an i.i.d.\ matrix with a random perturbation $\bm{M}=\frac{1}{\sqrt{d}}{\bm{W}}+{\bm{C}}$, where the entries of ${\bm{W}}$ are complex i.i.d.\ variables with zero mean, unit variance and finite fourth moment and $\bm{C}$ be an independent random matrix with bounded rank and $\ell_2$ norm. We study the eigenvalue outliers of $\bm{M}$, showing that they will converge to the eigenvalues of $\bm{C}$. This result is similar to the result of \cite{tao2014outliers} which considers the eigenvalue outliers of the i.i.d.\ matrix with fixed perturbation, and hence extending their work. For a specific form of matrix $\bm{C}$, we show that its ESD converges to the circular law. Meanwhile, the eigenvector of $\bm{M}$ will align with $\bm{C}$.

In the second aspect, we consider the application of i.i.d.\ matrices with random perturbations in randomized numerical linear algebra. We establish the first optimal query-complexity lower bound for approximating the top eigenvector of an asymmetric matrix. This lower bound is established by analyzing the eigenvector of $\frac{1}{\sqrt{d}}{\bm{W}}+\lambda{\bm{u}}{\bm{u}}^*$, which is a special case of i.i.d.\ matrices with random perturbations we consider above ($r=1, \lambda>1$ and ${\bm{u}}$ is uniformly drawn from $\operatorname{Stief}(d,1)$). We refer to the framework of \cite{DBLP:journals/corr/abs-1804-01221}, which provided an optimal lower bound for symmetric matrices. We extend their framework and can handle complex and asymmetric matrices.

\paragraph{Article Structure.} In the remaining part of this section, we will summarize our contribution, main techniques, related works, and notions. The remainder of this paper is organized as follows. The results on eigenvalues and eigenvectors of i.i.d.\ matrices with random perturbations are given in Section~\ref{section: deforemed i.i.d. random matrix}. The results on the lower bound for the eigenvector approximation are given in Section~\ref{section: application}. The corresponding upper bound are presented in Section~\ref{section: power method}. The missing proofs are given in Section~\ref{section: proof of deformed random matrix}-\ref{section: i-t lower bound}. We conclude our work in Section~\ref{section: summary}.

\subsection{Properties of i.i.d.\ Random Matrices with Random Perturbations}

Random matrices with structured low-rank perturbations, commonly referred to as spiked random matrices (or signal-plus-noise matrices), have been a central topic in modern random matrix theory and digital signal detection \cite{Jung2023DetectionPI, BENAYCHGEORGES2011494}. These models, typically formulated as $\bm{M}=\frac{1}{\sqrt{d}}\bm{W}+\bm{C}$, where $\bm{C}$ is a low-rank signal matrix and $\bm{W}$ is a noise matrix, aim to characterize the spectral and eigenstructure deviations caused by the perturbation $E$. Extensive studies have focused on several key aspects, such as the empirical spectral distribution (ESD) of $\bm{M}$, the eigenvalue outliers, and the alignment properties of eigenvectors or singular vectors with the signal subspace \cite{tao2014outliers, 2000On}.

Depending on the noise matrix $\bm{W}$, spiked models are categorized into several canonical classes such as spiked Wignar matrices, spiked Wishart matrices and signal-plus-noise matrix models \cite{10.1093/biomet/asy070}. While symmetric noise models (e.g., GOE/Wishart) dominate applications such as covariance estimation and network analysis \cite{DBLP:journals/corr/abs-1804-01221, 10.1093/biomet/asy070, PAUL20141, 10.5555/3122009.3242064}, asymmetric noise scenarios where $\bm{W}$ is asymmetric receive distinct theoretical challenges. These include the absence of spectral symmetry (e.g., non-real eigenvalues in real matrices) and the failure of classical tools like the semicircle law. Despite their prevalence in real-world systems (e.g., directed networks, non-Hermitian quantum systems), asymmetric spiked models remain understudied compared to their symmetric counterparts. In this work, we maily focus on the asymmetric case. For the ESD for i.i.d.\ $\bm{W}$, the circular law governs the limiting ESD of $\bm{M}$, asserting that eigenvalues distribute uniformly on the unit disk in $\mathbb{C}$. Remarkably, deterministic or low-rank perturbations $E$ do not alter this macroscopic limit \cite{tao10random, spectrum2010Bor}. For eigenvalue outliers, when the signal strength $\|\bm{C}\|_2$ exceeds a critical threshold, isolated eigenvalues (spikes) emerge outside the bulk. Their locations and fluctuations are characterized by deformed Tracy-Widom laws or deterministic shifts, depending on $\bm{C}$'s structure \cite{tao2014outliers, Bordenave2014OutlierEF}. The angles between spiked eigenvectors and the signal subspace exhibit phase transitions. For symmetric noise, near-orthogonality persists below the BBP threshold \cite{Baik2004PhaseTO}, while asymmetric settings reveal subtler alignment patterns influenced by the noise's singular value statistics \cite{Seddik2021WhenRT}.

We study the deformed i.i.d.\ matrix $\bm{M}=\frac{1}{\sqrt{d}}{\bm{W}}+\bm{C}$ with random perturbation $\bm{C}$ and rank$(\bm{C})\le r, \|\bm{C}\|_2\le B$. Our contribution in the aspects of the random matrix theory can be summarized as:
\begin{enumerate}[(i)]
    \item \textbf{Eigenvalue Outliers}: We show that the top-$r$ eigenvalues of $\bm{M}$ converge to the eigenvalues of $\bm{C}$. Thus we extend the results of \cite{tao2014outliers} from deterministic to random perturbations. Along the proof, we extend the Hanson-Wright inequality to asymmetric and complex settings, enabling non-asymptotic concentration results of eigenvalue outliers.

    \item \textbf{Empirical Spectral Distribution}: For a specific form of perturbations $\bm{C}={\bm{U}}\mathbf{\Lambda}{\bm{U}}^*$ where $\bm{U}$ is an independent random matrix over $\operatorname{Stief}_\mathbb{C}(d,r)$, and $\mathbf{\Lambda}=\operatorname{diag}(\lambda_1, \lambda_2, \ldots, \lambda_r)$ is diagonal with $\lambda_1\ge\lambda_2\ge \ldots\ge \lambda_r>1$, we show that the ESD of $\bm{M}$ converges to the circular law. As a direct corollary, we know the magnitude of the $(r+1)$-th eigenvalue of $\bm{M}$ concentrates around $1$.

    \item \textbf{Eigenvector Alignment}: For the case $r=1$ (i.e., $\bm{M}=\frac{1}{\sqrt{d}}\bm{W}+\lambda\bm{u}\bm{u}^*$), we prove that the top eigenvector (namely $\bm{v}_1(\bm{M})$, i.e., the eigenvector corresponding to the top eigenvalue) aligns with $\bm{u}$ in the sense that $\left|\langle\bm{v}_1(\bm{M}), \bm{u}\rangle\right|$ converges to a limit which depends on $\lambda$ (in probability).
\end{enumerate}

As we will see, our results on the random matrix theory are asymptotic, because we utilize some asymptotic tools when dealing with the random matrix ${\bm{W}}$, such as the strong circular law (Theorem~\ref{thm: strong circular law}, \cite{tao2008random}). However, some non-asymptotic results are established \cite{DBLP:journals/corr/abs-1804-01221} in the case of the deformed Gaussian Orthogonal Ensemble (GOE). Thus, we believe that it is possible to provide a non-asymptotic result using concentration tools and the Gaussian assumption. Meanwhile, in the apsect of eigenvector alignment, \cite{JMLR:v23:21-1038} also provides another approach to compute the alignment, but requires the matrix to be symmetric with Gaussian variables. Furthermore, although we mainly prove the alignment for $r=1$, we also provide a similar roadmap for the case $r>1$, where a similar alignment holds.

\subsection{Application: A Lower Bound for Eigenvector Approximation}
As for the application, a ubiquitous problem in numerical linear algebra is to approximate matrix functions that can only be accessed through matrix-vector multiplication queries, such as trace \cite{meyer2021hutch++}, diagonal \cite{Baston2022StochasticDE}, eigenvector \cite{Musco2015RandomizedBK}, etc. Among these problems, computing eigenvectors of a matrix is a fundamental theme and has widespread applications in machine learning, statistics, and computational mathematics fields \cite{trefethen2022numerical}. The symmetric matrix case has been extensively studied in the literature. Iterative methods such as Blocked Krylov methods \cite{Musco2015RandomizedBK}, the Arnoldi iteration \cite{greenbaum1994gmres}, the accelerated power method \cite{Xu2017AcceleratedSP}, and the Lanczos iteration \cite{komzsik2003lanczos} provide efficient algorithms with the tight query complexity lower bounds \cite{DBLP:journals/corr/abs-1804-01221, braverman2020gradient}. However, there are fewer algorithms specifically designed for the asymmetric case, such as the power method, and research on the lower bound for these algorithms remains open. 

On the other hand, there are several applications concerning the eigenvectors of an asymmetric matrix, such as Google's PageRank algorithm \cite{chung2014brief}, decentralized consensus \cite{bodkhe2020survey}, and other stable distributions of transition matrices. However, the theorem gaurantee for solving the eigenvector of an asymmetric matrix lacks, especially the lower bound. We note that we cannot directly obtain the eigenvector of an asymmetric matrix by some symmetrization routine, because we are studying the eigenvectors instead of singular vectors. Thus we need to consider the asymmetric matrix itself to obtain a lower bound for computing. More specifically, we are concerned with the eigenvector approximation problem of a matrix ${\bm{M}}\in \mathbb{R}^{d\times d}$, and we want to provide a tight lower bound for approximating the eigenvector corresponding to the largest eigenvalue among those algorithms that apply the matrix-vector product queries. Our motivation comes from \cite{DBLP:journals/corr/abs-1804-01221}, who proved that a query model needs at least $\mathcal{O}\left(\frac{r\log(d)}{\sqrt{\operatorname{gap}({\bm{M}})}}\right)$ queries to obtain the rank-$r$ eigenspace approximation of a symmetric positive definite matrix in $\mathbb{R}^{d\times d}$.

In this application, our goal is to establish query-complexity lower bounds for approximating the top eigenvector of an asymmetric matrix. Specifically, we consider randomized adaptive algorithms $\operatorname{Alg}$ which access an unknown asymmetric matrix ${\bm{M}}\in\mathbb{R}^{d\times d}$ via $T$ queries of the form
\[\left\{({\bm{v}}^{(t)}, \bm{w}^{(t)})\colon  \bm{w}^{(t)} = {\bm{M}}{\bm{v}}^{(t)}\right\}_{t\in [T]}.\] Let $\operatorname{gap}({\bm{M}})\colon =\frac{|\lambda_1({\bm{M}})|-|\lambda_2({\bm{M}})|}{|\lambda_1({\bm{M}})|}$ denote the normalized eigen-gap of ${\bm{M}}$. Let ${\bm{v}}_1({\bm{M}})$ be one of the right eigenvectors of ${\bm{M}}$ corresponding to the largest eigenvalue $\lambda_1({\bm{M}})$ in the absolute value. Then our main theorem can be informally formulated as follows (the formal version will be stated in Section~\ref{section: application}):

\begin{thm}\label{thm: main theorem informally} (Main theorem for application, informally). Given a fixed $\operatorname{gap}\in(0,1/2]$ and $\delta>0$, 
there exist a random matrix ${\bm{M}}\in\mathbb{R}^{d\times d}$ for which $\operatorname{gap}\left({\bm{M}}\right)\ge\frac{\operatorname{gap}}{2}$ with high probability, and an integer $d_0\in\mathbb{N}$ that depends on  $\operatorname{gap}$ and $\delta$, such that for any $d\ge d_0$ if an algorithm $\operatorname{Alg}$ makes $T\le \frac{c_1\log d}{\operatorname{gap}}$ queries, then with probability at least $1-\delta$, it cannot identify a unit vector $\hat{\bm{v}}\in \mathbb{C}^d$ satisfying $\|\hat{\bm{v}}-{\bm{v}}_1({\bm{M}})\|_2^2\le c_2\operatorname{gap}$. Here $c_1$ and $c_2$ are universal constants.
\end{thm}

Theorem shows that there are bad matrices (distribution of matrices) whose eigenvectors are difficult to distinguish. Thus provide a lower bound for the problem of approximating the eigenvector of a non-symmetric matrix. Our bound $\mathcal{O}\left(\frac{\log d}{\operatorname{gap}}\right)$ is tight about the parameter $\operatorname{gap}$ in the inverse polynomial precision regime where $\epsilon \ge \frac{1}{\operatorname{poly}(d)}$. Here, $\epsilon$ is the approximate accuracy. The corresponding upper bound, $\mathcal{O}\left(\frac{\log\left(\kappa d/\epsilon\right)}{\operatorname{gap}}\right)$, has been achieved by the well-known power method \cite{doi:10.1137/1.9781611971446}. For completeness, we provide an upper bound of the power method in Section~\ref{section: power method} as well as some examples showing how close the upper and lower bounds are. Specifically, there is another parameter in the upper bound, namely the condition number $\kappa$, which is due to initialization and is beyond our scope of consideration.

\paragraph{Comparison with the Symmetric Case} The different choice of the random matrix leads to several differences between the symmetric and asymmetric cases \cite{DBLP:journals/corr/abs-1804-01221}. For the dimension and eigenspace in the above theorem, the requirement of $d$ in \cite{DBLP:journals/corr/abs-1804-01221} is a polynomial of parameters while the exact $d$ could be unknown in our main theorem because we apply some asymptotic results in our proof. Moreover, we only consider the case $r=1$ for the eigenspace. The results with both the exact form of dimension $d$ and the case $r>1$ are difficult to obtain because there lacks the related result with the random matrix in the literature.

We now explain more about the approximation error. For a positive definite matrix ${\bm{M}}$, it is appropriate to consider the problem of finding a vector $\hat{\bm{v}}$ such that $\hat{\bm{v}}^*{\bm{M}}\hat{\bm{v}}\ge (1-\epsilon)\lambda_1({\bm{M}})$. 
In this case, there are two types of upper bounds $\mathcal{O}\left(\frac{\log d/\epsilon}{\sqrt{\operatorname{gap}({\bm{M}}))}}\right)$ and a $\operatorname{gap}$-free complexity $\mathcal{O}\left(\frac{\log d}{\sqrt{\epsilon}}\right)$ \cite{Musco2015RandomizedBK}. By setting $\epsilon =\operatorname{gap}({\bm{M}})$, the lower bound of \cite{DBLP:journals/corr/abs-1804-01221} matches both upper bounds. Note that the criterion $\hat{\bm{v}}^*{\bm{M}}\hat{\bm{v}}\ge (1-\epsilon)\lambda_1({\bm{M}})$ does not fit an asymmetric matrix ${\bm{M}}$ because $\sup_{\|{\bm{v}}\|_2=1}{\bm{v}}^*{\bm{M}}{\bm{v}}=\left\|\left({\bm{M}}+{\bm{M}}^*\right)/2\right\|_2$ can be larger than $\lambda_1({\bm{M}})$ in the asymmetric case \cite{golub2013matrix}, which is also the main difficulty in our approximation. However, when the approximation error $\|\hat{\bm{v}}-{\bm{v}}_1({\bm{M}})\|^2_2$ is small, $\hat{\bm{v}}^*{\bm{M}}\hat{\bm{v}}$ is around $\lambda_1({\bm{M}})$. Thus, we instead use the criterion $\|\hat{\bm{v}}-{\bm{v}}_1({\bm{M}})\|^2_2\le \operatorname{const}\times \operatorname{gap}$ in our work. According to this criterion, there exists only one upper bound $\mathcal{O}\left(\frac{\log d/\epsilon}{{\operatorname{gap}({\bm{M}}))}}\right)$ which is obtained using the power method \cite{doi:10.1137/1.9781611971446}. Our lower bound matches this upper bound for the parameter $\operatorname{gap}$ as well as the polynomial accuracy regime. Similarly, setting $\epsilon = \operatorname{gap}({\bm{M}})$, we have a lower bound $\tilde{\mathcal{O}}\left(\frac{\log d}{\epsilon}\right)$. However, we still need a suitable criterion to obtain a better upper bound.

One might consider a more general situation that the algorithm $\operatorname{Alg}$ accesses the vectors ${\bm{M}}{\bm{v}}$ and ${\bm{M}}^*{\bm{v}}$ at each iteration in the asymmetric case \cite{bakshi2022low}, while it is unnecessary in the symmetric case. 
However,  even in this situation, the asymmetric eigenvectors approximation problem can not be reduced to the symmetric case by  aligning the matrices ${\bm{M}}$ and ${\bm{M}}^*$ into a symmetric block matrix $\left[
\begin{array}{cc}
0 & {\bm{M}} \\
{\bm{M}}^* & 0
\end{array}
\right]$, 
because we are considering the eigenvectors rather than singular vectors and there does not exist an explicit relationship between the eigenvectors and the singular vectors of an asymmetric matrix. 
In this situation, fortunately,  we also obtain similar results which are not a simple generalization of the one-sided matrix-vector multiplication case. This requires us to use slightly different random matrices. 
For a detailed discussion and further results on the more general algorithm, we refer to Section~\ref{sec: two side situation}. 

Meanwhile, we note that query model with restriction to only one-sided matrix-vector products remains operationally reasonable in practical large-scale computing scenarios. This constraint aligns with the architectural realities of distributed numerical linear algebra systems, where massive matrices exceeding local memory capacity are typically partitioned across multiple storage nodes interconnected via pointer-based data structures \cite{10.1145/2312005.2312044, 10.1145/2780584}. In such implementations, full two-side matrix access incurs prohibitive communication costs, whereas one-side access can be efficiently parallelized through coordinated pointer traversal.

\paragraph{Techniques}
The roadmap of our results follows the idea of \cite{DBLP:journals/corr/abs-1804-01221}. We begin by computing the eigenvector of the i.i.d.\ matrix with a rank-1 perturbation ${\bm{M}}=\frac{1}{\sqrt{d}}{\bm{W}}+\lambda{\bm{u}}{\bm{u}}^*$, where $W_{ij}\stackrel{i.i.d.}{\sim}\mathcal{N}(0,1)$, ${\bm{u}}\stackrel{\text{unif}}{\sim}\operatorname{Stief}(d,1)$, and $\lambda=\frac{1}{1-\operatorname{gap}}$ is a parameter ensuring that $\operatorname{gap}({\bm{M}})$ concentrates around $\operatorname{gap}$. This distribution is exactly a special case of $\frac{1}{\sqrt{d}}{\bm{W}}+{\bm{U}}\mathbf{\Lambda}{\bm{U}}^*$ we mentioned earlier. We focus mainly on the case $\operatorname{gap}\to 0$. We note that the difference in the lower bounds between our work and \cite{DBLP:journals/corr/abs-1804-01221} exactly comes from the requirement of $\lambda$, which results in $\operatorname{gap}({\bm{M}})\approx\operatorname{gap}$. Our choice is $\lambda=\frac{1}{1-\operatorname{gap}}$, while the latter takes $\operatorname{gap}=\frac{\lambda+\lambda^{-1}-2}{\lambda+\lambda^{-1}}$, i.e., $\lambda=\mathcal{O}(\sqrt{\operatorname{gap}})$. Based on the properties of the deformed i.i.d.\ matrix, we then reduce the eigenvector estimation problem to an overlapping problem as in \cite{DBLP:journals/corr/abs-1804-01221}. We note that this reduction is exactly the eigenvector alignment in our result of the random matrix theory. As we can see, it is more complicated than in the symmetric case \cite{DBLP:journals/corr/abs-1804-01221}. 
After the reduction, we apply similar information-theoretic tools developed by \cite{DBLP:journals/corr/abs-1804-01221}. However, we need to modify their framework so that it can be used to deal with complex vectors and the asymmetric case.

\paragraph{Related Work}
The algorithm we consider in this paper is known as the matrix-vector product model or the query model. In this model, there is an \textit{implicit} ${\bm{M}}$ and the only access to ${\bm{M}}$ is via matrix-vector products ${\bm{M}}{\bm{v}}$ for any chosen ${\bm{v}}$. We refer to the number of queries needed by the algorithm to solve a problem with constant probability as the \textit{query complexity}. Similar definitions and applications are discussed in \cite{DBLP:journals/corr/abs-1804-01221, rashtchian_et_al:LIPIcs.APPROX/RANDOM.2020.26, braverman2020gradient, sun2021querying, bakshi2022low}. The upper bounds on the query complexity can be translated to running-time bounds for the RAM model because a matrix-vector product can be obtained in $\mathcal{O}\left(\mathrm{nnz}({\bm{M}})\right)$ time, i.e., the number of non-zero entries in the matrix. Thus, the model benefits from the sparsity \cite{bakshi2022low}.

There are a lot of works in the literature,  which provide the upper bounds for the eigenvector approximation. In the case of symmetric matrices, several efficient algorithms have been devised, such as the power method \cite{golub2013matrix}, Blocked Krylov \cite{Musco2015RandomizedBK}, Arnoldi iteration \cite{greenbaum1994gmres}, the Lanczos iteration \cite{komzsik2003lanczos} and GMRES \cite{saad1986gmres, FreitagKürschnerPestana+2018+203+222}   (See \cite{8187286, Woodruff2014SketchingAA} for surveys).  The best upper bound for the $\ell_2$-norm is due to \cite{Musco2015RandomizedBK}, as well as the tight lower bounds \cite{DBLP:journals/corr/abs-1804-01221, braverman2020gradient}. A similar problem is known as the low-rank approximation, where the Krylov method is also optimal \cite{bakshi2022low, 10353123}. A common technique in the analysis of the accelerated power method and Krylov methods is the Chebyshev polynomials \cite{Rivlin1990ChebyshevP}. In the asymmetric case, the analysis of the eigenvalue suffers from the instability of eigenvalue \cite{bhatia2013matrix}. The Chebyshev polynomials do not enjoy the good properties seen in the symmetric case \cite{doi:10.1137/090779486}. Methods like GMRES may not always work \cite{Anne1996Any}. Iterative methods such as the Rayleigh-Ritz method and Lanczos’s algorithm can be applied to this case, but their convergence usually depends on the matrix structure \cite{trefethen2022numerical}. While the power method is commonly used in the asymmetric case with convergence guarantee, the corresponding lower bound is also lacking. Our lower bounds extend a recent line of work on the lower bounds for linear algebra problems with the matrix-vector query model \cite{DBLP:journals/corr/abs-1804-01221, braverman2020gradient, meyer2021hutch++, Baston2022StochasticDE, bakshi2022low}.

\subsection{Notation}
We denote $[n]=\{1, \ldots, n\}$ for $n\in\mathbb{N}^+$. For $\bm{a}\in\mathbb{C}^d$ and $\bm{A}\in\mathbb{C}^{d\times d}$, $\bm{a}^*$ and $\bm{A}^*$ denote the complex adjoint. Let $\|\bm{A}\|_2, \|\bm{a}\|_2$ be the $\ell_2$-norm and $\|\bm{A}\|_F=\sqrt{\mathrm{tr}\left(\bm{A}^*\bm{A}\right)}$ be the $F$-norm.  Let $\left(\lambda_i(\bm{A})\right)_{i=1,\ldots,d}$ denote the eigenvalues of $\bm{A}$ with order $|\lambda_1(\bm{A})|\ge|\lambda_2(\bm{A})|\ge \cdots \ge |\lambda_d(\bm{A})|$ and $\left\{{\bm{v}}_i(\bm{A})\right \}_{i=1, \ldots,d}$ denote the right eigenvector corresponding to the eigenvalues $\lambda_i(\bm{A})$. We define $\text{gap}(\bm{A})\colon=\frac{|\lambda_1(\bm{A})|-|\lambda_2(\bm{A})|}{|\lambda_1(\bm{A})|}$ as the eigen-gap of the matrix $\bm{A}$. $\operatorname{Stief}(d, r)$ (and $\operatorname{Stief}_\mathbb{C}(d, r)$) denotes the Stiefel manifold consisting of matrices $\bm{V}\in\mathbb{R}^{d\times r}$ (and $\bm{V}\in\mathbb{C}^{d\times r}$) such that $\bm{V}^*\bm{V}=\bm{I}$ (and $\bm{V}^*\bm{V}=\bm{I}$). $\mathcal{S}^{d-1}$(and $\mathcal{S}_\mathbb{C}^{d-1}$) denotes $\operatorname{Stief}(d, 1)$ (and $\operatorname{Stief}_\mathbb{C}(d, 1)$) for convenience. Given a sequence of random variables $\{X_n\}_{n=1, 2, \ldots}$, we denote $X_n =o_{n,p}(1)$ if $\mathbb{P}\left(|X_n|\ge c\right)\to 0$ as $n \to \infty$ for an arbitrary $c>0$, and $X_n =\mathcal{O}_{n, p}(1)$ if $\mathbb{P}\left(|X_n|\ge C\right)\to 0$ as $n \to \infty$ for an absolute constant $C>0$. Further we say that $X_n$ has an $o_{n,p}(1)$ upper bound if there is a sequence of random variables $Y_n=o_{n,p}(1)$ such that $|Y_n|\ge |X_n|$ a.s.. Similarly, given a sequence of random vectors $\{\mathbf{x}_n\}_{n=1, 2, \ldots}$, we denote $\mathbf{x}_n =o_{n,p}(1)$ if $\mathbb{P}\left(\|\mathbf{x}_n\|_2\ge c\right)\to 0$ as $n \to \infty$ for an arbitrary $c>0$. Furthermore, for convenience in notation, we use $o(1)$ and $O(1)$ to represent $o_{n,p}(1)$ and $O_{n,p}(1)$, respectively. A $d$-dimensional complex random vector $\bm{x}:={\bm{y}} +\mathbf{i} {\bm{z}}$ in $\mathbb{C}^d$ is defined as a random vector $({\bm{y}}^T, {\bm{z}}^T)^T$ taking the value in $\mathbb{R}^{2d}$. We say $X$ is a standard complex Gaussian variable, or equivalently $X\sim\mathcal{CN}(0,1)$ if $X=Y+\mathbf{i} Z$ with independent $Y\sim\mathcal{N}\left(0,\frac{1}{2}\right)$ and $Z\sim\mathcal{N}\left(0,\frac{1}{2}\right)$. We denote ${\bm{U}}_{\bm{X}}$ the ESD of the matrix $\bm{X}$.

Because we will repeatedly consider random matrices in this paper, we give the following definition in advance for convenience.

\begin{defn}
Given a matrix $\bm{A} \in \mathbb{R}^{d \times d} $, we call it a Gaussian Orthogonal Ensemble (GOE) (denoted $\bm{A}  \sim \operatorname{GOE}(d)$) if $A_{ij} =A_{ji}$ for all $i, j \in [d]$,  the  $A_{i j}$ are i.i.d.\ from $ \mathcal{N}(0,1 / d)$ for $1\le i < j\le d$, and the $A_{i i}$ are i.i.d.\ from $ \mathcal{N}(0,2 / d)$ for $i\in  [d]$. We call $\bm{A}$ a Ginibre Orthogonal Ensemble (GinOE) (denoted $\bm{A} \sim \operatorname{GinOE}(d)$) if the  $A_{i j}$ are i.i.d.\ from $ \mathcal{N}(0,1 / d)$ for $i, j \in [d]$. We call $\bm{A}$ a Complex Ginibre Orthogonal Ensemble (CGinOE) (denoted $\bm{A} \sim \operatorname{GinOE}_\mathbb{C}(d)$) if the $A_{i j}$ are i.i.d.\ from $ \mathcal{CN}(0,1 / d)$ for $i, j \in [d]$.
\end{defn}

\section{i.i.d.\ Matrices with Random Perturbations}\label{section: deforemed i.i.d. random matrix}

In this section, we provide our results on eigenvalue outliers, ESD, and eigenvector alignments for i.i.d.\ matrices with random perturbations. Our main result for eigenvalue outliers is stated in Theorem~\ref{thm: outliers} and Theorem~\ref{thm: outliers, uniform}. The result for ESD is stated in Theorem~\ref{thm: esd of deformed iid} and the result for eigenvector alignments is stated in Theorem~\ref{thm: top eigenvector, non-Gaussian case}. In the remainder of this section, we will provide the roadmap for the theorems. 

Because we will repeatedly use the notation of (deformed) i.i.d.\ matrices, we give the definition here in advance.

\begin{defn}\label{defn: iid matrix}
    (i.i.d.\ matrix). We say a $d\times d$ matrix ${\bm{W}}$ is an i.i.d.\ matrix if the entries $W_{ij}$ are i.i.d.\ complex random variables with $\mathbb{E} W_{ij}=0, \mathbb{E}\left|W_{ij}\right|^2=1$ and $\mathbb{E}\left|W_{ij}\right|^4<\infty$ for all $i, j\in [d] $. Further if $W_{11}\sim\mathcal{N}(0,1)$ (or $W_{11}\sim\mathcal{CN}(0,1)$), we say $\bm{W}$ is a (complex) Gaussian i.i.d.\ matrix.
\end{defn}
\begin{defn}\label{defn: deformed iid matrix}
    (Deformed i.i.d.\ matrices). A real (or complex) $(d, r)$-deformed i.i.d.\ matrix is defined as $\bm{M}=\frac{1}{\sqrt{d}}\bm{W}+ \bm{U}\bm{\Lambda}\bm{U}^*$, where $\bm{\Lambda}=\operatorname{diag}\left(\lambda_1, \lambda_2,\dots,\lambda_r\right)$ is a fixed complex matrix with $|\lambda_1|\ge|\lambda_2|\ge\cdots\ge|\lambda_r|>1$. $\bm{U}\stackrel{\text{unif}}{\sim}\operatorname{Stief}(d, r)$ (or $\bm{U}\stackrel{\text{unif}}{\sim}\operatorname{Stief}_\mathbb{C}(d, r)$), and $\bm{W}$ is an i.i.d.\ matrix independent with $\bm{U}$. 
\end{defn}

We consider the matrix of the form ${\bm{M}}:=\frac{1}{\sqrt{d}}\bm{W}+\bm{C}$, where $\bm{W}$ is an i.i.d.\ matrix and the perturbation $\bm{C}$ is an independent matrix in $\mathbb{C}^{d\times d}$. As we are concerned about the eigenvalues and eigenvectors of the asymmetric matrix ${\bm{M}}$ that may not be diagonalizable, we make the following assumption for convenience.
\begin{assumption}
    We always assume that the eigenvalues of ${\bm{W}}$ and ${\bm{M}}$ are $2d$ distinct values in $\mathbb{C}$. Then, as a corollary, these two matrices are diagonalizable.
\end{assumption}

The assumption is reasonable because for most distributions (e.g., Gaussian distribution), the probability that the equation of $x$: $\operatorname{det}\left({\bm{W}}-x\right)\cdot\operatorname{det}\left({\bm{M}}-x\bm{I}\right)=0$ has two same roots equals $0$. Thus, the two matrices are diagonalizable with probability 1. 

We then present our main result for the eigenvalue outliers of i.i.d.\ matrices with random perturbations. It shows that the outliers converge to the eigenvalues of perturbations. We provide the proof roadmap in Section~\ref{sec: proof of thm: outliers}.
\begin{thm}
\label{thm: outliers}
Let $r\in\mathbb{N}^+$ be a fixed integer, $j\in[r]$, $\varepsilon>0$, $B>0$, $\bm{W}$ be an i.i.d.\ matrix, $\bm{C}$ be a $d\times d$ independent random complex matrix with $\operatorname{rank}(\bm{C})\le r$ and $\left\|\bm{C}\right\|_2\le B$, and $\bm{M}=\frac{1}{\sqrt{d}}\bm{W}+\bm{C}$. Suppose $|\lambda_{j+1}(\bm{C})|\le 1+\epsilon$ and $|\lambda_{j}(\bm{C})|\ge 1+3\epsilon$ hold for all sufficiently large $d$, and $\lambda_i\left(\bm{C}\right)\to\lambda_i\in\mathbb{C}$ in probability for $i\in[j]$. Then $|\lambda_{j}(\bm{M})|\ge 1+2\epsilon, |\lambda_{j+1}(\bm{M})|\le 1+2\epsilon$ hold for sufficiently large $d$, and $\lambda_{i}\left(\bm{M}\right)\to \lambda_i$ in probability for $i\in [j]$.
\end{thm}

In the above theorem, we only show the convergence in probability. We also provide another almost sure convergence for the uniform perturbations, which is stated in the following theorem. Meanwhile, we provide a potential approach for concentration along the proof. Due to the complexity of the proof, we leave it in Section~\ref{section: proof of deformed random matrix}.

\begin{thm}\label{thm: outliers, uniform}
    Let $r\in\mathbb{N}^+$ be a fixed integer. Let $\bm{M}=\frac{1}{\sqrt{d}}\bm{W}+ \bm{U}\bm{\Lambda}\bm{U}^*$ be a real (or complex) $(d, r)$-deformed i.i.d.\ matrix. Then $\lambda_{i}(\bm{M})\to\lambda_i$ and $\lambda_{r+1}(\bm{M})\le 1+\epsilon$ almost surely for $i\in [r]$ and any $\epsilon>0$.
\end{thm}

Although Theorem~\ref{thm: outliers} and Theorem~\ref{thm: outliers, uniform} are similar to the results of \cite{tao2014outliers}, a significant difference is that we consider random perturbations while \cite{tao2014outliers} considers the deterministic case. In Theorem~\ref{thm: outliers}, we modify the key lemma (Lemma 2.3) of \cite{tao2014outliers} to suit the random cases.

In the case of uniform perturbations, we extend the Hanson-Wright inequalities to the asymmetric and complex version as the symmetric approach does \cite{DBLP:journals/corr/abs-1804-01221}. Then we consider the Stieltjes transform of $\frac{1}{\sqrt{d}}{\bm{W}}$ and apply the strong circular law to bound the outlier eigenvalues, which leads to a different proof from \cite{tao2014outliers} for the outliers.

We then present our main result for the ESD of i.i.d.\ matrices with uniform perturbations. The theorem shows that the ESD of ${\bm{M}}$ will converge to the circular law. The result is a simple corollary of the rotational invariance of the Gaussian vectors and the results of \cite{tao10random}. We provide the proof roadmap in Section~\ref{sec: thm: ESD}. Although we consider a specific perturbation here, similar approach can be applied to general symmetric random perturbations.

\begin{thm}\label{thm: esd of deformed iid}
    Let $r\in\mathbb{N}^+$ be a fixed integer, $\bm{\Lambda}=\operatorname{diag}\left(\lambda_1, \lambda_2,\dots,\lambda_r\right)$ be a fixed complex matrix with $|\lambda_1|\ge|\lambda_2|\ge\cdots\ge|\lambda_r|>1$, $\bm{W}$ be an i.i.d.\ matirx, $\bm{U}$ be a independent random matirx over $\operatorname{Stief}_\mathbb{C}(d, r)$ and $\bm{M}=\frac{1}{\sqrt{d}}\bm{W}+ \bm{U}\bm{\Lambda}\bm{U}^*$. Then the empirical spectral distribution $\mu_{\bm{M}}$ converges to the uniform distribution over the unit disk in $\mathbb{C}$ almost surely. As a corollary, we have $|\lambda_{r+1}(\bm{M})|\ge 1+o(1)$ with $o(1)\to0$ almost surely.
\end{thm}

The next theorem focuses on the top eigenvector of ${\bm{M}}=\frac{1}{\sqrt{d}}{\bm{W}}+ \lambda{\bm{u}}{\bm{u}}^*$ with $\lambda>1$. Here $\bm{W}$ is a i.i.d.\ matirx and $\bm{u}$ is an independent random vector in $\mathcal{S}_{\mathbb{C}}^{d-1}$. We have known that the perturbation $\lambda{\bm{u}}{\bm{u}}^*$ has a certain effect on the largest eigenvalue, causing $\lambda_1({\bm{M}})\to\lambda_1(\lambda{\bm{u}}{\bm{u}}^*)=\lambda>1$ and $|\lambda_2({\bm{M}})|\to 1$. Intuitively, the perturbation may also have a certain effect on the top eigenvector such that ${\bm{v}}_1({\bm{M}})\to{\bm{v}}_1(\lambda{\bm{u}}{\bm{u}}^*)={\bm{u}}$. The following proposition provides the convergence of the inner product $\left|\langle{\bm{v}}_1({\bm{M}}),{\bm{u}}\rangle\right|$.

\begin{thm}\label{thm: top eigenvector, non-Gaussian case}
    Let $\bm{W}$ be an i.i.d.\ matirx, $\bm{u}$ be an independent random vector over $\mathcal{S}_{\mathbb{C}}^{d-1}$ and ${\bm{M}}=\frac{1}{\sqrt{d}}{\bm{W}}+\lambda{\bm{u}}{\bm{u}}^*$. Then we have $|{\bm{u}}^*{\bm{v}}_1({\bm{M}})|= \frac{\sqrt{\lambda^2-1}}{\lambda}+o(1)$.
\end{thm}

We will present the proof roadmap for Theorem~\ref{thm: top eigenvector, non-Gaussian case} in Section~\ref{sec: roadmap of top eigenvector}. Theorem~\ref{thm: top eigenvector, non-Gaussian case} is also known as the reduction in Section~\ref{section: application} and \cite{DBLP:journals/corr/abs-1804-01221}, and as the alignment in the spiked model. Once ${\bm{W}}$ is symmetric, the corresponding result is shown in \cite{DBLP:journals/corr/abs-1804-01221, JMLR:v23:21-1038}. Although we only provide the alignment for case $r=1$, we note that similar results hold for $r\ge 2$, the roadmap of which is given in Section~\ref{subsubsec: r>2}.

\subsection{Proof Roadmap of Theorem~\ref{thm: outliers}}
\label{sec: proof of thm: outliers}

We are now going to prove Theorem~\ref{thm: outliers}, which states that $\lambda_{i}({\bm{M}})=\lambda_i+o(1)$ for $i\in [j]$ and $|\lambda_{j+1}({\bm{M}})|\le 1+2\epsilon$ for large enough $d$.

By the assumption that $\operatorname{rank}(\bm{C})\le r$, we have a singular value decomposition $\bm{C}=\bm{A}\bm{B}^*$ with $\bm{A}, \bm{B}\in\mathbb{C}^{d\times r}$. By Lemma~\ref{lemma: largest eigenvalue of X}, we have $\left|\lambda_1\left(\frac{1}{\sqrt{d}}{\bm{W}}\right)\right|=1+o(1)$. In the following proof, we always consider the case under the event $\mathcal{E}_1=\left\{\left|\lambda_1\left(\frac{1}{\sqrt{d}}{\bm{W}}\right)\right|\le 1+\epsilon \right\}$. We then decompose the determinant as follows.
\[\det\left(z\bm{I}-\frac{1}{\sqrt{d}}{\bm{W}}-\bm{A}\bm{B}^*\right)=\operatorname{det}\left(z \bm{I}-\frac{1}{\sqrt{d}}{\bm{W}}\right)\cdot \operatorname{det}\left(\bm{I}_r-\bm{B}^*\left(z \bm{I}-\frac{1}{\sqrt{d}}{\bm{W}}\right)^{-1} \bm{A}\right)\]
for $|z|>1+\epsilon$ because $z \bm{I}-\frac{1}{\sqrt{d}}{\bm{W}}$ is invertible. Thus, we only need to consider the second determinant, whose roots are the eigenvalues of ${\bm{M}}$. By Roche's Theorem, once we have
\begin{equation}\label{eq: norm of (zI-W)-zI}
\sup _{|z| \geq 1+2 \varepsilon}\left\|\bm{B}^*\left(\left(z\bm{I}-\frac{1}{\sqrt{d}} \bm{W}\right)^{-1}-\frac{1}{z}\bm{I}\right) \bm{A}\right\|_2=o(1),
\end{equation}
the roots of $\operatorname{det}\left(\bm{I}_r-\bm{B}^*\left(z \bm{I}-\frac{1}{\sqrt{d}}{\bm{W}}\right)^{-1} \bm{A}\right)$ converge to the roots of $\operatorname{det}\left(\bm{I}_r-\bm{B}^*\bm{A}/z\right)$ in $\left\{z\ge 1+2\epsilon\right\}$, which are exactly the desired $\lambda_1, \cdots, \lambda_j$. To show Equation.(\ref{eq: norm of (zI-W)-zI}), we decompose the inverse matrix as
\[\left(z\bm{I}-\frac{{\bm{W}}}{\sqrt{d}}\right)^{-1}-\frac{\bm{I}}{z}=\frac{1}{z}\left(\frac{1}{\sqrt{d}}\frac{{\bm{W}}}{z}+\frac{1}{d}\frac{{\bm{W}}^2}{z^2}+\cdots+\left(\frac{\bm{W}}{z\sqrt{d}}\right)^{k-1}+\left(\frac{\bm{W}}{z\sqrt{d}}\right)^k\left(z\bm{I}-\frac{\bm{W}}{\sqrt{d}}\right)^{-1}\right).\]
Here $k$ remains to choose and thus we can upper bound the RHS of Equation.(\ref{eq: norm of (zI-W)-zI}) as
\begin{equation}
\left\|\bm{B}^*\left(\left(z\bm{I}-\frac{\bm{W}}{\sqrt{d}} \right)^{-1}-\frac{1}{z}\bm{I}\right) \bm{A}\right\|_2 \le \sum_{i=1}^{k-1} \left\|\bm{B}^*\left(\frac{\bm{W}}{z\sqrt{d}}\right)^i \bm{A}\right\|_2+B\left\|\left(\frac{\bm{W}}{z\sqrt{d}}\right)^k\left(z\bm{I}-\frac{\bm{W}}{\sqrt{d}}\right)^{-1} \right\|_2.
\end{equation}
By Lemma~\ref{lemma: inv(lambdaI-W)} and Lemma~\ref{lemma: theorem 5.17}, we can bound the last term in the above equation as
\[B\left\|\left(\frac{\bm{W}}{z\sqrt{d}}\right)^k\left(z\bm{I}-\frac{1}{\sqrt{d}} \bm{W}\right)^{-1} \right\|_2\le \frac{B(k+2)}{z^k}\left(\frac{2}{(1-1/z)^2}+1\right)\stackrel{(a)}{=}o_k(1)\]
with probability $1-o_d(1)$. Equation.(a) is because $|z|\ge 1+2\epsilon$. Then we remain to bound $\left\|\bm{B}^*\left(\frac{\bm{W}}{z\sqrt{d}}\right)^i \bm{A}\right\|_2$ for $i\in [k]$. Becasue both $\bm{A}, \bm{B}$ have rank $r$, we can break $\bm{A}, \bm{B}$ into components. It suffices to show the rank-1 case, then we finish the theorem by choosing a large enough $k$. For the rank-1 case, we have the following lemma, which extends the key lemma (Lemma 2.3) in \cite{tao2014outliers} to a random situation.
\begin{lemma}\label{lemma: uW^kv=o(1)}
    Let $\bm{W}$ be a $d\times d$ i.i.d.\ matrix, $m\in\mathbb{N}^+$, $\bm{v}$ and $\bm{u}$ be two random unit vectors (may be not independent with each other) which are independent with $\bm{W}$. Then $\left\langle \frac{\bm{W}^m}{d^{m/2}}\bm{u},\bm{v}\right\rangle=o(1)$.
\end{lemma}
\begin{proof}
By the standard truncation argument \cite{Bai1988Necessary}, we can suppose the entires of $\bm{W}$ to be bounded. Then it suffices to show $\mathbb{E}\left|\left\langle \frac{\bm{W}^m}{d^{m/2}}\bm{u},\bm{v}\right\rangle\right|^2=o_d(1)$. By Lemma~\ref{lemma: non-gaussian tech EW-EG=0} (similar argument as Lemma 2.3 in \cite{tao2014outliers}), we only need to bound the second moment of the target in the case where $W_{11}\sim\mathcal{CN}(0,1)$. To handle the Gaussian case, we have
\begin{equation}
\begin{aligned}
    &\mathbb{E}\left|\left\langle \frac{\bm{W}^m}{d^{m/2}}\bm{u},\bm{v}\right\rangle\right|^2=
    \mathbb{E}_{\bm{u},\bm{v}}\mathbb{E}_{\bm{W}}\left|\left\langle \frac{\bm{W}^m}{d^{m/2}}\bm{u},\bm{v}\right\rangle\right|^2\stackrel{(a)}{=}
    \mathbb{E}_{\bm{u},\bm{v}}\mathbb{E}_{\bm{W}}\left|\left\langle \frac{\bm{W}^m}{d^{m/2}}\bm{D}\bm{u},\bm{D}\bm{v}\right\rangle\right|^2\\
    &\stackrel{(b)}{=}
    \mathbb{E}_{\bm{u},\bm{v}}\mathbb{E}_{\bm{W}}\left(\left( \frac{\bm{W}^m}{d^{m/2}}\right)^*\bm{D}^*\bm{v}^*\bm{v}\bm{D}\frac{\bm{W}^m}{d^{m/2}}\right)_{11}\\
    &\stackrel{(c)}{=}\mathbb{E}_{\bm{u},\bm{v}}\mathbb{E}_{\bm{W}}\left(\left(|\bm{u}^*\bm{v}|^2-\frac{1-|\bm{u}^*\bm{v}|^2}{d}\right)\cdot\left| \frac{\bm{W}^m}{d^{m/2}}\right|_{11}^2 +\frac{1-|\bm{u}^*\bm{v}|^2}{d}\cdot\left(\left(\frac{\bm{W}^m}{d^{m/2}}\right)^*\frac{\bm{W}^m}{d^{m/2}}\right)_{11}\right)\\
    &\le\mathbb{E}_{\bm{W}}\left(\left| \frac{\bm{W}^m}{d^{m/2}}\right|_{11}^2 +\frac{1}{d}\cdot\left(\left(\frac{\bm{W}^m}{d^{m/2}}\right)^*\frac{\bm{W}^m}{d^{m/2}}\right)_{11}\right)=\mathbb{E}_{\bm{W}}\left( \frac{\bm{W}^m}{d^{m/2}}\right)_{11}^2 +\frac{1}{d^2}\mathbb{E}_{\bm{W}}\left\|\frac{\bm{W}^m}{d^{m/2}}\right\|_F^2\\
    &\le\mathbb{E}_{\bm{W}}\left( \frac{\bm{W}^m}{d^{m/2}}\right)_{11}^2 +\frac{1}{d}\mathbb{E}_{\bm{W}}\left\|\frac{\bm{W}}{\sqrt{d}}\right\|_F^{2m}\stackrel{(d)}{\le}\frac{c_m}{d}=o(1)
\end{aligned}
\end{equation}
for some constant $c_m$. Equation $(a)$ is because $\bm{W}$ has the same distribution as $\bm{D}^*\bm{W}\bm{D}$ for any unitary matrix $\bm{D}$. Equation $(b)$ is because we can always choose $\bm{D}$ such that $\bm{Dv}=\bm{e_1}$. Equation $(c)$ is because we choose $\bm{D}$ to be a uniform distribution over the unitary matrix that makes $\bm{Dv}=\bm{e_1}$. Equation $(d)$ is because Lemma~\ref{lemma: uWu=0}. Then we finish the proof.
\end{proof}

\begin{remark}
    A similar result for $\bm{u}^*\left(\frac{1}{\sqrt{d}}{\bm{W}}\right)^k \bm{v}$ was shown in \cite{tao2014outliers} (Lemma 2.3). They provide the same convergence in the case where vectors $\bm{u}, \bm{v}$ are deterministic. We cannot directly use it because it assumes that the left and right vectors are independent.
\end{remark}

\subsection{Proof Roadmap of Theorem~\ref{thm: esd of deformed iid}}
\label{sec: thm: ESD}
We are going to show that the ESD of ${\bm{M}}$ converges to the circular law, and as a corollary we obtain $|\lambda_{r+1}({\bm{M}})|\ge 1+o(1)$. The proof will first consider the special case where the entries of ${\bm{W}}$ in Theorem~\ref{thm: outliers} are all Gaussian variables and then reduce the general case to the Gaussian case.

The most important property of the Gaussian i.i.d.\ matrix is the rotational invariance. First, we give the following lemmas, which employ the rotational invariance. We will repeatedly use these two lemmas in this paper.
\begin{lemma}\label{lemma: rotation invariance}
    (Rotational invariance). A measurable function $F\colon\mathbb{R}^{d\times d}\times \operatorname{Stief}(d, r)\to \mathbb{R}$ (or function $F\colon\mathbb{C}^{d\times d}\times \operatorname{Stief}_{\mathbb{C}}(d, r)\to \mathbb{C}$) is rotational invariant if $F\left(\bm{A}, \bm{V}\right)=F\left(\bm{D}\bm{A}\bm{D}^*, \bm{D}\bm{V}\right)$ for any orthonormal $\bm{D}\in\mathbb{R}^{d\times d}$ (or unitary $\bm{D}\in\mathbb{C}^{d\times d}$) and $\bm{V}\in\operatorname{Stief}(d, r)$ (or $\bm{V}\in\operatorname{Stief}_{\mathbb{C}}(d, r)$). Let $\bm{W}$ be a (complex) Gaussian i.i.d.\ matrix, $\bm{U}$ is a random matrix over $\operatorname{Stief}(d, r)$ (or $\operatorname{Stief}_{\mathbb{C}}(d, r)$) and $F$ be a rotational invariant function, then $F\left({\bm{W}}, {\bm{U}}\right)$ and ${\bm{U}}$ are independent.
\end{lemma}
\begin{proof}[Proof]
We will only prove the real case, the complex case is similar. For the real case, we only need to show that $\mathbb{P}\left(F\left({\bm{W}}, \bm{U}\right)\in\mathfrak{B}\right)=\mathbb{P}\left(F\left({\bm{W}}, V\right)\in\mathfrak{B}\right)$ holds for any Borel set $\mathfrak{B}$ and $V\in\operatorname{Stief}(d, r)$. To show this, we have
\begin{equation*}
\begin{aligned}
    \mathbb{P}\left(F\left({\bm{W}}, {\bm{U}}\right)\in\mathfrak{B}\right)&=\mathbb{E}_{{\bm{U}}}\mathbb{P}\left(F\left({\bm{W}},{\bm{U}}\right)\in\mathfrak{B}\mid {\bm{U}}=U\right)\\
    &\stackrel{(a)}{=}\mathbb{E}_{{\bm{U}}}\mathbb{P}\left(F\left(\bm{D}{\bm{W}}\bm{D}^*,\bm{D}{\bm{U}}\right)\in\mathfrak{B}\mid {\bm{U}}=U\right)\\
    &\stackrel{(b)}{=}\mathbb{E}_{{\bm{U}}}\mathbb{P}\left(F\left({\bm{W}},{\bm{DU}}\right)\in\mathfrak{B}\mid {\bm{U}}= U\right)\\
    &=\mathbb{E}_{{\bm{U}}}\mathbb{P}\left(F\left({\bm{W}},{\bm{U}}\right)\in\mathfrak{B}\mid {\bm{U}}=\bm{D}^* U\right)\\
    &=\mathbb{P}\left(F\left({\bm{W}},V\right)\in\mathfrak{B}\right).
\end{aligned}
\end{equation*}
Equality $(a)$ is because $F$ is rotational invariant for orthonormal (or unitary) $\bm{D}$. Equality $(b)$ is because $\bm{D}{\bm{W}}\bm{D}^*$ has the same distribution as ${\bm{W}}$. To obtain the last equality, we take $\bm{D}^* U=V$ for each realization $U$.
\end{proof}

There are some common examples of rotational invariant maps such as $F_1(\bm{A}, \bm{V}) :=\lambda_i(\bm{A}+\bm{V}\bm{\Lambda}\bm{V}^*)$ for $i\in[d]$ and $F_2(\bm{A}, {\bm{v}}) :={\bm{v}}^*\bm{A}^k{\bm{v}}$. For rotational invariant functions, the above lemma helps reduce the random perturbation to a deterministic perturbation when we consider the probability. Then we can use the existing results for the deterministic case. We next to prove Theorem~\ref{thm: esd of deformed iid}, we first consider the Gaussian case.

\begin{lemma}\label{lemma: limit distribution, gaussian case}
    Let ${\bm{M}}=\frac{1}{\sqrt{d}}{\bm{W}}+ {\bm{U}}\bm{\Lambda}{\bm{U}}^*$ be the random matrix in Theorem~\ref{thm: esd of deformed iid}. If $W_{11}\sim\mathcal{CN}(0,1)$, then the ESD ${\bm{U}}_{\bm{M}}$ of ${\bm{M}}$ converges to the circular law ${\bm{U}}$ (i.e., uniform distribution over the unit disk in $\mathbb{C}$) almost surely.
\end{lemma}
\begin{proof}[Proof]
    Note that $F(\bm{A}, \bm{V})=\lambda_{i}(\bm{A}+\bm{V}\bm{\Lambda}\bm{V}^*)$ is a rotational invariant function as defined in Lemma~\ref{lemma: rotation invariance} for $i\in [d]$. Then we can consider ${\bm{U}}$ to be a deterministic matrix in $\operatorname{Stief}_\mathbb{C}(d,r)$. Obviously, we have $\left\|\bm{V}\bm{\Lambda}\bm{V}^*\right\|_F=O(1)$. Then we know that the ESD of ${\bm{M}}$ converges to the circular law by Theorem~\ref{thm: perbutation circular law}.
\end{proof}

We have now finished the Gaussian case where the entries of $\bm{W}$ are Gaussian variables, we then consider the general case. Let ${\bm{W}}, {\bm{U}}$ be the random matrices constructed in Theorem~\ref{thm: esd of deformed iid}. Further consider two independent $d\times d$ matrices ${\bm{W}}_g$ whose entries are standard complex Gaussian variables. Let
\begin{equation*}
    {\bm{M}}_g :=\frac{1}{\sqrt{d}}{\bm{W}_g}+ {\bm{U}}\bm{\Lambda}{\bm{U}}^*, \quad {\bm{M}} :=\frac{1}{\sqrt{d}}{\bm{W}}+ {\bm{U}}\bm{\Lambda}{\bm{U}}^*,
\end{equation*}
By Theorem~\ref{thm: universality from a random base matrix, Theorem 1.17 in [TV07])} and Lemma~\ref{lemma: bound the perturbation}, we know that the difference of the empirical spectral distributions, ${\mu}_{\bm{M}_g}-{\mu}_{\bm{M}}$, converges to $0$ almost surely. By Lemma~\ref{lemma: limit distribution, gaussian case}, we know that $\mu_{\bm{M}_g}$ converges to the circular law almost surely. Thus, $\mu_{\bm{M}}$ converges to the circular law almost surely. As a direct corollary of the limit ESD, we can derive $|\lambda_{r+1}({\bm{M}})|\ge 1+o(1)$ with $o(1)\to0$ almost surely.

\begin{remark}
    Note that once the entries $W_{ij}$ are Gaussian variables, Theorem~\ref{thm: outliers}, Theorem~\ref{thm: outliers, uniform} and Theorem~\ref{thm: esd of deformed iid} can be directly derived from Lemma~$\ref{lemma: rotation invariance}$ and Theorem 1.5 and 1.7 in \cite{tao2014outliers}. As we will see in Section~\ref{section: application}, the Gaussian case is sufficient for our main application (eigenvector approximation), but we believe that the general case would be of independent interest.
\end{remark}

\subsection{Proof Roadmap of Theorem~\ref{thm: top eigenvector, non-Gaussian case}}
\label{sec: roadmap of top eigenvector}
The roadmap begins with the following proposition, which is used to deal with the products of random matrices.

\begin{prop}\label{prop: uWWu=0}
    Let $\bm{W}$ be an i.i.d.\ matrix, ${\bm{u}}$ be an independent random vector over $\mathcal{S}_\mathbb{C}^{d-1}$ and fixed $k_1, k_2\in\mathbb{N}$. Then we have 
    \begin{equation}\label{eq: collection of W^k}
       d^{-\frac{k_1+k_2}{2}}{\bm{u}}^*\left(\bm{W}^{k_1}\right)^*\bm{W}^{k_2}{\bm{u}} = \mathbf{1}_{\{k_1=k_2 \}}+o(1).
    \end{equation}
\end{prop}
Based on the above proposition, our reduction and alignment begin with the observation:
\begin{equation}\label{eq: inner product}
\begin{aligned}
&{\bm{M}}{\bm{v}}_1({\bm{M}})=\lambda_1{\bm{v}}_1({\bm{M}})
\quad \\ 
\Rightarrow \quad&\lambda{\bm{u}}\cdot({\bm{u}}^*{\bm{v}}_1({\bm{M}}))=\left(\lambda_1\bm{I}-\frac{{\bm{W}}}{\sqrt{d}}\right){\bm{v}}_1({\bm{M}})\stackrel{(a)}{=}\left(\lambda\bm{I}-\frac{1}{\sqrt{d}}{\bm{W}}\right){\bm{v}}_1({\bm{M}})+o(1).\\
\Rightarrow \quad& |{\bm{u}}^*{\bm{v}}_1({\bm{M}})|=\frac{1}{\left\|\lambda\left(\lambda\bm{I}-\frac{1}{\sqrt{d}}{\bm{W}}\right)^{-1} {\bm{u}}  \right\|_2} +o(1).
\end{aligned}
\end{equation}
Equality $(a)$ is because Theorem~\ref{thm: outliers}. To get the limit, we have
\begin{equation}\label{eq: inv(I-W)u}
\begin{aligned}
&\left\|\lambda\left(\lambda\bm{I}-\frac{1}{\sqrt{d}}{\bm{W}}\right)^{-1} {\bm{u}}  \right\|_2\\
&=\left\|\left(\sum_{s=0}^{t-1}(\sqrt{d}\lambda)^s{\bm{W}}^s\right){\bm{u}} + \frac{{\bm{W}}^t}{d^{-t/2}\lambda^{t-1}}\left(\lambda\bm{I}-\frac{1}{\sqrt{d}}{\bm{W}}\right)^{-1} {\bm{u}}  \right\|_2 \\
& \le \left\|\left(\sum_{s=0}^{t-1}(\sqrt{d}\lambda)^s{\bm{W}}^s\right){\bm{u}}\right\|_2+\left\| \frac{{\bm{W}}^t}{d^{-t/2}\lambda^{t-1}}\left(\lambda\bm{I}-\frac{1}{\sqrt{d}}{\bm{W}}\right)^{-1} {\bm{u}}  \right\|_2\\
& \stackrel{(a)}{\le} \sqrt{{\bm{u}}^*\left(\sum_{s=0}^{t-1}(\sqrt{d}\lambda)^s{\bm{W}}^s\right)^*\left(\sum_{s=0}^{t-1}(\sqrt{d}\lambda)^s{\bm{W}}^s\right){\bm{u}} } +\frac{(t+1)}{\operatorname{gap}^2\lambda^t}+o(1)\\
& = \sqrt{\sum_{i=0}^{t-1}\frac{{\bm{u}}^*\left({\bm{W}}^i\right)^*{\bm{W}}^i{\bm{u}}}{d^i\lambda^{2i}}+\sum_{i,j=0,i\neq j}^{t-1}\frac{{\bm{u}}^*\left({\bm{W}}^i\right)^*{\bm{W}}^j{\bm{u}}}{d^{(i+j)/2}\lambda^{i+j}} }+\frac{(t+1)}{\operatorname{gap}^2\lambda^t}+o(1)\\
& \stackrel{(b)}{=} \sqrt{\sum_{i=0}^{t-1}\frac{1}{\lambda^{2i}}+o(1) }+\frac{(t+1)}{\operatorname{gap}^2\lambda^t}+o(1)\\
&\le \frac{\lambda}{\sqrt{\lambda^2-1}}+\frac{(t+1)}{\operatorname{gap}^2\lambda^t}+o(1),
\end{aligned}
\end{equation}
where $t\in\mathbb{N}^+$ is an integer to be chosen. To get Inequality $(a)$, we apply Lemma~\ref{lemma: inv(lambdaI-W)} and  Lemma~\ref{lemma: theorem 5.17}. To obtain Equality $(b)$, we apply Proposition~\ref{prop: uWWu=0}. Taking sufficiently large $t$ in the above equation, we obtain $\left\|\lambda\left(\lambda\bm{I}-\frac{1}{\sqrt{d}}{\bm{W}}\right)^{-1} {\bm{u}} \right\|_2 \le \frac{\lambda}{\sqrt{\lambda^2-1}}+o(1)$, giving rise to the lower bound of the inner product $|{\bm{u}}^*{\bm{v}}_1({\bm{M}})|$. We note that the upper bound can be proved in the same way. Then we finish the proof.

\begin{remark}\label{remark: two bounds of gap}
    Theorem~\ref{thm: top eigenvector, non-Gaussian case} provides an $\mathcal{O}\left(\sqrt{\operatorname{gap}}\right)$ lower bound for the inner product $\left|\langle {\bm{u}},{\bm{v}_1(\bm{M})}\rangle\right|$. We can get another weaker bound by simply handling the denominator in Equation~(\ref{eq: inner product}) with Lemma~\ref{lemma: inv(lambdaI-W)}. That is,  
    \begin{equation*}
        \left\|\lambda\left(\lambda\bm{I}-\left(\frac{1}{\sqrt{d}}{\bm{W}}\right)\right)^{-1} {\bm{u}}  \right\|_2\le \left\|\lambda\left(\lambda\bm{I}-\left(\frac{1}{\sqrt{d}}{\bm{W}}\right)\right)^{-1}\right\|_2\le \frac{1}{\operatorname{gap}^2}+o(1).
    \end{equation*}

    Then we get an $\mathcal{O}\left({\operatorname{gap}}^2\right)$ lower bound for the inner product $\left|\langle {\bm{u}},{\bm{v}_1(\bm{M})}\rangle\right|$. The lower bound is sufficient for our application in the eigenvector approximation problem. We seek an exact alignment because we believe that it is of independent interest. We refer to Remark~\ref{remark: why better bound} for more details.
\end{remark}

\subsubsection{Alignment for Case \texorpdfstring{$r>1$}{r>1}}\label{subsubsec: r>2}

Although we only provide the theorem for the case $r=1$, we note that our approach can be extended to the case $r>1$. However, the alignment may not hold in another sense. As a result, the reduction for the case $r>1$ is not valid, while it holds for the symmetric matrices \cite{DBLP:journals/corr/abs-1804-01221}. Here we provide the intuition and roadmap for the case $r>1$ and omit the corresponding proof.

We take $r=2$ as an example, and the cases $r>2$ are similar. We denote $\tilde{\lambda}_1, \tilde{\lambda}_2$ as the top 2 eigenvalues of ${\bm{M}}$. We start from the observation:
\begin{equation*}
\begin{aligned}
&\tilde\lambda_1{\bm{v}}_1({\bm{M}})={\bm{M}}{\bm{v}}_1({\bm{M}})= \left(\frac{1}{\sqrt{d}}{\bm{W}}+{\bm{U}}\Lambda{\bm{U}}^*\right){\bm{v}}_1({\bm{M}})\\
\Rightarrow\quad & \left(\tilde\lambda_1\bm{I}-\frac{{\bm{W}}}{\sqrt{d}}\right){\bm{v}}_1({\bm{M}})=\lambda_1{\bm{u}}_1\cdot({\bm{u}}_1^*{\bm{v}}_1({\bm{M}}))+\lambda_2{\bm{u}}_2\cdot({\bm{u}}_2^*{\bm{v}}_1({\bm{M}}))\\
\stackrel{(a)}{\Rightarrow}\quad& {\bm{v}}_1({\bm{M}})=\lambda_1\left(\lambda_1\bm{I}-\frac{{\bm{W}}}{\sqrt{d}}\right)^{-1}{\bm{u}}_1\cdot({\bm{u}}_1^*{\bm{v}}_1({\bm{M}})) \\
&+\lambda_2\left(\lambda_1\bm{I}-\frac{{\bm{W}}}{\sqrt{d}}\right)^{-1}{\bm{u}}_2\cdot({\bm{u}}_2^*{\bm{v}}_1({\bm{M}}))+o(1).
\end{aligned}
\end{equation*}
Equality $(a)$ is because $\tilde\lambda_1=\lambda_1+o(1)$. Using similar techniques as in Proposition~\ref{prop: uWWu=0} and Theorem~\ref{thm: top eigenvector, non-Gaussian case}, we  obtain that
\[\left\|\lambda_1\left(\lambda_1\bm{I}-\frac{{\bm{W}}}{\sqrt{d}}\right)^{-1}{\bm{u}}_1\right\|=\frac{\lambda_1}{\sqrt{\lambda_1^2-1}}, \ \left\|\lambda_2\left(\lambda_1\bm{I}-\frac{{\bm{W}}}{\sqrt{d}}\right)^{-1}{\bm{u}}_2\right\|=\frac{\lambda_2}{\sqrt{\lambda_1^2-1}}. 
\]
Meanwhile, using similar approach as in the proof of Proposition~\ref{prop: uWWu=0}, we can obtain that $d^{-(k_1+k_2)/2}{\bm{u}}_1^*\left({\bm{W}}^{k_1}\right)^*{\bm{W}}^{k_2}{\bm{u}}_2=o(1)$ for $k_1, k_2\in \mathbb{N}$, and hence
\[\left|\left\langle\lambda_1\left(\lambda_1\bm{I}-\frac{{\bm{W}}}{\sqrt{d}}\right)^{-1}{\bm{u}}_1, \lambda_2\left(\lambda_1\bm{I}-\frac{{\bm{W}}}{\sqrt{d}}\right)^{-1}{\bm{u}}_2\right\rangle\right|=o(1).\]
Thus we obtain
\[\frac{\lambda_1^2}{\lambda_1^2-1}\cdot|{\bm{u}}_1^*{\bm{v}}_1({\bm{M}})|^2+\frac{\lambda_2^2}{\lambda_1^2-1}\cdot|{\bm{u}}_2^*{\bm{v}}_1({\bm{M}})|^2=1+o(1).\]
Similarly, for general $r$, we can obtain $\sum_{i=1}^r {\lambda_i^2}\cdot|{\bm{u}}_i^*{\bm{v}}_j({\bm{M}})|^2=\lambda_j^2-1+o(1)$ for $j\in [r]$. The result shows that their is an overlap of the space spanned by $\bm{U}$ and the space spanned by eigenvectors of $\bm{M}$, which is also an alignment.

\section{Application: Eigenvector Approximation Lower Bounds}\label{section: application}
Now we consider an application of the random matrix we construct in Section~\ref{section: deforemed i.i.d. random matrix}. It is a common target to approximate the eigenvector of a matrix ${\bm{M}}$ in machine learning and data science. Several methods, such as the power method, can be used to solve the problem iteratively, finding a vector $\hat{{\bm{v}}}$ with small $\left\|\hat{{\bm{v}}}-{\bm{v}}_1({\bm{M}})\right\|_2$. In this section, we employ the ESD result in Section~\ref{section: deforemed i.i.d. random matrix}, especially with real Gaussian entries and $r=1$, to provide an optimal query complexity lower bound for approximating the top eigenvector of an asymmetric matrix.

We will introduce the setting and the main theorem in Section~\ref{sec: application, setting and main theorem}. The proof roadmap of the main theorem (Theorem~\ref{thm: main theorem, to be prove}) can be divided into the reduction part and the estimation part as \cite{DBLP:journals/corr/abs-1804-01221} does. The reduction part will be stated in Section~\ref{sec: reduction}. Then we provide some observations to simplify the estimation part in Section~\ref{subsec: observations} and then consider the estimation part in Section~\ref{subsec: estimation}. Finally, we consider a two-side extension of the algorithm class, as well as a different random matrix in Section~\ref{sec: two-side}.

\subsection{The Setting and the Main Theorem}\label{sec: application, setting and main theorem}

The algorithm class that we are concerned with in this application to solve the eigenvector is the matrix-vector product model. First, we introduce the query model (i.e., the matrix-vector product model). 

\paragraph{Query model.} (Matrix-vector product model). We are concerned with a query model, which is described as follows. 
The objective of a \textit{randomized adaptive query} algorithm $\operatorname{Alg}$ is to return a vector $\hat{\bm{v}}\in\mathbb{C}^d$ such that $\|\hat{\bm{v}}-{\bm{v}}_1({\bm{M}})\|_2^2\leq \epsilon$ for $\epsilon>0$, where ${\bm{M}}\in\mathbb{R}^{d\times d}$ is a diagonalizable matrix. For $\operatorname{Alg}$, the only access to ${\bm{M}}$ is via matrix-vector products, i.e., $\operatorname{Alg}$ receives $\bm{w}^{(t)}={\bm{M}} {\bm{v}}^{(t)}$ in iteration $t$. The next ${\bm{v}}^{(t+1)}$ of the query can be randomized and adaptive, i.e., ${\bm{v}}^{(t+1)}$ is a function of $\left\{\left({\bm{v}}^{(1)}, \bm{w}^{(1)}\right), \ldots,\left({\bm{v}}^{(t)}, \bm{w}^{(t)}\right)\right\}$ as well as the random seed. The complexity ${T} \in\mathbb{N}$ is the length of queries. At the end of $T$ queries, the algorithm returns a unit vector $\hat{\bm{v}} \in\mathbb{C}^d$, which is also a function of the past vectors and the random seed.

There are several examples of this class of algorithms, such as the Lanczos method, the Block Krylov method, and the power method. Meanwhile, several lower bounds for numerical linear algebra problems are established on this model. \cite{DBLP:journals/corr/abs-1804-01221} studied a symmetric query model in which ${\bm{M}}$ is taken as a GOE plus a rank-$r$ random perturbation, and gave the optimal lower bound for approximating the top $r$ eigenspace of the symmetric matrix ${\bm{M}}$. 

In our application, we would address an asymmetric query model, and we will consider approximating only the top eigenvector of an asymmetric matrix ${\bm{M}}$, a GinOE matrix with rank-$1$ perturbation. In particular, we set the distribution of the matrix ${\bm{M}}$ as ${\bm{M}} :=\frac{1}{\sqrt{d}}{\bm{W}}+\lambda{\bm{u}}{\bm{u}}^*$, where ${\bm{W}}$ and ${\bm{u}}$ are independent and the entries of ${\bm{W}}$ are standard Gaussian variable, ${\bm{u}} \stackrel{\text { unif }}{\sim} \mathcal{S}^{d-1}$, and $\lambda>1$ is the level of perturbation. For convenience of notation, we denote $\bm{G}=\frac{1}{\sqrt{d}}{\bm{W}}$ in this section, then $\bm{G}\sim\operatorname{GinOE}(d)$. Note that the matrix we construct here is a special case of Theorem~\ref{thm: outliers} and Theorem~\ref{thm: esd of deformed iid}, then we obtain that $\lambda_1({\bm{M}})$ convergess to $\lambda$ and $|\lambda_2({\bm{M}})|$ converges to 1 in probability. Hence, $\operatorname{gap}({\bm{M}})=\frac{\lambda-1}{\lambda}+o(1)$. 
\begin{corollary}
\label{corollary: outliers, r=1}
Let $\bm{G}\sim\operatorname{GinOE}(d)$ and indepdent ${\bm{u}}\stackrel{\text{unif}}{\sim}\operatorname{Stief}(d, 1)$, then $\lambda_{1}({\bm{M}})=\lambda+o(1), |\lambda_{2}({\bm{M}})|= 1+o(1)$.
\end{corollary}

\begin{remark}
    Here we explain some intuition about our choice of the random matrix,  the main difference from the symmetric case. In the symmetric case, \cite{DBLP:journals/corr/abs-1804-01221} considered the deformed Gaussian Orthogonal Ensemble ${\bm{M}}$, where $\operatorname{gap}({\bm{M}})\to \frac{\lambda+\lambda^{-1}-2}{\lambda+\lambda^{-1}}$. So their eigen-gap is $\frac{\lambda+\lambda^{-1}-2}{\lambda+\lambda^{-1}}=\operatorname{gap}$, i.e., $\lambda=\mathcal{O}(\sqrt{\operatorname{gap}})$ for small $\operatorname{gap}$. In the case of $\operatorname{GinOE}$, the eigen-gap is $\operatorname{gap}({\bm{M}})\to \frac{\lambda-1}{\lambda}$. Thus, we choose the eigen-gap to be $\frac{\lambda-1}{\lambda}=\operatorname{gap}$, i.e., $\lambda=\mathcal{O}(\operatorname{gap})$. The different eigen-gap results in different complexity.
\end{remark}

Now we state a precise version of Theorem~\ref{thm: main theorem informally}:
\begin{thm}\label{thm: main theorem, to be prove}Given a fixed $\operatorname{gap}\in(0,1/2]$, let $\lambda=\frac{1}{1-\operatorname{gap}}=1+\mathcal{O}(\operatorname{gap})$. Define ${\bm{M}}=\bm{G}+\lambda {\bm{u}}{\bm{u}}^*$ where ${\bm{u}} \stackrel{\text { unif }}{\sim} \mathcal{S}^{d-1}$ and $\bm{G} \sim \operatorname{GinOE}(d)$. Then for any query algorithm $\operatorname{Alg}$ with output $\hat{\bm{v}}$, we have
\begin{equation}\label{eq: probability of main theorem}
    \mathbb{P}_{\operatorname{Alg},\bm{G}, {\bm{u}}}\left(\left\|\hat{{\bm{v}}} {-} {\bm{v}}_1({\bm{M}}) \right\|_2\le \frac{\sqrt{\operatorname{gap}}}{2}  \right) \le \exp\left\{-\frac{1}{\operatorname{gap}} {-} \lambda^{-4{T}} \frac{d\operatorname{gap}^3}{256}\right\}+o_d(1),
\end{equation}
where the probability $\mathbb{P}_{\mathrm{Alg}}$ is taken with respect to the randomness of the algorithm.
\end{thm}

The proof of Theorem~\ref{thm: main theorem, to be prove} is provided in Section~\ref{section: i-t lower bound}, and the proof roadmap is provided in the remaining parts of this section. Taking ${T}=o\left(\frac{\log{d}}{\operatorname{gap}}\right)$ in the above theorem, we easily reach the lower bound of complexity. That is, it needs at least $\mathcal{O}\left(\frac{\log{d}}{\operatorname{gap}}\right)$ queries to get an approximated eigenvector. 
\begin{remark}
    Taking $\epsilon =\operatorname{gap}$, we get another lower bound $\mathcal{O}\left(\frac{\log{d}}{\epsilon}\right)$ to identify a vector ${\bm{v}}$ with $\|{\bm{v}}-{\bm{v}}_1({\bm{M}})\|^2_2\le\epsilon$. However, as we have mentioned earlier, this lower bound does not have the corresponding upper bound under the criterion $\|{\bm{v}}-{\bm{v}}_1({\bm{M}})\|^2_2\le\epsilon$. But in the symmetric case, the upper bound and the lower bound $\Theta\left(\frac{\log{d}}{\sqrt{\epsilon}}\right)$ are matched.
\end{remark}
We state the formal version of the main theorem below, which could directly follow from Theorem~\ref{thm: main theorem, to be prove}.

\begin{thm}\label{thm: main theorem} (Main theorem for application). Given a fixed $\operatorname{gap}\in(0,1/2]$,  let $\lambda=\frac{1}{1-\operatorname{gap}}=1+\mathcal{O}(\operatorname{gap})$. There exists a random matrix ${\bm{M}}\in\mathbb{R}^{d\times d}$ for which $\operatorname{gap}\left({\bm{M}}\right)=\operatorname{gap}+o(1)$, such that for any query algorithm $\operatorname{Alg}$ if making ${T}\le \frac{\log d}{5\operatorname{gap}}$ queries, with probability at least $1-o_{d}(1)$, $\operatorname{Alg}$ cannot identify a unit vector $\hat{\bm{v}}\in \mathbb{C}^d$ for which $\|\hat{\bm{v}}-{\bm{v}}_1({\bm{M}})\|_2^2\le \frac{\operatorname{gap}}{4}$.
\end{thm}
\begin{proof}[Proof]
Our random matrix is exactly ${\bm{M}}$ in Theorem~\ref{thm: main theorem, to be prove}. The result about eigen-gap, i.e., $\operatorname{gap}\left({\bm{M}}\right)=\operatorname{gap}+o(1)$ comes from Corollary~\ref{corollary: outliers, r=1}. For the probability, taking ${T}\le \frac{\log d}{5\operatorname{gap}}$ in Theorem~\ref{thm: main theorem, to be prove}, we know the right hand side of Equation~(\ref{eq: probability of main theorem}) is $o_{d}(1)$. 
\end{proof}
\begin{remark}\label{remark: excat form of d}
    Note that our main theorem and the precise version differ slightly from those in \cite{DBLP:journals/corr/abs-1804-01221}. As mentioned above, the main distinction lies in the probability and the requirement for the dimension $d$. The specific form of $o_d(1)$ in the above theorem remains unclear because we have applied asymptotic results of the random matrix theory. It is possible to provide a specific form of $o(1)$ with the corresponding non-asymptotic analysis of the random matrix.
\end{remark}

Now we have finished the lower bound, based on the analysis of the random matrix. The distribution $\bm{M}$ seems unnatural because it contains an asymmetric part $\bm{G}$ and a symmetric part ${\bm{u}}{\bm{u}}^*$. However, the unnatural setting is necessary for us when considering the eigenvector approximation rather than the singular vector approximation. For a general i.i.d.\ matrix with perturbation $\frac{1}{\sqrt{d}}\bm{W}+\bm{C}$, its top eigenvalue will concentrate around $\lambda_1(\bm{C})$. Then if $\bm{C}$ is asymmetric, $\|\bm{C}\|_2$ may be much larger than $\lambda_1(\bm{C})$, which is terrible for the estimation process. That is a reason why we consider the symmetric perturbation.

In addition, because we are considering the query model for asymmetric matrices, it is natural for the query model to access products ${\bm{M}}^*{\bm{v}}^{(t)}$ \cite{bakshi2022low}. In this case, the algorithm makes a triple in each round: 
\[\left\{({\bm{v}}^{(t)}, \bm{w}^{(t)}, \bm{z}^{(t)})\colon  \bm{w}^{(t)} = {\bm{M}}{\bm{v}}^{(t)}, \bm{z}^{(t)} = {\bm{M}}^*{\bm{v}}^{(t)}\right\}.\] 
We consider this model in Section~\ref{sec: two-side} and thus extend the algorithm to a larger class. We find that similar intuition, bounds, and results hold for the triple case, with the analysis of another random matrix.

\subsection{Proof Roadmap: Reduction}\label{sec: reduction}
Recall that our random matrix is ${\bm{M}}=\bm{G}+\lambda{\bm{u}}{\bm{u}}^*$ with $\lambda>1$. Our target is to identify a unit vector $\hat{\bm{v}}\in \mathbb{C}^d$ for which $\|\hat{\bm{v}}-{\bm{v}}_1({\bm{M}})\|_2^2\le \frac{\operatorname{gap}}{4}$. The reduction part aims to reduce our target to the problem of identifying a unit vector $\hat{\bm{v}}\in \mathbb{C}^d$ for which $\left|\langle\hat{\bm{v}}^*{\bm{u}}\rangle\right|_2^2\ge \frac{\operatorname{gap}}{4}$, which is also known as the overlap problem \cite{DBLP:journals/corr/abs-1804-01221}. We formulate the reduction in the following corollary.

\begin{corollary}\label{corollary: reduction}
    Let ${\bm{M}}=\bm{G}+\lambda{\bm{u}}{\bm{u}}^*$ be the random matrices defined in Theorem~\ref{thm: main theorem, to be prove} and let $v$ be a unit vector in $\mathbb{C}^d$. Then if $\|v-{\bm{v}}_1({\bm{M}})\|_2\le\frac{\sqrt{\operatorname{gap}}}{2}$, we have $|{\bm{u}}^*v|\ge \frac{\sqrt{\operatorname{gap}}}{2}+o(1)$ where $o(1)$ is independent of ${\bm{u}}$. 
\end{corollary}
\begin{proof}[Proof]
    By Theorem~\ref{thm: top eigenvector, non-Gaussian case}, we obtain $|{\bm{u}}^*{\bm{v}}_1({\bm{M}})|\ge \sqrt{\operatorname{gap}}+o(1)$. In addition, note that $F(\bm{A}, {\bm{v}})=\left|{\bm{v}}^*{\bm{v}}_1(\bm{A}+\lambda{\bm{v}}{\bm{v}}^*)\right|$ is rotational invariant by Lemma~\ref{lemma: rotation invariance}. This is because
    \begin{equation*}
    \begin{aligned}
        F(\bm{D}\bm{A}\bm{D}^*, \bm{D}{\bm{v}})&=\left|{\bm{v}}^*\bm{D}^*{\bm{v}}_1(\bm{D}\bm{A}\bm{D}^*+\lambda\bm{D}{\bm{v}}{\bm{v}}^*\bm{D}^*)\right|\\
        &=\left|{\bm{v}}^*\bm{D}^*\bm{D}{\bm{v}}_1(\bm{A}+\lambda{\bm{v}}{\bm{v}}^*)\right|=F(\bm{A}, {\bm{v}}).
    \end{aligned}
    \end{equation*}
    Thus the $o(1)$ is independent of ${\bm{u}}$. Thus we obtain
    \[|{\bm{u}}^*v|\ge |{\bm{v}}_1({\bm{M}})^*v| - \|v-{\bm{v}}_1({\bm{M}})\|_2 \ge \frac{\sqrt{\operatorname{gap}}}{2}+o(1),\]
    where $o(1)$ is independent of ${\bm{u}}$.
\end{proof}

\begin{remark}\label{remark: why better bound}
    We have mentioned in Remark~\ref{remark: two bounds of gap} that the inner product has a weaker bound $\mathcal{O}\left(\frac{1}{\operatorname{gap}^2}\right)$. As we will see in Section~\ref{subsec: estimation} (estimation part), both lower bounds for the inner product can lead to the desired lower bound $\mathcal{O}\left(\frac{\log d}{\operatorname{gap}}\right)$. The difference comes out when we consider the approximation accuracy. Setting $\epsilon =\operatorname{gap}$ and $\epsilon =\operatorname{gap}^{\frac{1}{4}}$ in two lower bounds, we obtain new lower bounds $\mathcal{O}\left(\frac{\log d}{\epsilon}\right)$ and $\mathcal{O}\left(\frac{\log d}{\epsilon^4}\right)$ for the target $\left\|{\bm{v}}_1({\bm{M}})-\hat{{\bm{v}}}\right\|_2^2\le \epsilon$. 
    
    In the symmetric, the optimal lower bound for this criterion is $\mathcal{O}\left(\frac{\log d}{\sqrt{\epsilon}}\right)$. Thus, we believe that $\mathcal{O}\left(\frac{\log d}{{\epsilon}}\right)$ is the optimal lower bound for the asymmetric case when considering the accuracy. However, it is not clear whether $\left\|{\bm{v}}_1({\bm{M}})-\hat{{\bm{v}}}\right\|_2^2\le \epsilon$ is a suitable criterion.
\end{remark}

\subsection{Proof Roadmap: Observations}\label{subsec: observations}
The reduction part helps reduce the original problem to estimate ${\bm{u}}$ from the noised matrix $\bm{G}+\lambda{\bm{u}}{\bm{u}}^*$. In this section, we will provide some observations to simplify the problem of estimation. We modify the framework (especially the observations) in \cite{DBLP:journals/corr/abs-1804-01221} to adapt the complex case and the different choices of the random matrix. Similarly, we first show that a good estimator $\hat{\bm{v}}$ for the eigenvector implies the existence of a \textit{deterministic} algorithm, which receives a sequence of orthonormal queries: 
\[\mathfrak{Re}({\bm{v}}^{(1)}), \mathfrak{Im}({\bm{v}}^{(1)}), \mathfrak{Re}({\bm{v}}^{(2)}), \mathfrak{Im}({\bm{v}}^{(2)}), \ldots, \mathfrak{Re}({\bm{v}}^{(k)}), \mathfrak{Im}({\bm{v}}^{(k)}),
\] and outputs a unit vector $\hat{\bm{v}}$ with lower bounded $|\hat{\bm{v}}^*{\bm{u}}|$. The first observation is that we assume without loss of generality that our queries are orthonormal:

\begin{obs}\label{obs: orthonormal}
Consider the real and imaginary parts of queries,
\[\bm{V}_k:=\left[\mathfrak{Re}({\bm{v}}^{(1)}), \mathfrak{Im}({\bm{v}}^{(1)}),\mathfrak{Re}({\bm{v}}^{(2)}), \mathfrak{Im}({\bm{v}}^{(2)}), \ldots , \mathfrak{Re}({\bm{v}}^{(k)}), \mathfrak{Im}({\bm{v}}^{(k)})\right].\]
Assume the queries to be orthonormal, i.e., $\bm{V}_k\in \operatorname{Stief}(d, 2k)$. 
\end{obs}
The assumption that the real parts and the imaginary parts of the queries are orthonormal is valid. Because once we receive a complex vector ${\bm{v}}$, we have actually received two real vectors $\mathfrak{Re}({\bm{v}})$ and $\mathfrak{Im}({\bm{v}})$. Then we can apply the Gram-Schmidt procedure to the $2k$-queries \[\left[\mathfrak{Re}({\bm{v}}^{(1)}), \mathfrak{Im}({\bm{v}}^{(1)}), \ldots, \mathfrak{Re}({\bm{v}}^{(k)}), \mathfrak{Im}({\bm{v}}^{(k)})\right],\] and we can always reconstruct $k$-queries ${\bm{v}}^{(1)}, \ldots, {\bm{v}}^{(k)}$ from an associated orthonormal sequence obtained via the Gram-Schmidt procedure.

There are several other alternatives to the observation, and each one needs to be dealt with carefully. One way is to regard the real part and the imaginary part of ${\bm{v}}$ as two vectors and adapt the method in \cite{DBLP:journals/corr/abs-1804-01221} with $r=2$. Additionally, we can assume the ${\bm{v}}^{(i)}$ to be real vectors for convenience. This is because the real part and the imaginary part can be inputted separately. We can also consider the queries to be orthonormal complex vectors as another alternative. 

\begin{remark}
    We do not assume that $\operatorname{Alg}$ receives $\bm{w}^{(k)}=(\bm{I}-\bm{V}_{k-1}\bm{V}_{k-1}^*){\bm{M}}{\bm{v}}^{(k)}$ as \cite{DBLP:journals/corr/abs-1804-01221} does. There are two main reasons. On the one hand, we cannot reconstruct the original $\bm{w}^{(k)} = {\bm{M}}{\bm{v}}^{(k)}$ from $(\bm{I}-\bm{V}_{k-1}\bm{V}_{k-1}^*){\bm{M}}{\bm{v}}^{(k)}$ because we do not know ${\bm{M}}^*\bm{V}_{k-1}$ (the oracle only receives ${\bm{M}}\bm{V}_{k-1}$).

    On the other hand, in our setting, $\bm{G}{\bm{v}}_1$ and $\bm{G}{\bm{v}}_2$ are independent of orthogonal ${\bm{v}}_1$ and ${\bm{v}}_2$, respectively. This is a necessary property in the estimation part, and we refer to the proof of Lemma~\ref{lemma: conditional likelihoods} for details. This property does not hold for the symmetric case \cite{DBLP:journals/corr/abs-1804-01221}, so they must assume that the algorithm receives $\bm{w}^{(k)}=(\bm{I}-\bm{V}_{k-1}\bm{V}_{k-1}^*){\bm{M}}{\bm{v}}^{(k)}$.
\end{remark}

The next observation shows that the new orthonormal query suffices to the upper bound $\lambda_{1}\left(\hat{{\bm{v}}}^* uu^* \hat{{\bm{v}}}\right)$ with
$\lambda_{1}\left(\bm{V}_{T+1}^{*} uu^{*} \bm{V}_{T+1}\right)$ for an arbitrary realization $u\in\mathbb{C}^d$ of the perturbation vector ${\bm{u}}$.

\begin{obs}\label{obs: eigenvalue extend}We assume without loss of generality that $\operatorname{Alg}$ makes an extra vector ${\bm{v}}^{(T+1)}$ after having outputted $\hat{
{\bm{v}}}$, and that $\lambda_{1}\left(\hat{{\bm{v}}}^* uu^* \hat{{\bm{v}}}\right) \leq \lambda_{1}\left(\bm{V}_{T+1}^{*} uu^{*} \bm{V}_{T+1}\right)$ for an arbitrary $u\in\mathbb{C}^d$.
\end{obs}
\begin{proof}[Proof]
    It is valid that one can always choose the extra ${\bm{v}}^{(T+1)}\in\mathbb{C}^{d}$ and ensures that
    \begin{equation*}
        \mathfrak{Re}(\hat{\bm{v}}), \mathfrak{Im}(\hat{\bm{v}})\in \operatorname{span}\left(\mathfrak{Re}({\bm{v}}^{(1)}), \mathfrak{Im}({\bm{v}}^{(1)}),\ldots, \mathfrak{Re}({\bm{v}}^{(T+1)}), \mathfrak{Im}({\bm{v}}^{(T+1)})\right).
    \end{equation*}
    In this case, we have $\mathfrak{Re}(\hat{\bm{v}})=\bm{V}_{T+1}\bm{a}, \mathfrak{Im}(\hat{\bm{v}})=\bm{V}_{T+1}\bm{b}$ for some $\bm{a}, \bm{b}\in\mathbb{C}^{2T+2}$ with $\left\|\bm{a}\right\|_2=\|\mathfrak{Re}(\hat{\bm{v}})\|_2$ and $\|\bm{b}\|_2=\|\mathfrak{Im}(\hat{\bm{v}})\|_2$. Thus,
    \begin{equation*}
    \begin{aligned}
        \bm{V}_{T{+}1}\bm{V}_{T{+}1}^{*}&=\bm{V}_{T{+}1}\left(\frac{\bm{a}\bm{a}^*}{\|\bm{a}\|_2^2} {+} \bm{I} {-} \frac{\bm{a}\bm{a}^*}{\|\bm{a}\|_2^2}\right)\bm{V}_{T+1}^{*}\\
        &=\frac{\mathfrak{Re}(\hat{\bm{v}})\mathfrak{Re}(\hat{\bm{v}})^*}{\|\bm{a}\|_2^2}+\bm{V}_{T+1}\left(\bm{I}-\frac{\bm{a}\bm{a}^*}{\|\bm{a}\|_2^2}\right)\bm{V}_{T+1}^{*}\succeq \frac{\mathfrak{Re}(\hat{\bm{v}})\mathfrak{Re}(\hat{\bm{v}})^*}{\|\bm{a}\|_2^2}.
    \end{aligned}
    \end{equation*}
     Then we obtain
\begin{equation*}
    \lambda_{1}\left(\bm{V}_{T+1}^{*} uu^{*} \bm{V}_{T+1}\right)=u^*\bm{V}_{T+1}\bm{V}_{T+1}^{*}u\ge \frac{u^*\mathfrak{Re}(\hat{\bm{v}})\mathfrak{Re}(\hat{\bm{v}})^*u}{\|\bm{a}\|_2^2}.
\end{equation*}
For the same reason we get $\lambda_{1}\left(\bm{V}_{T+1}^{*} uu^{*} \bm{V}_{T+1}\right)\ge {u^*\mathfrak{Im}(\hat{\bm{v}})\mathfrak{Im}(\hat{\bm{v}})^*u/\|\bm{b}\|_2^2}.$ Thus we have
\begin{equation*}
\begin{aligned}
    \lambda_{1}\left(\bm{V}_{T+1}^{*} uu^{*} \bm{V}_{T+1}\right)&= \left(\left\| \mathfrak{Re}(\hat{\bm{v}})\right\|_2^2+\left\| \mathfrak{Im}(\hat{\bm{v}})\right\|_2^2\right)\lambda_{1}\left(\bm{V}_{T+1}^{*} uu^{*} \bm{V}_{T+1}\right)\\
    &\ge {{u^*\mathfrak{Re}(\hat{\bm{v}})\mathfrak{Re}(\hat{\bm{v}})^*u}+{u^*\mathfrak{Im}(\hat{\bm{v}})\mathfrak{Im}(\hat{\bm{v}})^*u}}={u^* \hat{{\bm{v}}}\hat{{\bm{v}}}^* u}\\
    &={\lambda_{1}\left(\hat{{\bm{v}}}^* uu^* \hat{{\bm{v}}}\right)}.    
\end{aligned}
\end{equation*}
\end{proof}

Finally, we consider the randomness of the algorithms. The following observation helps remove the randomness:
\begin{obs}\label{obs:deterministic}We  assume that for all $k \in[T+1]$, the query ${\bm{v}}^{(k+1)}$ is deterministic given the previous query-observation pairs $\left({\bm{v}}^{(i)}, \bm{w}^{(i)}\right)_{1 \leq i \leq k}$.
\end{obs}
The observation because if we have shown that \[\mathbb{P}_{\bm{u}, {\bm{M}}}\left[\lambda_{1}\left(\bm{V}_{T+1}^{*} \bm{u}\bm{u}^{*} \bm{V}_{T+1}\right) \geq B \right] \leq b\] for any deterministic algorithm and bounded $B, b$, then we derive that 
\begin{equation*}
\begin{aligned}
&\mathbb{P}_{\bm{u}, {\bm{M}}, \mathrm{Alg}}\left[\lambda_{1}\left(\bm{V}_{T+1}^{*} \bm{u}\bm{u}^{*} \bm{V}_{T+1}\right) \geq B\right]\\
&=\mathbb{E}_{\mathbf{A l g}} \mathbb{P}_{\bm{u}, {\bm{M}}}\left[\lambda_{1}\left(\bm{V}_{T+1}^{*} \bm{u}\bm{u}^{*} \bm{V}_{T+1}\right) \geq B \mid\operatorname{Alg}\right] \leq b.
\end{aligned}
\end{equation*}
Thus, we can only consider the deterministic case.

\subsection{Proof Roadmap: Estimation}\label{subsec: estimation}

As discussed earlier, we have reduced the original problem to an overlapping problem by Corollary~\ref{corollary: reduction}. In addition, the above observations help us to obtain sequentially selected orthogonal measurements ${\bm{v}}^{(1)}, \ldots, {\bm{v}}^{(T+1)}$ with a lower bounded overlap $\|\bm{V}_{{T}+1}^* \bm{u}\|_2$. It remains to give a lower bound for the statistical problem of estimating the perturbation $\bm{u}$. More specifically, let $\hat{\bm{v}}$ be the output of the query model, we have
\begin{equation*}
\begin{aligned}
   &\text{query lower bound of approximating } \bm{u}^*\bm{V}_k\bm{V}_k^*\bm{u} \\
   \stackrel{Obs~\ref{obs: eigenvalue extend}}{\Rightarrow} &\text{lower bound of } \bm{u}^*\hat{\bm{v}}\hat{\bm{v}}^*\bm{u}\\
   \stackrel{Coro~\ref{corollary: reduction}}{\Rightarrow}& \text{lower bound of }\|\hat{\bm{v}}-{\bm{v}}_1({\bm{M}})\|_2.
\end{aligned}
\end{equation*}

Then our estimation part aims to lower bound $u^*\bm{V}_k\bm{V}_k^*u$. The technique of this section is drawn from \cite{DBLP:journals/corr/abs-1804-01221}, with a slight difference due to the complex region and the different choice of the random matrix. Our main result for the overlapping problem is as follows. 
\begin{thm}\label{thm: upper bound of information}Let ${\bm{M}}=\bm{G}+\lambda {\bm{u}}{\bm{u}}^{*}$, where $\bm{G} \sim \mathrm{GinOE}(d)$ and is independent of ${\bm{u}} \stackrel{\text { unif }}{\sim} \mathcal{S}^{d-1}$. Then for all $\delta \in(0, 1/e)$, we have that
\[\mathbb{E}_{{\bm{u}}} \mathbb{P}_{u}\left[\exists k \geq 1: u^{*} \bm{V}_k \bm{V}_k^{*} u \geq  \frac{64 \lambda^{4 k-4}\left(\log\delta^{-1}+\operatorname{gap}^{-1}\right)}{d\operatorname{gap}^2}\right] \leq \delta.\]
\end{thm}

This theorem shows that the information about the overlap $\|\bm{V}_{{T}+1}^*{\bm{u}}\|_2$ will grow at most in a linear rate in the parameter $\lambda$. Then the result for the original problem (see Theorem~\ref{thm: main theorem, to be prove}) can be derived by combining Theorem~\ref{thm: upper bound of information}, Observation~\ref{obs: eigenvalue extend}, Corollary~\ref{corollary: reduction} and Corollary~\ref{corollary: outliers, r=1}. The proof of estimation is similar to \cite{DBLP:journals/corr/abs-1804-01221}, we leave the roadmap and detailed proof in Section~\ref{section: i-t lower bound}.

\subsection{Two-side Query Model}\label{sec: two-side}

As discussed above, when dealing with the asymmetric matrix, it is a more common case that the model accesses both ${\bm{M}}{\bm{v}}$ and ${\bm{M}}^*{\bm{v}}$ for the input vector ${\bm{v}}$ at each iteration. In this section, we give our results about such a two-side model.

We first note that the aforementioned random matrix ${\bm{M}}=\bm{G}+\lambda{\bm{u}}{\bm{u}}^*$ cannot be applied here. That is because we can access $\Tilde{\bm{M}}=\frac{\bm{G}+\bm{G}^*}{2}+\lambda{\bm{u}}{\bm{u}}^*$ via query model. The eigen-gap of $\Tilde{\bm{M}}$ is $O(1)$, and thus the matrix can be solved by the power method easily with $\mathcal{O}(\log d)$ complexity. To emphasize the usage of the two-side model, we consider the matrix to be $\bm{G}+\lambda{\bm{u}}_l{\bm{u}}_r^*$ in this section, where $\bm{G}\sim \text{GinOE}(d), {\bm{u}}_l, {\bm{u}}_r \stackrel{\text { unif }}{\sim} \mathcal{S}^{d-1}, \lambda>1$ and ${\bm{u}}_l, {\bm{u}}_r$ are independent. Because the main results, observations, intuition, and the roadmap are all similar to the one-side case (where the model only receives ${\bm{M}}{\bm{v}}$), we will only briefly introduce our main results and leave the details in Section~\ref{sec: two side situation}. First, the two-side query model is an algorithm $\operatorname{Alg}$ that receives a vector triple 
\begin{equation*}
    \left\{({\bm{v}}^{(t)}, \bm{w}^{(t)}, \bm{z}^{(t)})\colon  {\bm{v}}^{(t)}, \bm{w}^{(t)} = {\bm{M}}{\bm{v}}^{(t)}, \bm{z}^{(t)} = {\bm{M}}^*{\bm{v}}^{(t)}\right\}
\end{equation*}
at iteration $t$. Our main result provides a similar lower bound for the problem of approximating ${\bm{u}}_l$ and ${\bm{u}}_r$. That is, 

\begin{thm}(Main theorem for application in two side case, Theorem~\ref{thm: main theorem, two side case}). Let ${\bm{M}}=\bm{G}+\lambda {\bm{u}}_l{\bm{u}}_r^*$ where $\lambda>1, \bm{G}\sim \operatorname{GinOE}(d), {\bm{u}}_l, {\bm{u}}_r\stackrel{\text{unif}}{\sim}\mathcal{S}^{d-1}$, and $\bm{G}, {\bm{u}}_l, {\bm{u}}_r$ are independent. Then for any two-side query algorithm $\operatorname{Alg}$ if making ${T}\le \frac{\log d}{5(\lambda-1)}$ queries, with probability at least $1-o_{d}(1)$, $\operatorname{Alg}$ cannot identify unit vectors $\hat{\bm{v}}_l, \hat{\bm{v}}_r\in \mathbb{R}^d$ for which $\left|\langle\hat{\bm{v}}_l, {\bm{u}}_l\rangle\right|^2+\left|\langle\hat{\bm{v}}_r, {\bm{u}}_r\rangle\right|^2\le \frac{\lambda-1}{4}$.
\end{thm}

\section{Upper bounds: the Power Method}\label{section: power method}
For completeness, we give an upper bound for the complexity of approximating the top eigenvector of an asymmetric matrix $\bm{A}\in\mathbb{C}^{d\times d}$, together with some illustrative examples. We consider the power method: 
\begin{equation}
\label{standard pm}
\text{Power Method: } {\bm{v}}_{t}=\frac{\bm{A}{\bm{v}}_{t-1}}{\left\|\bm{A}{\bm{v}}_{t-1}\right\|_2},
\end{equation}
where ${\bm{v}}_0\in\mathbb{C}^{d}$ is an initial vector.

\begin{thm}
    For any $\operatorname{gap}\in (0,1)$, we are given $\bm{A}\in\mathbb{C}^{d\times d}$ with $\operatorname{gap}(\bm{A})=\operatorname{gap}$ and distinct eigenvalues, ${\bm{v}}_0$ is uniformly sampled from $\mathbb{C}^d$, and ${\bm{v}}_k$ is the output of the Power Method (\ref{standard pm}) at iteration $k$, i.e. ${\bm{v}}_k=\frac{\bm{A}^k{\bm{v}}_0}{\left\|\bm{A}^k{\bm{v}}_0\right\|_2}$. Then with probability at least $1-\exp{(-Cd)}$ with constant $C$, for any $\epsilon>0$, we have $\left\|{\bm{v}}_k-{\bm{v}}_1(\bm{A})\right\|_2^2 \le \epsilon$ if $k\ge\mathcal{O}\left(\frac{\log(\kappa\left(\Sigma\right)d/\epsilon)}{\operatorname{gap}}\right)$. Here ${\bm{v}}_1(\bm{A})$ is (one of) the top right eigenvector of $\bm{A}$, and $\kappa\left(\Sigma\right)$ is the condition number of the matrix of right eigenvectors, where $\bm{A}= \Sigma\Lambda\Sigma^{-1}$ is the eigen-decomposition of $\bm{A}$.
\end{thm}
\begin{proof}[Proof]
    Let $\left({\bm{v}}_i(\bm{A})\right)_{i=1,\ldots,d}$ be (one of) the right eigenvectors corresponding to the eigenvalues $\left(\lambda_i(\bm{A})\right)_{i=1,\ldots,d}$. Let $\bm{c}=(c_1, c_2, \ldots, c_d)$ be a vector of constants that depend on $\bm{A}$ and ${\bm{v}}_0$ such that ${\bm{v}}_0=\sum_{j=1}^{d}c_j{\bm{v}}_j(\bm{A})$, i.e., ${\bm{v}}_0=\Sigma\bm{c}$. Then we obtain  
    \[
        \bm{A}^k{\bm{v}}_0=\sum_{i=1}^{d}c_i\bm{A}^k{\bm{v}}_i(\bm{A})=\sum_{i=1}^{d}c_i\lambda_i^k{\bm{v}}_i(\bm{A}).
    \]
    Accordingly, we have
    \[
        \frac{\left\|\bm{A}^k{\bm{v}}_0-c_1\lambda_1^k{\bm{v}}_1(\bm{A})\right\|_2}{\left|\lambda_1^k\right|}\le \frac{\left\|\sum_{i=2}^{d}c_i\lambda_i^k{\bm{v}}_i(\bm{A})\right\|_2}{\left|\lambda_1^k\right|} \le \max_{i\ge 2}|c_i|\cdot d \frac{\left|\lambda_2\right|^k}{\left|\lambda_1\right|^k},
    \]
    and
    \begin{equation}\label{eq: coefficient of power method}
            \frac{\left\|\bm{A}^k{\bm{v}}_0\right\|_2}{\left|\lambda_1^k\right|}\ge \frac{\left\|c_1\lambda_1^k{\bm{v}}_1(\bm{A})\right\|_2-\left\|\sum_{i=2}^{d}c_i\lambda_i^k{\bm{v}}_i(\bm{A})\right\|_2}{\left|\lambda_1^k\right|} \ge |c_1|- \max_{i\ge 2}|c_i| d \frac{\left|\lambda_2\right|^k}{\left|\lambda_1\right|^k}.    
    \end{equation}
    For the constants $c_i$, we note that $\frac{\max_i |{\bm{v}}_{0,i}|}{\min_i |{\bm{v}}_{0,i}|}$ has the same distribution as $\frac{\max_i \left|X_{i}\right|}{\min_i |X_{i}|}$, where $\bm{X} =(X_1, \ldots, X_d)^T \sim\mathcal{N}\left(0,\bm{I}\right)$. Thus we have $\frac{\max_i |{\bm{v}}_{0,i}|}{\min_i |{\bm{v}}_{0,i}|}=\mathcal{O}\left(\sqrt{d\log d}\right)$ with probability at least $1-\exp{(-Cd)}$. Then we have 
    \[\frac{\max_i |c_i|}{\min_i |c_i|}\le \kappa(\Sigma)\frac{\max_i |{\bm{v}}_{0,i}|}{\min_i |{\bm{v}}_{0,i}|}=\mathcal{O}\left(\kappa(\Sigma)\sqrt{d\log d}\right).\] If $k\ge\mathcal{O}\left(\frac{\log(\kappa\left(\Sigma\right)d)}{\operatorname{gap}}\right)$, the right hand side of Equation~(\ref{eq: coefficient of power method}) will be positive with high probability. Furthermore, we have
    \[
        \frac{\left\|\bm{A}^k{\bm{v}}_0\right\|_2}{\left|\lambda_1^k\right|}\le \frac{\sum_{i=1}^{d}\|c_i\lambda_i^k{\bm{v}}_i(\bm{A})\|_2}{\left|\lambda_1^k\right|} \le |c_1|+ \max_{i\ge 2}|c_i|\cdot d\cdot \frac{\left|\lambda_2\right|^k}{\left|\lambda_1\right|^k}.
    \]
Letting $c_0 =\max_{i\ge 2} \frac{|c_i|}{|c_1|}$ and $\tilde{\bm{v}}_1(\bm{A}) = \frac{c_1\lambda_1^k}{|c_1\lambda_1^k|} {\bm{v}}_1(\bm{A})$, we obtain
\begin{equation*}
\begin{aligned}
    &\left\|{\bm{v}}_k-\tilde{\bm{v}}_1(\bm{A})\right\|_2=\left\|\frac{\bm{A}^k{\bm{v}}_0}{\left\|\bm{A}^k{\bm{v}}_0\right\|_2}-\tilde{\bm{v}}_1(\bm{A})\right\|_2\\
    &\le \frac{\left\|\bm{A}^k{\bm{v}}_0-c_1\lambda_1^k{\bm{v}}_1(\bm{A})\right\|_2}{\left\|\bm{A}^k{\bm{v}}_0\right\|_2}+\left\|\frac{c_1\lambda_1^k{\bm{v}}_1(\bm{A})}{\left\|\bm{A}^k{\bm{v}}_0\right\|_2}-\tilde{\bm{v}}_1(\bm{A})\right\|_2\\
    &= \frac{\left\|\bm{A}^k{\bm{v}}_0-c_1\lambda_1^k{\bm{v}}_1(\bm{A})\right\|_2}{\left\|\bm{A}^k{\bm{v}}_0\right\|_2}+\left|\frac{c_1\lambda_1^k}{\left\|\bm{A}^k{\bm{v}}_0\right\|_2}-\frac{c_1\lambda_1^k}{|c_1\lambda_1^k|}\right|\\
    &\le \frac{\max_{i\ge 2}|c_i|\cdot d\cdot \frac{\left|\lambda_2\right|^k}{\left|\lambda_1\right|^k}}{\left||c_1|-\max_{i\ge 2}|c_i|\cdot d\cdot \frac{\left|\lambda_2\right|^k}{\left|\lambda_1\right|^k}\right|}+|c_1|\cdot\left|\frac{1}{|c_1|- \max_{i\ge 2}|c_i|\cdot d\cdot \frac{|\lambda_2|^k}{|\lambda_1|^k}}-\frac{1}{|c_1|}\right|\\
    &= \frac{2\max_{i\ge 2}|c_i|\cdot d\cdot \frac{|\lambda_2|^k}{\left|\lambda_1\right|^k}}{\left||c_1|-\max_{i\ge 2}|c_i|\cdot d\cdot \frac{\left|\lambda_2\right|^k}{|\lambda_1|^k}\right|}\le \frac{2  d\cdot \frac{\left|\lambda_2\right|^k}{|\lambda_1|^k}}{\frac{1}{|c_0|}-d\cdot \frac{\left|\lambda_2\right|^k}{|\lambda_1|^k}}=\frac{2}{\frac{1}{d\cdot|c_0|\cdot (1-\operatorname{gap})^k}- 1}.
\end{aligned}
\end{equation*}
Thus we have
\[
k\ge \mathcal{O}\left(\frac{\log(\kappa\left(\Sigma\right)d)+\log\left(\frac{2}{\epsilon}+1\right)}{\log (\frac{1}{1-\operatorname{gap}})}\right)\quad \Longrightarrow \quad\left\|{\bm{v}}_k-\tilde{\bm{v}}_1(\bm{A})\right\|_2\le \epsilon.
\]
\end{proof}
Summarily, we know that the complexity is $\mathcal{O}\left(\frac{\log(\kappa\left(\Sigma\right)d/\epsilon)}{\operatorname{gap}}\right)$.


\subsection{Examples}

We now present two instances, showing how close the upper and lower bounds are.

\textbf{Example 1 (small condition number).} Let $\bm{A}$ be a diagonal matrix with $\bm{A}_{11}=1$ and $\bm{A}_{ii}=1-\operatorname{gap}$ for $i=2,3,\dots, d$.

\begin{enumerate}
    \item[] \textbf{Power Method}. In iteration $k$, the output of the power method with initial vector ${\bm{v}}_0$ will be $c({\bm{v}}_{0,1}, (1-\operatorname{gap})^k{\bm{v}}_{0,2}, \ldots ,(1-\operatorname{gap})^k{\bm{v}}_{0,d-1},(1-\operatorname{gap})^k{\bm{v}}_{0,d})^\top$, where $c$ is the normalized constant. Because $\frac{\max_i |{\bm{v}}_{0,i}|}{\min_i |{\bm{v}}_{0,i}|}=\mathcal{O}\left(\sqrt{d\log d}\right)$, the complexity will be $\mathcal{O}\left(\frac{\log(d/\epsilon)}{\operatorname{gap}}\right)$, which matches our lower bounds well.
\end{enumerate}

\textbf{Example 2 (large condition number).} We consider the following matrix:  
\begin{equation*}
    \bm{\Sigma}:=\left[
\begin{array}{cc}
1 & 0 \\
\sin\theta & \cos\theta
\end{array}
\right], \bm{A} :=\bm{\Sigma}\operatorname{diag}(1,\operatorname{gap}-1)\bm{\Sigma}^{-1}.
\end{equation*}
\begin{enumerate}
    \item[] \textbf{Power Method}. At the end of iteration $k$, the output of the power method will be $\bm{v}_k = c\bm{\Sigma}\operatorname{diag}(1,(\operatorname{gap}-1)^k)\bm{\Sigma}^{-1}{\bm{v}}_0$, where $c$ is the normalized constant. Here, the condition number will be $\kappa(\Sigma)=\frac{1+\sin\theta}{1-\sin\theta}$. Denote $\bm{c} :=\bm{\Sigma}^{-1}{\bm{v}}_0$. The output will be $\left[
\begin{array}{c}
c_1+c_2(\operatorname{gap}-1)^k\sin\theta \\
c_2(\operatorname{gap}-1)^k\cos\theta
\end{array}
\right]$. Thus, the complexity will be $\mathcal{O}\left(\frac{\log(\kappa(\Sigma)d/\epsilon)}{\operatorname{gap}}\right)$ with at least constant probability.
    \item[] \textbf{Our bounds}. Our lower bound does not involve the condition number. Thus our lower bounds and the upper bounds show a significant discrepancy when the condition number is extremely large, e.g., $\theta\to\frac{\pi}{2}$. For the random matrix that we constructed, the condition number will be a polynomial of $d$, leading to the same bounds.
\end{enumerate}

\section{Missing Proofs of deformed i.i.d.\ Random Matrices}
\label{section: proof of deformed random matrix}

In this section, we provide the proof of Theorem~\ref{thm: outliers, uniform}, Proposition~\ref{prop: uWWu=0}, as well as the related lemmas.

\subsection{Proof of Theorem~\ref{thm: outliers, uniform}}
\begin{proof}[Proof]
    Given $\epsilon>0,\ \delta>0$ and $1+3\epsilon< \lambda_r$, denote $\mathcal{E}_{1}:=\left\{{\bm{M}}: \left|\lambda_1\left(\frac{1}{\sqrt{d}}{\bm{W}}\right)\right|\le 1+\epsilon\right\}$, by Lemma~\ref{lemma: largest eigenvalue of X}, we know $\mathcal{E}_{1}$ holds almost surely. Within $\mathcal{E}_1$, we have 
    \begin{equation*}
    \det\left(z\bm{I}-\frac{1}{\sqrt{d}}{\bm{W}}-{\bm{U}}\bm{\Lambda}{\bm{U}}^*\right)=\operatorname{det}\left(z \bm{I}-\frac{1}{\sqrt{d}}{\bm{W}}\right)\cdot \operatorname{det}\left(\bm{I}_r-\bm{\Lambda}{\bm{U}}^{*}\left(z \bm{I}-\frac{1}{\sqrt{d}}{\bm{W}}\right)^{-1} {\bm{U}}\right)     
    \end{equation*}
    for $|z|>1+2\epsilon$, because $z\bm{I}-\frac{1}{\sqrt{d}}{\bm{W}}$ is non-degenerate. Then we only need to consider the equation: 
    \begin{equation}\label{eq: second determinant}
        \operatorname{det}\left(\bm{I}_r-\bm{\Lambda}{\bm{U}}^{*}\left(z \bm{I}-\frac{1}{\sqrt{d}}{\bm{W}}\right)^{-1} {\bm{U}}\right)=0,
    \end{equation}
    whose roots are the eigenvalues of ${\bm{M}}$. Here we are going to apply the variants of Hanson-Wright inequalities (Lemma~\ref{lemma: HW general}). To do this, we need to bound the norm of $\left(z \bm{I}-\frac{1}{\sqrt{d}}{\bm{W}}\right)^{-1}$. By Lemma~\ref{lemma: inv(lambdaI-W)}, we have
    \begin{equation*}
        \left\|\left(z \bm{I}-\frac{1}{\sqrt{d}}{\bm{W}}\right)^{-1}\right\|_{2}\le \frac{2|z|}{(|z|-1)^2}+o(1)
    \end{equation*}
    with $o(1)\to0$ almost surely and hence,
    \begin{equation*}
        \left\|\left(z \bm{I}-\frac{1}{\sqrt{d}}{\bm{W}}\right)^{-1}\right\|_{\operatorname{F}}\le \sqrt{d}\cdot\left(\frac{2|z|}{(|z|-1)^2}+o(1) \right).
    \end{equation*}
    Taking $t=cd$ for an arbitrary $c>0$ in the asymmetric complex Stiefel Hanson-Wright (Lemma~\ref{lemma: HW general}), we obtain
    \begin{equation}\label{eq: result of HW}
    \left\|\boldsymbol{U^*}\left(z \bm{I}-\frac{1}{\sqrt{d}}{\bm{W}}\right)^{-1} \boldsymbol{U}-\frac{1}{d} \operatorname{tr}\left(\left(z \bm{I}-\frac{1}{\sqrt{d}}{\bm{W}}\right)^{-1}\right)\cdot\bm{I}_r \right\|_2=o(1).
    \end{equation}
    Here $\boldsymbol{U}$ can be either ${\bm{U}}\stackrel{\text{unif}}{\sim}\operatorname{Stief}(d, r)$ or ${\bm{U}}\stackrel{\text{unif}}{\sim}\operatorname{Stief}_\mathbb{C}(d, r)$. The $o(1)$ converges to 0 almost surely by the Borel-Cantelli Lemma. (The probabiity of concentration is exponential in $d$, and thus the summation converges.)

    We remain to evaluate the term $\operatorname{tr}\left(\left(z \bm{I}-\frac{1}{\sqrt{d}}{\bm{W}}\right)^{-1}\right)$, having
    \begin{equation}\label{eq: decomposion of stiej transform}
    \begin{aligned}
        &S_{\bm{W}}(z) := \frac{1}{d} \operatorname{tr}\left(\left(z \bm{I}-\frac{1}{\sqrt{d}}{\bm{W}}\right)^{-1}\right) \\
        =& \frac{1}{d} \operatorname{tr}\left(\left(z \bm{I}-\frac{1}{\sqrt{d}}{\bm{W}}\right)^{-1}\right)\mathbf{1}_{\mathcal{E}_{1}}+\frac{1}{d} \operatorname{tr}\left(\left(z \bm{I}-\frac{1}{\sqrt{d}}{\bm{W}}\right)^{-1}\right)\mathbf{1}_{\mathcal{E}_{1}^c}\\
        =& \frac{1}{z}\operatorname{tr}\left(1+\sum_{k=1}^{\infty} \frac{1}{z^k}\cdot\frac{1}{d} \left(\frac{1}{\sqrt{d}}{\bm{W}}\right)^k\right)\mathbf{1}_{\mathcal{E}_{1}}+\frac{1}{d} \operatorname{tr}\left(\left(z \bm{I}-\frac{1}{\sqrt{d}}{\bm{W}}\right)^{-1}\right)\mathbf{1}_{\mathcal{E}_{1}^c}\\
        =& \frac{1}{z}\left(1+\sum_{k=1}^{l} \frac{1}{z^k}\left(\frac{1}{d} \operatorname{tr}\left( \left(\frac{1}{\sqrt{d}}{\bm{W}}\right)^k\right)\right)\right)\mathbf{1}_{\mathcal{E}_{1}}\\
        &+\frac{1}{dz}\operatorname{tr}\left(\sum_{k=l+1}^{\infty} \frac{1}{z^k} \left(\frac{1}{\sqrt{d}}{\bm{W}}\right)^k\right)\mathbf{1}_{\mathcal{E}_{1}}+\frac{1}{d} \operatorname{tr}\left(\left(z \bm{I}-\frac{1}{\sqrt{d}}{\bm{W}}\right)^{-1}\right)\mathbf{1}_{\mathcal{E}_{1}^c}\\
        =& \frac{1}{z}\left(1+\sum_{k=1}^{l} \frac{1}{z^k}\left(\frac{1}{d} \operatorname{tr}\left( \left(\frac{1}{\sqrt{d}}{\bm{W}}\right)^k\right)\right)\right)\mathbf{1}_{\mathcal{E}_{1}}\\
        &+\frac{1}{d}\operatorname{tr}\left(\frac{{\bm{W}}^{l+1}}{d^{(l+1)/2}z^{l+1}} \left(z \bm{I} {-} \frac{{\bm{W}}}{\sqrt{d}}\right)^{-1}\right)\mathbf{1}_{\mathcal{E}_{1}}+\frac{1}{d} \operatorname{tr}\left(\left(z \bm{I} {-}\frac{{\bm{W}}}{\sqrt{d}}\right)^{-1}\right)\mathbf{1}_{\mathcal{E}_{1}^c}
    \end{aligned}
    \end{equation}
    for $|z|>1+2\epsilon$. Here $l\in\mathbb{N}^+$ remains to choose. We first consider the second term on the right-hand side of the above equation. Suppose the eigenvalue decomposition of $\frac{1}{\sqrt{d}}{\bm{W}}$ is $\frac{1}{\sqrt{d}}{\bm{W}}=\mathbf{\Sigma}\Bar{\mathbf\Lambda}\mathbf\Sigma^{-1}$ and the function $h(x):=x^{l+1}/(z-x)$ for convenience. We have
    \begin{equation}\label{eq: W^l+1 residual}
    \begin{aligned}
        &\left|\frac{1}{d}\operatorname{tr}\left(\frac{{\bm{W}}^{l+1}}{d^{(l+1)/2}z^{l+1}} \left(z \bm{I}-\frac{1}{\sqrt{d}}{\bm{W}}\right)^{-1}\right)\mathbf{1}_{\mathcal{E}_{1}}\right|\\
        \le& \left|\frac{1}{d}\operatorname{tr}\left(\frac{{\bm{W}}^{l+1}}{d^{(l+1)/2}z^{l+1}} \left(z \bm{I}-\frac{1}{\sqrt{d}}{\bm{W}}\right)^{-1}\right)\right|\\
        =& \left|\frac{1}{d}\operatorname{tr}\left(\frac{\Bar{\mathbf\Lambda}^{l+1}}{z^{l+1}} \left(z \bm{I}-\Bar{\mathbf\Lambda}\right)^{-1}\right)\right|=\left|\frac{1}{dz^{l+1}}\sum_{j=1}^d h(\lambda_j({\bm{W}}/\sqrt{d})) \right|\\ 
        \le& \frac{\left|\lambda^{l+1}_1\left(\frac{1}{\sqrt{d}}{\bm{W}}\right)\right|}{\left|z^{l+1}\right|\cdot\left|\left|z\right|-\left|\lambda_1\left(\frac{1}{\sqrt{d}}{\bm{W}}\right)\right|\right|}= \frac{1}{\left|z^{l+1}\right|\cdot(|z|-1)}+o(1).
    \end{aligned}
    \end{equation}
For the first term in the right-hand side of Equation~(\ref{eq: decomposion of stiej transform}), we have
\begin{equation*}
    \frac{1}{d} \operatorname{tr}\left( d^{-\frac{k}{2}}{\bm{W}}^k\right)\mathbf{1}_{\mathcal{E}_{1}}=\mathbb{E}_{\lambda}\left(\mathfrak{Re}\left(\lambda^k\right);\mathcal{E}_{1}\right)+i\mathbb{E}_{\lambda}\left(\mathfrak{Im}\left(\lambda^k\right);\mathcal{E}_{1}\right),
\end{equation*}
where the expectation is $\lambda\sim\lambda\left({\bm{W}}/\sqrt{d}\right)$ and $\lambda\left({\bm{W}}/\sqrt{d}\right)$ denotes the empirical distribution of the eigenvalues of ${\bm{W}}/\sqrt{d}$. Let
\begin{equation*}
    {\bm{U}}_d(s, t):=\frac{1}{d} \#\left\{k \leq d \mid \operatorname{Re}\left(\lambda_k\left({\bm{W}}/\sqrt{d}\right)\right) \leq s ; \operatorname{Im}\left(\lambda_k\left({\bm{W}}/\sqrt{d}\right)\right) \leq t\right\}
\end{equation*}
 be the CDF of the empirical distribution, and 
\begin{equation*}
    {\bm{U}}(s, t):=\frac{1}{\pi} \operatorname{mes}\left\{z\in \mathbb{C} \mid |z|\le 1; \operatorname{Re}\left(z\right) \leq s ; \operatorname{Im}\left(z\right) \leq t\right\}
\end{equation*}
 be the uniform distribution over the unit disk in $\mathbb{C}$, where $\operatorname{mes}(\cdot)$ denotes the Lebesgue measure. By Theorem~\ref{thm: strong circular law} (strong circular law \cite{tao2008random}), we have
\begin{equation*}
    {\bm{U}}_d(s, t)\to {\bm{U}}(s, t) \quad a.s.
\end{equation*}
 for all $s, t\in\mathbb{R}$. Then for any Borel set $\mathfrak{B}\subset\mathbb{R}^2$, the convergence that
\begin{equation}\label{eq: strong circular law}
\begin{aligned}
    &{\bm{U}}_d(\mathfrak{B}):=\frac{1}{d} \#\left\{k \leq d \mid \left(\operatorname{Re}(\lambda),\operatorname{Im}(\lambda)\right)\in \mathfrak{B}\right\}\\
    \to \quad &{\bm{U}}(\mathfrak{B}):=\frac{1}{\pi} \operatorname{mes}\left(\left\{z\in \mathbb{C} \mid |z|\le 1; (\operatorname{Re}(z),\operatorname{Im}(z))\in \mathfrak{B}\right\}\right)\quad a.s.
\end{aligned}
\end{equation}
holds. Letting $\epsilon_1>0$ be an arbitrary constant, we further get
\begin{equation}\label{eq: E(Re(lambda))}
\begin{aligned}
    &\mathbb{E}_{\lambda\sim\lambda\left({\bm{W}}/\sqrt{d}\right)}\left(\mathfrak{Re}\left(\lambda^k\right);\mathcal{E}_{1}\right)\\
    & = \mathbb{E}_{\lambda\sim\lambda\left({\bm{W}}/\sqrt{d}\right)}\sum_{j\in \mathbb{Z}}\left(\mathfrak{Re}\left(\lambda^k\right);\mathcal{E}_{1},\mathfrak{Re}\left(\lambda^k\right)\in [j\epsilon_1,(j+1)\epsilon_1)\right) \\
     & \ge \mathbb{E}_{\lambda\sim\lambda\left({\bm{W}}/\sqrt{d}\right)}\sum_{j\in \mathbb{Z}}\left(j\epsilon_1;\mathcal{E}_{1},\mathfrak{Re}\left(\lambda^k\right)\in \left[j\epsilon_1,\left(j+1\right)\epsilon_1\right)\right) \\
     & \stackrel{(a)}{=} \sum_{j=[- (1+\epsilon)^k/\epsilon_1]}^{[(1+\epsilon)^k/\epsilon_1]}j\epsilon_1\mathbb{E}_{\lambda\sim\lambda\left({\bm{W}}/\sqrt{d}\right)}\mathbf{1}_{\mathcal{E}_{1}}\mathbf{1}_{\mathfrak{Re}\left(\lambda^k\right)\in [j\epsilon_1,(j+1)\epsilon_1)} \\ 
    & = \sum_{j=[- (1+\epsilon)^k/\epsilon_1]}^{[(1+\epsilon)^k/\epsilon_1]}j\epsilon_1{\bm{U}}_d\left(\left\{\lambda\mid \mathfrak{Re}\left(\lambda^k\right)\in[j\epsilon_1,(j+1)\epsilon_1)\right\}\right)\\
    & \quad -\sum_{j=[- (1+\epsilon)^k/\epsilon_1]}^{[(1+\epsilon)^k/\epsilon_1]}j\epsilon_1\mathbb{E}_{\lambda\sim\lambda\left({\bm{W}}/\sqrt{d}\right)}\mathbf{1}_{\mathcal{E}_{1}^c}\mathbf{1}_{\mathfrak{Re}\left(\lambda^k\right)\in [j\epsilon_1,(j+1)\epsilon_1)}. \\ 
\end{aligned}
\end{equation}
We can converts the infinite summation into a finite summation in Equality $(a)$, because $|\mathfrak{Re}\left(\lambda^k\right)|\le (1+\epsilon)^k$ under the event $\mathcal{E}_{1}$. Here [$\cdot$] denotes the rounding function. Further for the first term in the right-hand side of Equation~(\ref{eq: E(Re(lambda))}), we have
\begin{equation}\label{eq: strong circular law 2}
\begin{aligned}
    &{\bm{U}}_d\left(\left\{\lambda\mid\left(\mathfrak{Re}\left(\lambda^k\right),\mathfrak{Im}\left(\lambda^k\right)\right)\in[j\epsilon_1,(j+1)\epsilon_1\times (-\infty,+\infty)\right\}\right)\\
    \stackrel{(\ref{eq: strong circular law})}{\to} \quad &{\bm{U}}\left(\left\{\lambda|\left(\mathfrak{Re}\left(\lambda^k\right),\mathfrak{Im}\left(\lambda^k\right)\right)\in[j\epsilon_1,(j+1)\epsilon_1)\times (-\infty,+\infty)\right\}\right)\quad a.s..
\end{aligned}
\end{equation}
For the second term in the right-hand side of Equation~(\ref{eq: E(Re(lambda))}), we have $\mathbf{1}_{\mathcal{E}_{1}^c}\to 0$ a.s. by Lemma~\ref{lemma: largest eigenvalue of X}. Then by the dominated convergence theorem, the second term converges to 0 with probability 1. Thus we have
\begin{equation*}
\begin{aligned}
    &\varliminf_{d\to \infty}\mathbb{E}_{\lambda\sim\lambda\left({\bm{W}}/\sqrt{d}\right)}\left(\mathfrak{Re}\left(\lambda^k\right);\mathcal{E}_{1}\right) \\
  &  \stackrel{(\ref{eq: E(Re(lambda))})(\ref{eq: strong circular law 2})}{\ge} \sum_{j=[- (1+\epsilon)^k/ \epsilon_1]}^{[(1+\epsilon)^k/\epsilon_1]}j\epsilon_1{\bm{U}}\left(\left\{\lambda|\mathfrak{Re}\left(\lambda^k\right)\in[j\epsilon_1,(j+1)\epsilon_1)\right\}\right)\\
    & \ge \; \sum_{j=[- (1+\epsilon)^k/\epsilon_1]}^{[(1+\epsilon)^k/\epsilon_1]}\mathbb{E}_{\lambda\sim{\bm{U}}}\left(\mathfrak{Re}\left(\lambda^k\right)-\epsilon_1\right)\mathbf{1}_{(\mathfrak{Re}\left(\lambda^k\right),\mathfrak{Im}\left(\lambda^k\right))\in[j\epsilon_1,(j+1)\epsilon_1)\times (-\infty,+\infty)}\\
    & = \; \sum_{j=[- (1+\epsilon)^k/\epsilon_1]}^{[(1+\epsilon)^k/\epsilon_1]}\mathbb{E}_{\lambda\sim{\bm{U}}}\left(\mathfrak{Re}\left(\lambda^k\right)\mathbf{1}_{(\mathfrak{Re}\left(\lambda^k\right),\mathfrak{Im}\left(\lambda^k\right))\in[j\epsilon_1,(j+1)\epsilon_1)\times (-\infty,+\infty)}\right)-\epsilon_1\\
    &= \mathbb{E}_{\lambda\sim{\bm{U}}}\mathfrak{Re}\left(\lambda^k\right)\mathbf{1}_{(\mathfrak{Re}\left(\lambda^k\right),\mathfrak{Im}\left(\lambda^k\right))\in[-(1+\epsilon)^k,(1+\epsilon)^k)\times (-\infty,+\infty)}-\epsilon_1\\
    & = -\epsilon_1, \quad a.s..
\end{aligned}
\end{equation*}
The last equality is because $\lambda\sim{\bm{U}}$ follows the uniform distribution over the unit dist in $\mathbb{C}$. For the same reason, we have $\varlimsup\mathbb{E}_{\lambda\sim\lambda\left({\bm{W}}/\sqrt{d}\right)}\left(\mathfrak{Re}\left(\lambda^k\right);\mathcal{E}_{1}\right)\le \epsilon_1$ a.s. for any $\epsilon_1>0$. Thus we obtain
\begin{equation*}
\mathbb{E}_{\lambda\sim\lambda\left({\bm{W}}/\sqrt{d}\right)}\left(\mathfrak{Re}\left(\lambda^k\right);\mathcal{E}_{1}\right)\to 0 \quad \text{a.s..}
\end{equation*}
Similar result also holds that $\mathbb{E}_{\lambda\sim\lambda\left({\bm{W}}/\sqrt{d}\right)}\left(\mathfrak{Im}\left(\lambda^k\right);\mathcal{E}_{1}\right)\to 0$ a.s., so almost surely we have
\begin{equation}\label{eq: strong circular law3}
\frac{1}{d} \operatorname{tr}\left( d^{-k/2}{\bm{W}}^k\right)\mathbf{1}_{\mathcal{E}_{1}}=\mathbb{E}_{\lambda\sim\lambda\left({\bm{W}}/\sqrt{d}\right)}\left(\mathfrak{Re}\left(\lambda^k\right);\mathcal{E}_{1}\right)+i\mathbb{E}_{\lambda\sim\lambda\left({\bm{W}}/\sqrt{d}\right)}\left(\mathfrak{Im}\left(\lambda^k\right);\mathcal{E}_{1}\right)\to 0.
\end{equation}
Then we have
\begin{equation}\label{eq: residual}
\begin{aligned}
    &\frac{1}{d} \operatorname{tr}\left(\left(z \bm{I}-\frac{1}{\sqrt{d}}{\bm{W}}\right)^{-1}\right)\mathbf{1}_{\mathcal{E}_{1}}-\frac{1}{d}\operatorname{tr}\left(\frac{{\bm{W}}^{l+1}}{d^{(l+1)/2}z^{l+1}} \left(z \bm{I}-\frac{1}{\sqrt{d}}{\bm{W}}\right)^{-1}\right)\mathbf{1}_{\mathcal{E}_{1}}\\
   & \stackrel{(\ref{eq: decomposion of stiej transform})}{=} \frac{1}{z}\left(1+\sum_{k=1}^{l} \frac{1}{z^k}\left(\frac{1}{d} \operatorname{tr}\left( d^{-k/2}{\bm{W}}^k\right)\right)\right)\mathbf{1}_{\mathcal{E}_{1}}\stackrel{(\ref{eq: strong circular law3})}{\to} \frac{1}{z} \quad \text{a.s.}.
\end{aligned}
\end{equation}
We have known that $\mathbf{1}_{\mathcal{E}_{1}^c}\to 0$ almost surely. Thus, $
    \frac{1}{d} \operatorname{tr}\left(\left(z \bm{I}-\frac{1}{\sqrt{d}}{\bm{W}}\right)^{-1}\right)\mathbf{1}_{\mathcal{E}_{1}^c}\to 0$ almost surely. Combining with Equations~(\ref{eq: residual}) and (\ref{eq: W^l+1 residual}), we obtain that
\begin{equation*}
    \left|\frac{1}{d} \operatorname{tr}\left(\left(z \bm{I}-\frac{1}{\sqrt{d}}{\bm{W}}\right)^{-1}\right)-\frac{1}{z}\right|\le \frac{1}{\left|z^{l+1}\right|\cdot(|z|-1)}+o(1)
\end{equation*}
for all $l\in\mathbb{N}$ and $|z|>1+2\epsilon$ . Taking sufficiently large $l$, we have
\[\left|\frac{1}{d} \operatorname{tr}\left(\left(z \bm{I}-\frac{1}{\sqrt{d}}{\bm{W}}\right)^{-1}\right)-\frac{1}{z}\right|=o(1).\]
Then we further get
\begin{equation*}
    \boldsymbol{U^*}\left(z \bm{I}-\frac{1}{\sqrt{d}}{\bm{W}}\right)^{-1} \boldsymbol{U}\stackrel{(\ref{eq: result of HW})}{=}\bm{I}_r\cdot\operatorname{tr}\left(\left(z \bm{I}-\frac{1}{\sqrt{d}}{\bm{W}}\right)^{-1}\right)/d +o(1)=\frac{\bm{I}_r}{z}+o(1).
\end{equation*}
Combine with the determinant Equation~\eqref{eq: second determinant}, we obtain
\begin{equation}\label{eq: second determinant 2}
    \operatorname{det}\left(\bm{I}_r-\bm{\Lambda}{\bm{U}}^{*}\left(z \bm{I}-\frac{1}{\sqrt{d}}{\bm{W}}\right)^{-1} {\bm{U}}\right)=\operatorname{det}\left(\bm{I}_r-\frac{\bm{\Lambda}}{z}\right)+o(1).
\end{equation}
Note that we can't conclude that there is a single eigenvalue $\lambda_i({\bm{M}})=\lambda_i+o(1)$ for $i\in[r]$ in the region $\{z||z|>1+2\epsilon\}$. Because the $o(1)$ in the right-hand side in Equation~(\ref{eq: second determinant 2}) is still a function of $z$. We then apply Roche's theorem to the function $\operatorname{det}\left(\bm{I}_r-\bm{\Lambda}{\bm{U}}^{*}\left(z \bm{I}-\frac{1}{\sqrt{d}}{\bm{W}}\right)^{-1} {\bm{U}}\right)$ in the region $|z|>1+3\epsilon$. Let $g(z):=\operatorname{det}\left(\bm{I}_r-\frac{\bm{\Lambda}}{z}\right)$ and
\begin{equation*}
    f_{{\bm{W}}, {\bm{U}}}(z):=\operatorname{det}\left(\bm{I}_r-\bm{\Lambda}^\frac{1}{2}{\bm{U}}^{*}\left(z \bm{I}-\frac{1}{\sqrt{d}}{\bm{W}}\right)^{-1} {\bm{U}}\bm{\Lambda}^\frac{1}{2}\right).
\end{equation*}
Let $\Gamma:=\left\{z||z|=1+3\epsilon\right\}$ and $D:=\{z||z|\ge 1+3\epsilon\}$. Then we have
\begin{equation*}
\left|f_{{\bm{W}}, {\bm{U}}}(z)-g(z)\right|=\left|\operatorname{det}\left(\bm{I}_r-\bm{\Lambda}^\frac{1}{2}{\bm{U}}^{*}\left(z \bm{I}-\frac{1}{\sqrt{d}}{\bm{W}}\right)^{-1} {\bm{U}}\bm{\Lambda}^\frac{1}{2}\right)-\operatorname{det}\left(\bm{I}_r-\frac{\bm{\Lambda}}{z}\right)\right|=o(1)
\end{equation*}
on $\Gamma$ and
\begin{equation*}
    |g(z)|=\left|\operatorname{det}\left(\bm{I}_r-\frac{\bm{\Lambda}}{z}\right)\right|\ge\prod_{i=1}^r\left(\frac{1}{1+3\epsilon}-\frac{1}{\lambda_i}\right) \quad \text{on} \ \Gamma.
\end{equation*}
Thus we have
\begin{equation*}
    \mathbb{P}\left(\left|f_{{\bm{W}}, {\bm{U}}}(z)-g(z)\right|\le |g(z)|\right)\to 1 \quad \text{on} \ \Gamma.
\end{equation*}
By Roche's theorem, once $|f_{{\bm{W}}, {\bm{U}}}(z)-g(z)|\le |g(z)|$ on $\Gamma$, then $f_{{\bm{W}}, {\bm{U}}}(z)$ and $g(z)$ have same number of roots in $D$. Thus $f_{{\bm{W}}, {\bm{U}}}(z)$ has only $r$ roots in $D$, that is, there are only $r$ eigenvalues in $D$. Thus we obtain 
\begin{equation*}
    \lambda_i({\bm{M}})=\lambda_i+o(1) \quad \text{for}\quad i\in[r],
\end{equation*}
and $|\lambda_{r+1}({\bm{M}})|\le 1+3\epsilon+o(1).$ Taking $\epsilon\to 0$, we know that $|\lambda_{r+1}({\bm{M}})|\le 1+o(1)$, and the $o(1)$ converges to 0 almost surely.
\end{proof}

\subsection{Proof of Proposition~\ref{prop: uWWu=0}}
Let $\bm{u}$, $\bm{W}$ be the random variables defined in Proposition~\ref{prop: uWWu=0}. By the standard truncation argument \cite{Bai1988Necessary}, we can suppose the entires of $\bm{W}$ to be bounded. We first reduce the proposition to the case where the entries of $\bm{W}$ are Gaussian variables. Let $\bm{W}_g$ be an independent complex Gaussian i.i.d.\ matrix, by Lemma~\ref{lemma: non-gaussian tech EW-EG=0}, we have
\begin{equation*}
    \mathbb{E}\left({\bm{u}}^*\left(\frac{\bm{W}}{\sqrt{d}}\right)^{k_1*}\left(\frac{\bm{W}}{\sqrt{d}}\right)^{k_2}{\bm{u}}- \mathbf{1}_{k_1=k_2 }\right)^2-\left({\bm{u}}^*\left(\frac{\bm{W}_g}{\sqrt{d}}\right)^{k_1*}\left(\frac{\bm{W}_g}{\sqrt{d}}\right)^{k_2}{\bm{u}}- \mathbf{1}_{k_1=k_2 }\right)^2  = o_d(1).
\end{equation*}
Then it suffices to show that the variance of ${\bm{u}}^*\left(\frac{\bm{W}_g}{\sqrt{d}}\right)^{k_1*}\left(\frac{\bm{W}_g}{\sqrt{d}}\right)^{k_2}{\bm{u}}- \mathbf{1}_{k_1=k_2}$ is $o_d(1)$. Then Proposition~\ref{prop: uWWu=0} follows by Markov's inequality. We then finish the proof with the following lemma.
\begin{lemma}\label{lemma: uGGu=0}
    Let $\bm{W}_g$ be an independent complex Gaussian i.i.d.\ matrix, ${\bm{u}}$ be an independent random vector over $\mathcal{S}_\mathbb{C}^{d-1}$ and fixed $k_1, k_2\in\mathbb{N}$. Then we have 
    \begin{equation*}
    \mathbb{E}\left({\bm{u}}^*\left(\frac{\bm{W}_g}{\sqrt{d}}\right)^{k_1*}\left(\frac{\bm{W}_g}{\sqrt{d}}\right)^{k_2}{\bm{u}}- \mathbf{1}_{k_1=k_2 }\right)^2  = o_d(1).
    \end{equation*}
\end{lemma}
We will prove Lemma~\ref{lemma: uGGu=0} for cases $k_1=k_2$ and $k_1\neq k_2$ separately. We denote $\bm{G}=\bm{W}_g/\sqrt{d}$ for notation convenience.
\paragraph{Case 1: $k_1=k_2$.} Note that the function $F({\bm{v}},\bm{A})={\bm{v}}^*\left(\bm{A}^k\right)^*\bm{A}^k{\bm{v}}$ is rotation invariance by Lemma~\ref{lemma: rotation invariance}. Thus we can take $\bm{u}$ to be an arbitrary deterministic unit vector. We take $\bm{u}=\bm{e}_1$. We remain to show
\begin{equation}\label{eq: non-gaussian target, 4}
\mathbb{E}\left(\left(\bm{G}^k\right)^*\bm{G}^k\right)_{11} =1+ o(1) \quad\text{and}\quad \mathbb{E} \left(\left(\bm{G}^k\right)^*\bm{G}^k\right)^2_{11} =1+ o(1).
\end{equation}
To get the first equation in Equation~(\ref{eq: non-gaussian target, 4}), we expand the left-hand side as
\begin{equation}\label{eq: EGG=}
\mathbb{E}\left(\left(\bm{G}^k\right)^*\bm{G}^k\right)_{11}=
\sum_{\bm{a},\bm{b}\in [d]^{k+1}}\mathbb{E}\prod_{i=1}^{k}G_{a_{i-1}a_i}G^*_{b_ib_{i-1}}
\end{equation}
where $\bm{a}=(a_0,a_1,\ldots,a_{k}),\bm{b}=(b_0,b_1,\ldots,b_{k})$ are indexes with $a_0=b_0=1, a_k=b_k$.


Let $g(\bm{a},\bm{b})$ be the number of distinct indexes of $(a_{i}), (b_j), i,j=0,1,2,\ldots, k$. Notice that if $g(\bm{a},\bm{b})\ge k+2$, then there will be at least one index that appears only once in $(a_{i}), (b_j), i,j=0,1,2,\ldots, k$, leading to zero mean. Thus we only need to consider the case $g(\bm{a},\bm{b})\le k+1$, further we can decompose the summation in Equation (\ref{eq: EGG=}) as
\begin{equation}\label{eq: EGG=EGG1+EGG2}
\sum_{\begin{array}{c}
\bm{a},\bm{b}\in [d]^{k+1}, \\
g(\bm{a},\bm{b})=k+1
\end{array}}
\mathbb{E}\prod_{i=1}^{k}G_{a_{i-1}a_i}G^*_{b_{i}b_{i-1}}+ \sum_{\begin{array}{c}
\bm{a},\bm{b}\in [d]^{k+1}, \\
g(\bm{a},\bm{b})\le k
\end{array}}\mathbb{E}\prod_{i=1}^{k}G_{a_{i-1}a_i}G^*_{b_ib_{i-1}}.
\end{equation}
The second summation will be $\mathcal{O}\left(d^{-1}\right)=o(1)$, because the expectation for each term in the summation will be $\mathcal{O}(d^{-k})$ but the summation will be at most $\mathcal{O}(d^{k-1})$. Thus we remain to bound the first term. Notice that under the condition that $g(\bm{a},\bm{b})= k+1$, each index should appear exactly twice, otherwise there will be at least one index that appears only once and the expectation will be zero. Under this condition, in order to get the non-zero expectation, we should match $(a_0, a_1)$ and $(b_0, b_1)$ because $b_1$ and $a_1$ will only apprear twice, and thus $b_1=a_1$. For the same reason, we obtain $a_i=b_i$ for $i\in[k]$. Thus, we obtain
\begin{equation*}
\begin{aligned}
\mathbb{E}\left(\left(\bm{G}^k\right)^*\bm{G}^k\right)_{11}&=
\sum_{\begin{array}{c}
\bm{a},\bm{b}\in [d]^{k+1}, \\
\bm{a}=\bm{b}\\
g(\bm{a},\bm{b})=k+1
\end{array}}\mathbb{E}\prod_{i=1}^{k}G_{a_{i-1}a_i}G^*_{b_ib_{i-1}}+\mathcal{O}(\frac{1}{d})\\
&=\sum_{\begin{array}{c}
\bm{a}\in [d]^{k+1}, \\
g(\bm{a},\bm{a})=k+1
\end{array}}\mathbb{E}\prod_{i=1}^{k}G_{a_{i-1}a_i}G^*_{a_ia_{i-1}}+\mathcal{O}(\frac{1}{d})\\
&=\sum_{\begin{array}{c}
\bm{a}\in [d]^{k+1}, \\
g(\bm{a},\bm{a})=k+1
\end{array}}\frac{1}{d^k}+\mathcal{O}(\frac{1}{d})=\frac{A_d^k}{d^k}+\mathcal{O}(\frac{1}{d})=1+\mathcal{O}(\frac{1}{d}).
\end{aligned}
\end{equation*}
Then we get the first equation in Equation~(\ref{eq: non-gaussian target, 4}). For the second equation, we expand its left-hand side as
\begin{equation*}\mathbb{E}\left(\left(\bm{G}^k\right)^*\bm{G}^k\right)^2_{11}=
\sum_{\bm{a},\bm{b},\bm{c},\bm{d}\in [d]^{k+1}}\mathbb{E}\prod_{i=1}^{k}G_{a_{i-1}a_i}G^*_{b_ib_{i-1}}G_{c_{i-1}c_i}G^*_{d_id_{i-1}},
\end{equation*}
$\bm{a}=(a_0,a_1,\ldots,a_{k}),\bm{b}=(b_0,b_1,\ldots,b_{k}), \bm{c}=(c_0,c_1,\ldots,c_{k}), \bm{d}=(d_0,d_1,\ldots,d_{k})$ are the indexes with $a_0=b_0=c_0=d_0=1, a_k=b_k, c_k=d_k$. Let $g(\bm{a},\bm{b},\bm{c},\bm{d})=2k+1$ be the number of different indexes among $\bm{a},\bm{b},\bm{c},\bm{d}$. For similar reasons, we only need to consider the case $g(\bm{a},\bm{b},\bm{c},\bm{d})=2k+1$, the summation of the remainder will be $\mathcal{O}(d^{-1})=o(1)$.

In the case $g(\bm{a},\bm{b},\bm{c},\bm{d})=2k+1$, each index will appear only twice except $a_0=b_0=c_0=d_0=1$. Thus in order to get non-zero expectation, we need to match the pair $(a_{k-1}, a_k)$ and $(b_{k-1}, b_k)$ and thus $a_{k-1}=b_{k-1}$. For the same reason, we obtain $\bm{a}=\bm{b}$ and $\bm{c}=\bm{d}$.
\begin{equation*}
\begin{aligned}
&\mathbb{E}\left(\left(\bm{G}^k\right)^*\bm{G}^k\right)^2_{11}
\\&=
\sum_{\begin{array}{c}
\bm{a},\bm{b},\bm{c},\bm{d}\in [d]^{k+1}, \\
\bm{a}=\bm{b}, \bm{c}=\bm{d}\\
g(\bm{a},\bm{b},\bm{c},\bm{d})=2k+1
\end{array}}\mathbb{E}\prod_{i=1}^{k}G_{a_{i-1}a_i}G^*_{b_ib_{i-1}}G_{c_{i-1}c_i}G^*_{d_id_{i-1}}+\mathcal{O}(\frac{1}{d})\\
&=\sum_{\begin{array}{c}
\bm{a},\bm{c}\in [d]^{k+1}, \\
g(\bm{a},\bm{a},\bm{c},\bm{c})=2k+1
\end{array}}\mathbb{E}\prod_{i=1}^{k}G_{a_{i-1}a_i}G^*_{a_ia_{i-1}}G_{c_{i-1}c_i}G^*_{c_ic_{i-1}}+\mathcal{O}(\frac{1}{d})\\
&=\sum_{\begin{array}{c}
\bm{a},\bm{c}\in [d]^{k+1}, \\
g(\bm{a},\bm{a},\bm{c},\bm{c})=2k+1
\end{array}}\frac{1}{d^{2k}}+\mathcal{O}(\frac{1}{d})=1+\mathcal{O}(\frac{1}{d}).
\end{aligned}
\end{equation*}
Then we finish the case $k_1=k_2$.

\paragraph{Case 2: $k_1\neq k_2$.} Similar to the case $k_1=k_2$, we take $\bm{u}=\bm{e}_1$ because of the rotation invariance. We remain to show
\begin{equation}\label{eq: non-gaussian target, 6}
\mathbb{E}\left|\left(\left(\bm{G}^{k_1}\right)^*\bm{G}^{k_2}\right)_{11}\right|^2 = o(1).
\end{equation}
We also expand the left-hand side as
\begin{equation*}
\mathbb{E}\left|\left(\left(\bm{G}^{k_1}\right)^*\bm{G}^{k_2}\right)_{11}\right|^2=\sum_{\begin{array}{c}
\bm{a},\bm{c}\in [d]^{k_1+1},\\
\bm{b},\bm{d}\in [d]^{k_2+1}      
\end{array}}\mathbb{E}\prod_{i=1}^{k_1}G_{a_{i-1}a_i}G^*_{d_{i}d_{i-1}}\prod_{i=1}^{k_2}G^*_{b_{i}b_{i-1}}G_{c_{i-1}c_i},
\end{equation*}
$\bm{a}=(a_0,a_1,\ldots,a_{k_1}),\bm{b}=(b_0,b_1,\ldots,b_{k_2}), \bm{c}=(c_0,c_1,\ldots,c_{k_2}),\bm{d}=(d_0,d_1,\ldots,d_{k_1})$ are indexes with $a_0=b_0=c_0=d_0=1, a_{k_1}=b_{k_2}, c_{k_2}=d_{k_1}$. For similar reasons, it suffices to consider the case $g(\bm{a},\bm{b},\bm{c},\bm{d})=k_1+k_2+1$ and the remainder will be at most $O(d^{-1})=o(1)$ in total. Further we only need to consider the case that each index appears twice except that $a_0=b_0=c_0=d_0=1$ appear four times. Without loss of generality, we assume $k_1<k_2$. For the same reason, we need to match $(a_{k_1-1}, a_{k_1})$ and $(b_{k_2-1}, b_{k_2})$ in the case $g(\bm{a},\bm{b},\bm{c},\bm{d})=k_1+k_2+1$. Thus we know $a_{k_1-1}=b_{k_2-1}$. Similarly, we obtain $a_{k_1-i}=b_{k_2-i}$ for $i\in[k_1]$, leading to $a_{0}=b_{k_2-k_1}=1$, which makes the index $1$ appear at least five times. Thus the summation will be zero in this case. Then we finish the case $k_1\neq k_2$.

\subsection{Technical Lemmas}

\begin{lemma}\label{lemma: inv(lambdaI-W)}
    Let ${\bm{W}}$ be an $d\times d$ i.i.d.\ matrix, $\lambda>1$ and $\operatorname{gap}:=\frac{\lambda-1}{\lambda}$. Then we have
    \begin{equation*}
    \left\|\lambda\left(\lambda\bm{I}-\frac{1}{\sqrt{d}}{\bm{W}}\right)^{-1}\right\|_2\le \frac{2}{\operatorname{gap}^2}+o(1),
    \end{equation*}
    here the $o(1)$ converges to 0 almost surely.
\end{lemma}
\begin{proof}[Proof]
    Denoting $\mathcal{E}_1:=\left\{\left|\lambda_1\left(\frac{1}{\sqrt{d}}{\bm{W}}\right)\right|\le 1+\epsilon\right\}$ where $\epsilon>0$ and $1+2\epsilon<\lambda$, we know $\mathcal{E}_1$ holds almost surely by Lemma~\ref{lemma: largest eigenvalue of X}. Within the event $\mathcal{E}_1$, we have $\left|\lambda_1\left(\frac{1}{\sqrt{d}}{\bm{W}}\right)\right|< \lambda$. Then for any $k\in\mathbb{N}^+$, we obtain
    \begin{equation*}
    \begin{aligned}
    &\left\|\lambda\left(\lambda\bm{I}-\frac{1}{\sqrt{d}}{\bm{W}}\right)^{-1}\right\|_2=\left\|\left(\bm{I}-\frac{1}{\sqrt{d}\lambda}{\bm{W}}\right)^{-1}\right\|_2
    =\left\|\bm{I}+\sum_{t=1}^\infty \left(\frac{1}{\sqrt{d}\lambda}{\bm{W}}\right)^t\right\|_2\\
    &\le 1+\sum_{t=1}^\infty\frac{1}{\left|\sqrt{d}\lambda\right|^t}\left\| {\bm{W}}^t\right\|_2= 1+\sum_{s=0}^\infty\sum_{t=1}^{k}\frac{1}{\left|\sqrt{d}\lambda\right|^{sk+t}}\left\| {\bm{W}}^{sk+t}\right\|_2\\
    &\stackrel{(a)}{\le} \left(1+\sum_{t=1}^{k}\frac{1}{\left|\sqrt{d}\lambda\right|^t}\left\| {\bm{W}}^t\right\|_2\right)\left(1+\sum_{s=1}^\infty\frac{\left\|{\bm{W}}^k\right\|_2^s}{\left|\sqrt{d}\lambda\right|^{sk}}\right)\\
    &= \left(1+\sum_{t=1}^{k}\frac{1}{\left|\sqrt{d}\lambda\right|^t}\left\| {\bm{W}}^t\right\|_2\right) \times \frac{1}{1-\frac{\left\|{\bm{W}}^k\right\|_2}{\left|\sqrt{d}\lambda\right|^k}}\\
    &\stackrel{(b)}{\le} \left(1+\sum_{t=1}^{k-1}\frac{t+1}{|\lambda|^t}+o(1)\right)\times \frac{1}{1-\frac{k+1+o(1)}{|\lambda|^k}}\stackrel{(c)}{\le}\frac{1}{\operatorname{gap}^2}+o(1).
    \end{aligned}
    \end{equation*}
To obtain Inequality $(a)$, we apply the inequality $\left\|{\bm{W}}^a\right\|_2\le \left\|{\bm{W}}^k\right\|_2^{[a/k]}\times \left\|{\bm{W}}^{a-k[a/k]}\right\|_2 $.  Inequality $(b)$ is derived from Lemma~\ref{lemma: theorem 5.17}. To get  Equality $(c)$, we first choose a sufficiently large $k$ such that $\frac{k+1}{|\lambda|^k}\le \epsilon$ for any $\epsilon>0$. Further because $1_{\mathcal{E}_{1}}\to 1$ almost surely, we conclude the result.
\end{proof}

\begin{lemma}\label{lemma: non-gaussian tech EW-EG=0} Let $\bm{W}_g$ be a (complex) Gaussian i.i.d.\ matrix, ${\bm{W}}$ be an independent i.i.d.\ matrix with bounded entries, $\bm{v}$ and $\bm{u}$ be random vectors over $\operatorname{Steif}_\mathbb{C}(d,1)$. Suppose $\bm{v}$ and $\bm{u}$ are independent with $\bm{W}$ ($\bm{v}$ and $\bm{u}$ may be not independent). Let $k\in\mathbb{N}$, $\bm{c}=(c_1, c_2,\ldots, c_k)\in\left\{\bm{+,-}\right\}^k$. For a matrix $\bm{A}$, we denote $\bm{A}^+=\bm{A}^*$ and $\bm{A}^-=\bm{A}$. Then we have
\begin{equation*}
\mathbb{E}d^{-k}\bm{u}^*\prod_{i=1}^k\bm{W}^{c_i}\left(\prod_{i=1}^k\bm{W}^{c_i}\right)^*\bm{v}- d^{-k}\bm{u}^*\prod_{i=1}^k\bm{W}_g^{c_i}\left(\prod_{i=1}^k\bm{W}_g^{c_i}\right)^*\bm{v} = o_d(1).
\end{equation*}
and 
\begin{equation*}
\mathbb{E}d^{-k}\left|\bm{u}^*\prod_{i=1}^k\bm{W}^{c_i}\bm{v}\right|^2-d^{-k}\left|\bm{u}^*\prod_{i=1}^k\bm{W}_g^{c_i}\bm{v}\right|^2 = o_d(1).
\end{equation*}
\end{lemma}
\begin{proof}[Proof]
Without loss of generality, it suffices to deal with the case where all the random variables are real because we can sperate each random variable to real and imaginary parts. Meanwhile, we denote $\bm{G}=\bm{W}_g$ for notation conveniance. 

\paragraph{The first equation.} To show the first result, we expand its left-hand side as
\begin{equation}\label{eq: EWW-EGG}
d^{-k}\sum_{\bm{a},\bm{b}\in [d]^{k+1}}u_{a_0}v^*_{b_0}\times\left(\mathbb{E}\prod_{i=1}^{k}W^{c_i}_{a_{i-1}a_i}\prod_{i=1}^{k}W^{c_i*}_{b_ib_{i-1}}-\prod_{i=1}^{k}G^{c_i}_{a_{i-1}a_i}\prod_{i=1}^{k}G^{c_i*}_{b_ib_{i-1}}\right)
\end{equation}
where $\bm{a}=(a_0,a_1,\ldots,a_{k}),\bm{b}=(b_0,b_1,\ldots,b_{k})$ are indexes with $a_k=b_k$. Consider the of $2k$ ordered pairs $(a_{i-1}, a_i)^{c_i}, (b_{j-1}, b_j)^{c_j}$ with $i=1,2,\ldots,k$, here $(a,b)^{+}\colon=(a,b)$ and $(a,b)^{-}\colon=(b, a)$. Because the entries of ${\bm{W}}$ and $\bm{G}$ are all i.i.d. with zero mean and unit variance, each term in the summation of Equation~(\ref{eq: EWW-EGG}) equals 0 unless each ordered pair appears with multiplicity at least two, and at least one pair appears at least three times. Thus, there are at most $k-1$ distinct ordered pairs. Because each term in Equation~(\ref{eq: EWW-EGG}):
\[\left|\mathbb{E}\prod_{i=1}^{k}W^{c_i}_{a_{i-1}a_i}\prod_{i=1}^{k}W^{c_i*}_{b_ib_{i-1}}-\prod_{i=1}^{k}G^{c_i}_{a_{i-1}a_i}\prod_{i=1}^{k}G^{c_i*}_{b_ib_{i-1}}\right|\]
is bounded. Then we can bound the summation as $\mathcal{O}\left(d^{-k}\sum_{*}\left|u_{a_0}u^*_{b_0}\right|\right)$, and hence, bounded by $\mathcal{O}\left(d^{-k}\sum_{*}|u_{a_0}|^2|+|u_{b_0}|^2\right)$.
Here the summation $*$ is the summation of choices $\bm{a},\bm{b}\in [d]^{k+1}$ which satisfy the at most $k-1$ distinct ordered pairs we mentioned above. For a fixed pair $(a_0, b_0)$, there is at most $\mathcal{O}(d^{k-1})$ choices of vectors $\bm{a},\bm{b}\in [d]^{k+1}$. Thus the summation can be further bounded by
\begin{equation*} \mathcal{O}\left(d^{-1}\sum_{a_0,b_0=1,2,\ldots, d}|u_{a_0}|^2|+|u_{b_0}|^2\right)=\mathcal{O}\left(d^{-1}\right)=o(1).
\end{equation*}

\paragraph{The seconde equation.} To show the second result, we expand its left-hand side as
\begin{equation*}
    d^{-k}\mathbb{E}\sum_{\bm{a},\bm{b}\in [d]^{k+1} }u_{a_0}^*v_{a_k}v^*_{b_k}u_{b_0}\times\left(\prod_{i=1}^{k}W^{c_i}_{a_{i-1}a_i}W^{c_i*}_{b_{i}b_{i-1}}-\prod_{i=1}^{k}G^{c_i}_{a_{i-1}a_i}G^{c_i*}_{b_{i}b_{i-1}}\right),
\end{equation*}
where $\bm{a}=(a_0,a_1,\ldots,a_{k}),\bm{b}=(b_0,b_1,\ldots,b_{k})$ are indexes. Consider the of $2k$ ordered pairs $(a_{i-1}, a_i)^{c_i}, (b_{j-1}, b_j)^{c_j}$ with $i, j=1,2,\ldots,k$. Because the entries of ${\bm{W}}$ and $\bm{G}$ are all i.i.d. with zero mean and unit variance, each term in the above summation equals 0 unless each ordered pair appears at least twice and at least one
pair appears at least three times. Thus, there are at most $k-1$ distinct ordered pairs. Then we can bound the above equation as
\begin{equation*}
    \mathcal{O}\left(d^{-k}\sum_{*}\left|u_{a_0}^*v_{a_k}v^*_{b_k}u_{b_0}\right|\right)\le \mathcal{O}\left(d^{-k}\sum_{*}|u_{a_0}|^2||v_{a_k}|^2+|u_{b_0}|^2|v_{b_k}|^2\right).
\end{equation*}
Here the summation $*$ is the summation of $\bm{a},\bm{b}\in [d]^{k+1}$ which satisfied the at most $k-1$ distinct ordered pairs we mentioned above. For a fixed $(a_0, a_k, b_0, b_k)$, we have that the summation is at most $\mathcal{O}(d^{k-1})$. Thus the equation can be further bounded by
\begin{equation*} \mathcal{O}\left(d^{-1}\sum_{*}|u_{a_0}|^2||v_{a_k}|^2+|u_{b_0}|^2|v_{b_k}|^2\right)=\mathcal{O}\left(d^{-1}\right)=o(1).
\end{equation*}
\end{proof}

\begin{lemma}\label{lemma: uWu=0} 
Let ${\bm{W}}$ be a ${d{\times} d}$ (complex) Gaussian i.i.d.\ matrix, and $\bm{u}$ be independent random vector over $\mathcal{S}_{\mathbb{C}}^{d-1}$. Then we have
\[\mathbb{E} \left| {\bm{u}}^*\left(\frac{1}{\sqrt{d}}{\bm{W}}\right)^k {\bm{u}}\right|^2{\le} \frac{(2k)!!}{d}k^{2k-1}=o_{d}(1)\]
for all fixed $k\in\mathbb{N}^+$.
\end{lemma}
\begin{proof}
We denote $\bm{G}=\frac{1}{\sqrt{d}}{\bm{W}}$ for notation convenience in the following proof. Notice that the function $f(v,\bm{A})=v^*\bm{A}^kv$ is rotational invariant by Lemma~\ref{lemma: rotation invariance}. Then we only need to consider the case ${\bm{u}}=\bm{e}_1$. That is, we only need to compute $\mathbb{E} \left|(\bm{G}^k)_{11}\right|^2=o_{d}(1)$. Typically, we have
\begin{equation*}
\begin{aligned}
\mathbb{E} \left|(\bm{G}^k)_{11}\right|^2&= \sum_{\begin{array}{cc}
 (i_1,i_2,\ldots,i_{k-1})\in [d]^{k-1}  \\
 (j_1,j_2,\ldots,j_{k-1})\in [d]^{k-1}     
\end{array}}\mathbb{E}\prod_{t=1}^k G_{i_{t-1}i_t}G^*_{j_{t-1}j_t}\\ 
&\stackrel{(a)}{=}\sum_{l=1}^k\sum_{\begin{array}{cc}
 (i_1,i_2,\ldots,i_{k-1})\in [d]^{k-1}  \\
 (j_1,j_2,\ldots,j_{k-1})\in [d]^{k-1}   \\ 
 \#\{1\}\cup\{\bm{i}\}\cup\{\bm{j}\}=l
\end{array}}\mathbb{E} \prod_{t=1}^k G_{i_{t-1}i_t}G^*_{j_{t-1}j_t}\\ 
&\stackrel{(b)}{=}\sum_{l=1}^k C_{d-1}^{l-1} \sum_{
 \{1\}\cup\{\bm{i}\}\cup\{\bm{j}\}=[l]}\mathbb{E} G_{1i_1} G_{i_1i_2}\cdots G_{i_{k-1}1}\cdot G_{1j_1}^* G_{j_1j_2}^*\cdots G_{j_{k-1}1}^*\\  
&\stackrel{(c)}{\le}\sum_{l=1}^kd^{k-1}\sum_{ \{1\}\cup\{\bm{i}\}\cup\{\bm{j}\}=[l]}\mathbb{E} |G_{11}|^{2k}\\ 
&\stackrel{(d)}{\le}\frac{(2k)!!}{d}\sum_{l=1}^k\sum_{ \{1\}\cup\{\bm{i}\}\cup\{\bm{j}\}=[l]}1\\
&\stackrel{(e)}{\le} \frac{(2k)!!}{d}k^{2k-1}=o_{d}(1),
\end{aligned}
\end{equation*}
where $i_0=i_k=1$. In Equality $(a)$, we denote the sets $\{\bm{i}\}$ and $\{\bm{j}\}$ as $\{i_1, i_2, \ldots, i_{k-1}\}$ and $\{j_1, j_2, \ldots, j_{k-1}\}$ respectively, and $l$ represents the cardinality of $\{1\} \cup \{\bm{i}\} \cup \{\bm{j}\}$. To derive  Equality $(a)$, we observe that if there are at least $k+1$ distinct indices, then at least one index must appear only once in the product, implying that the expectation is zero.  Equality $(b)$ is due to the symmetry of the indices, and the number of distinct index sets with the same cardinality $l$ is given by $C_{d-1}^{l-1}$ and thus can be bounded by $d^{k-1}$.  Inequality $(c)$ is because $\mathbb{E} |G_{11}|^{2k}$ is the largest among the combinations of $\mathbb{E} G_{1i_1} G_{i_1i_2}\cdots G_{i_{k-1}1}\cdot G_{1j_1}^* G_{j_1j_2}^*\cdots G_{j_{k-1}1}^*$. Furthermore, we simply bound $\prod_{t=1}^{l-1}(d-t)$ by $d^{k-1}$.  Equality $(d)$ is due to the fact that $\mathbb{E} |G^{2k}_{11}|\le\frac{(2k)!!}{d^k}$.  Inequality $(e)$ is because each element in $\{\bm{i}\}$ and $\{\bm{j}\}$ has at most $l\le k$ possible choices. Subsequently, we get the desired result. 

\end{proof}

\section{Supplementaries for Information Theoretic Lower Bound}\label{section: i-t lower bound}

In this section, we provide the proofs that are omitted in the estimation part in Section~\ref{subsec: estimation}. In addition, we provide the proof of the main theorems of the application. Meanwhile, we consider the two-side query model at the end of this section. We use bold letters ${\bm{u}}$ to denote random vectors and use standard typesetting $u$ to denote their corresponding fixed realizations in this section.

\subsection{Proof of Theorem~\ref{thm: upper bound of information}.}

We follow the proof of the symmetric case \cite{DBLP:journals/corr/abs-1804-01221} to show the lower bound. The main difference is the conditional distribution, which will be stated in Lemma~\ref{lemma: conditional likelihoods}.
Let $\mathrm{Z}_i:=\left({\bm{v}}^{(1)}, \bm{w}^{(1)}, \ldots, {\bm{v}}^{(i)}, \bm{w}^{(i)}\right)$  be the history of the algorithm after round $i$. Let $\mathbb{P}_u$ be the law of $\mathrm{Z}_{{T}+1}$ conditional on $\{{\bm{u}}=u\}$ and $\mathbb{P}_0$ be the law of $\mathrm{Z}_{{T}+1}$ conditional on $\{{\bm{u}}=0\}$, so the randomness comes from the algorithms and $\bm{G}$. For convenience, we denote the potential function $\Phi\left(\bm{V}_k ; u\right):=u^*\bm{V}_k\bm{V}_k^*u$, which is what we consider in Theorem~\ref{thm: upper bound of information}. First, we decompose our target into a series of expectations. That is, for a sequence of non-decreasing thresholds $0\le\tau_0\le\tau_1\le\ldots$, we have
\begin{equation}\label{eq: decomposition of potential}
\begin{aligned}
&\mathbb{E}_{{\bm{u}}} \mathbb{P}_{u}\left[\exists k \geq 0: \Phi\left(\bm{V}_{k+1} ; u\right)>\tau_{k+1} / d\right]\\
&\leq \sum_{k=0}^{\infty} \mathbb{E}_{{\bm{u}}} \mathbb{P}_{u}\left[\left\{\Phi\left(\bm{V}_k ; u\right) \leq \tau_k / d\right\} \cap\left\{\Phi\left(\bm{V}_{k+1} ; u\right)>\tau_{k+1} / d\right\}\right].
\end{aligned}    
\end{equation}
To deal with the terms in the summation, we present the following proposition.

\begin{prop}\label{prop:proposition 4.2 in 18}(Proposition 4.2 in \cite{DBLP:journals/corr/abs-1804-01221}) Let $\mathcal{D}$ be any distribution supported on $\mathcal{S}^{d-1}$, $0<\tau_k<\tau_{k+1}$, and  $\eta>0$. Then,
\[
\begin{aligned}
\lefteqn{ \mathbb{E}_{{\bm{u}} \sim \mathcal{D}} \mathbb{P}_{u}\left[\left\{\Phi\left(\bm{V}_k ; u\right) \leq \tau_k / d\right\} \cap\left\{\Phi\left(\bm{V}_{k+1} ; u\right)>\tau_{k+1} / d\right\}\right] } \\
& \leq  \left(\mathbb{E}_{{\bm{u}} \sim \mathcal{D}} \mathbb{E}_{\mathrm{Z}_k \sim \mathbb{P}_0}\left[\left(\frac{\mathrm{d} \mathbb{P}_{u}\left(\mathrm{Z}_k\right)}{\mathrm{d} \mathbb{P}_0\left(\mathrm{Z}_k\right)}\right)^{1+\eta} \mathbb{I}\left(\left\{\Phi\left(\bm{V}_k ; u\right)\leq \tau_k / d\right\}\right)\right]\right)^{\frac{1}{1+\eta}} \\
& \quad \times \left(\sup _{V \in \mathcal{S}^{d-1}} \mathbb{P}_{{\bm{u}} \sim \mathcal{D}}\left[\Phi\left(V ; u\right)>\tau_{k+1} / d\right]^\eta\right)^{\frac{1}{1+\eta}}.
\end{aligned}
\]
\end{prop}

Although the above proposition in \cite{DBLP:journals/corr/abs-1804-01221} is proposed in the context of symmetric matrices, it also holds for asymmetric matrices, as the proof does not rely on the symmetry property of the matrix. Taking the exact form of the potential function, the above proposition decomposes the probability of two events $\left\{u^*\bm{V}_k\bm{V}_k^*u\le \tau_k/d\right\}$ and $\left\{u^*\bm{V}_{k+1}\bm{V}_{k+1}^*u> \tau_{k+1}/d\right\}$ into an information term (i.e., the first term in the above equation)  and an entropy term (the second term). The information term captures the likelihood ratio restricted to the first event, while the entropy term captures the randomness or uncertainty associated with the second event. Next, Proposition \ref{prop: entropy term} and Proposition \ref{prop: information term} provide upper bounds for these two terms. For the entropy term, we have:
\begin{prop}\label{prop: entropy term}
(Lemma 4.3 in \cite{DBLP:journals/corr/abs-1804-01221}) For any fixed $V\in\operatorname{Stief}(d,k)$, $\tau_{k}\ge 2k$ and ${\bm{u}}\stackrel{\text{unif}}{\sim}\mathcal{S}^{d-1}$, we have 
    \begin{equation*}
       \mathbb{P}\left[{\bm{u}}^* VV^* {\bm{u}} \geq \tau_{k} / d\right] \leq \exp \left\{-\frac{1}{2}\left(\sqrt{\tau_{k}}-\sqrt{2k}\right)^2\right\}.
    \end{equation*}
\end{prop}

This lemma is the same as that in \cite{DBLP:journals/corr/abs-1804-01221} because it does not concern ${\bm{M}}$. The next proposition states an upper bound for the information term:
\begin{prop}\label{prop: information term}
 For any $\tau_k \geq 0$ and any fixed $u \in \mathcal{S}^{d-1}$, we have   
 \begin{equation*}
     \mathbb{E}_{\mathbb{P}_0}\left[\left(\frac{\mathrm{d} \mathbb{P}_{u}\left(\mathrm{Z}_k\right)}{\mathrm{d} \mathbb{P}_0\left(\mathrm{Z}_k\right)} \mathbb{I}\left(\Phi\left(\bm{V}_k ; u\right) \leq \tau_k / d\right)\right)^{1+\eta}\right] \leq \exp \left({\frac{\eta(1+\eta)}{2}} \lambda^2 \tau_k\right) .
 \end{equation*}
 In particular, taking expectation over ${\bm{u}} \stackrel{\text{unif}}{\sim} \mathcal{S}^{d-1}$, we have
 \begin{equation*}
 \mathbb{E}_{{\bm{u}}} \mathbb{E}_{\mathbb{P}_0}\left[\left(\frac{\mathrm{d} \mathbb{P}_{u}\left(\mathrm{Z}_k\right)}{\mathrm{d} \mathbb{P}_0\left(\mathrm{Z}_k\right)} \mathbb{I}\left(\Phi\left(\bm{V}_k ; u\right) \leq \tau_k / d\right)\right)^{1+\eta}\right] \leq  \exp \left({\frac{\eta(1+\eta)}{2}} \lambda^2 \tau_k\right) .   
 \end{equation*}
\end{prop}

The detailed proof is left in Section~\ref{sec: proof of information term} and here we give a brief roadmap. Notice that ${\bm{v}}^{(i)}$ is determined by the history $\mathrm{Z}_{i-1}$, so we can decompose the likelihood ratio as 
\begin{equation*}
    \frac{\mathrm{d} \mathbb{P}_{u}\left(\mathrm{Z}_i\right)}{\mathrm{d} \mathbb{P}_0\left(\mathrm{Z}_i\right)}\mathbb{I}\left(\Phi\left(\bm{V}_i ; u\right) \leq \tau_i / d\right)=\frac{\mathrm{d} \mathbb{P}_{u}\left(\mathrm{Z}_i\mid \mathrm{Z}_{i-1}\right)\mathbb{P}_{u}\left(\mathrm{Z}_{i-1}\right)}{\mathrm{d} \mathbb{P}_{0}\left(\mathrm{Z}_i\mid \mathrm{Z}_{i-1}\right)\mathbb{P}_0\left(\mathrm{Z}_{i-1}\right)}\mathbb{I}\left(\Phi\left(\bm{V}_i ; u\right) \leq \tau_i / d\right).
\end{equation*}

Because the event $\left\{\Phi\left(\bm{V}_i ; u\right) \leq \tau_i / d\right\}$ depends only on $\mathrm{Z}_{i-1}$, we can deal with the ratio of the conditional probabilities separately. We have the following proposition.

\begin{prop}\label{prop: Generic upper bound on likelihood ratios}
(Decomposition of likelihood ratios, Proposition 4.5 in \cite{DBLP:journals/corr/abs-1804-01221}). Given $u \in \mathcal{S}^{d-1}$ and $r>0$, we define the likelihood function:
\begin{equation*}
    g_i\left(\tilde{V}_i\right):=\mathbb{E}_{\mathbb{P}_0}\left[\left(\frac{\mathrm{d} \mathbb{P}_u\left(\mathrm{Z}_i \mid \mathrm{Z}_{i-1}\right)}{\mathrm{d} \mathbb{P}_0\left(\mathrm{Z}_i \mid \mathrm{Z}_{i-1}\right)}\right)^r \mid\bm{V}_i=\tilde{V}_i\right] .  
 \end{equation*}
Then for any subset $\mathcal{V}_k \subset \operatorname{Stief}(d, k)$, we have
\[
\mathbb{E}_{\mathbb{P}_0}\left[\left(\frac{\mathrm{d} \mathbb{P}_u\left(\mathrm{Z}_k \right)}{\mathrm{d} \mathbb{P}_0\left(\mathrm{Z}_k \right)}\right)^r \mathbb{I}\left(\bm{V}_k \in \mathcal{V}_k\right)\right] \leq \sup _{\tilde{V}_k \in \mathcal{V}_k} \prod_{j=1}^k g_j\left(\tilde{V}_{1: j}\right),
\]
where $\tilde{V}_{1: j}$ denotes the first $j$ columns of $\tilde{V}_k$.
\end{prop}

Similarly, the above proposition also holds for the asymmetric matrix case because the proof does not use the properties of the symmetric matrix. In this context, we still need to compute the conditional probability, and our conditional likelihoods of the queries are different from that in \cite{DBLP:journals/corr/abs-1804-01221}:

\begin{lemma}\label{lemma: conditional likelihoods}(Conditional likelihoods). Let $\bm{w}^{(i)}={\bm{M}}{\bm{v}}^{(i)}$. Under $\mathbb{P}_u$, the distribution of $\bm{w}^{(i)}$ is
\begin{equation*}
\begin{aligned}
\left(\begin{array}{c}
   \mathfrak{Re}( \bm{w}^{(i)} \mid \mathrm{Z}_{i-1})    \\
    \mathfrak{Im}(\bm{w}^{(i)} \mid \mathrm{Z}_{i-1})  
\end{array}\right)\sim \mathcal{N}\left(\left[\begin{array}{c}
  \lambda \mathfrak{Re}\left(u^{*} {\bm{v}}^{(i)}\right)  u    \\
   \lambda \mathfrak{Im}\left(u^{*} {\bm{v}}^{(i)}\right)  u 
\end{array}\right], \; \frac{1}{d}\bm{I}_{2d} \right).
\end{aligned}  
\end{equation*}
\end{lemma}
\begin{proof}[Proof]
Conditional on $\mathrm{Z}_{i-1}$ and $\left\{{\bm{u}}=u\right\}$, the randomness of $(\bm{w}^{(i)} \mid \mathrm{Z}_{i-1})$ comes from the construction of $\bm{G}$. This is first because we assume that the algorithm is deterministic by Observation~\ref{obs:deterministic}, and so ${\bm{v}}^{(i)}$ is deterministic given $\mathrm{Z}_{i-1}$. Meanwhile, vectors $\left({\bm{v}}^{(j)}\right)_{j=1,2,\dots}$ are orthogonal and thus $\bm{w}^{(i)}$ is independent of $\bm{w}^{(1)}, \bm{w}^{(2)}, \ldots, \bm{w}^{(i-1)}$ given $\mathrm{Z}_{i-1}$. Consequently, the conditional distribution of ${\bm{v}}^{(i)}$ given $\mathrm{Z}_{i-1}$ can be computed as the queries ${\bm{v}}^{(1)}, {\bm{v}}^{(2)}, \ldots, {\bm{v}}^{(i-1)}$ were fixed in advance. Then we compute the mean and variance as below.
\begin{equation*}
    \mathbb{E}\left(\bm{w}^{(i)} \mid \mathrm{Z}_{i-1}\right) =\mathbb{E}\left(\left(\bm{G}+\lambda uu^*\right){\bm{v}}^{(i)} \mid \mathrm{Z}_{i-1}\right)\stackrel{(a)}{=}\lambda uu^*{\bm{v}}^{(i)}=\lambda\left(u^{*} {\bm{v}}^{(i)}\right) u.
\end{equation*}
Equality $(a)$ is because ${\bm{v}}^{(i)}$ is determined by $\mathrm{Z}_{i-1}$. Because $\mathfrak{Re}\left({\bm{v}}^{(i)}\right)$ and $\mathfrak{Im}\left({\bm{v}}^{(i)}\right)$ are orthogonal, the Gaussian variables $\mathfrak{Re}( \bm{w}^{(i)} \mid \mathrm{Z}_{i-1})$ and $\mathfrak{Im}( \bm{w}^{(i)} \mid \mathrm{Z}_{i-1})$ are independent. Then we can compute the variance separately:
\begin{equation*}
\begin{aligned}
    \operatorname{Var}\left(\mathfrak{Re}( \bm{w}^{(i)} \mid \mathrm{Z}_{i-1})\right)&=\operatorname{Var}\left(\bm{G}\mathfrak{Re}\left({\bm{v}}^{(i)}\right) \mid \mathrm{Z}_{i-1}\right)\\
    &=\mathbb{E}\left(\bm{G}\mathfrak{Re}\left({\bm{v}}^{(i)}\right)\mathfrak{Re}\left({\bm{v}}^{(i)}\right)^*\bm{G}^* \mid \mathrm{Z}_{i-1}\right)=\frac{1}{d}\bm{I}_d.    
\end{aligned}
\end{equation*}
The same variance holds for $\mathfrak{Im}( \bm{w}^{(i)} \mid \mathrm{Z}_{i-1})$. 
\end{proof}

Based on the conditional distribution, we can bound the terms $g_i\left(\tilde{V}_i\right)$ in Proposition \ref{prop: Generic upper bound on likelihood ratios} and get the upper bound of the information term by computing the expectation of the Gaussian ratio. Combining the analysis with both the terms, we get the upper bound of the right-hand side of Proposition \ref{prop:proposition 4.2 in 18}. The following proposition can be derived from Proposition \ref{prop: entropy term} and Proposition \ref{prop: information term}.
\begin{prop}\label{prop: mixture term}
    Let $\left(\tau_k\right)_{k\ge0}$ be a nondecreasing sequence with $\tau_0=0$ for $k \geq 1, \tau_k \geq 2 k$. Then for all $\eta>0$, we have
\begin{equation}\label{eq: each term of decomposion}
\begin{aligned}
    &\mathbb{E}_{{\bm{u}}} \mathbb{P}_{u}\left[\left\{\Phi\left(\bm{V}_k ; u\right) \leq \frac{\tau_k}{d}\right\} \cap\left\{\Phi\left(\bm{V}_{k+1} ; u\right)>\frac{\tau_{k+1}}{d}\right\}\right] \\
& {\leq} \exp \left\{\frac{\eta}{2(1{+}\eta)}\left((1{+}\eta) \lambda^2 \tau_k-\left(\sqrt{\tau_{k+1}}-\sqrt{4 k {+} 4}\right)^2\right)\right\} .
\end{aligned}
\end{equation}
\end{prop}
Selecting an appropriate sequence $\left(\tau_k\right)_{k=0,1,\dots}$ in the above proposition, we can ensure that the probability is small. We then know that, given the event that $d\Phi\left(\bm{V}_k ; u\right)$ does not exceed the threshold $\tau_k$, the probability that $d\Phi\left(\bm{V}_{k+1} ; u\right)$ exceeds the threshold $\tau_{k+1}$ is small. Now, we can  bound the right hand side of Equation~(\ref{eq: decomposition of potential}):
\begin{equation*}
\begin{aligned}
&\mathbb{E}_{{\bm{u}}} \mathbb{P}_{u}\left[\exists k \geq 0: \Phi\left(\bm{V}_{k+1} ; u\right)>\tau_{k+1} / d\right] \\
&\stackrel{(\ref{eq: decomposition of potential})}{\leq}  \sum_{k=0}^{\infty} \mathbb{E}_{{\bm{u}}} \mathbb{P}_{u}\left[\left\{\Phi\left(\bm{V}_k ; u\right) \leq \tau_k / d\right\} \cap\left\{\Phi\left(\bm{V}_{k+1} ; u\right)>\tau_{k+1} / d\right\}\right]\\
&\stackrel{(\ref{eq: each term of decomposion})}{\le} \sum_{k=0}^{\infty} \exp \left\{\frac{\eta}{2(1+\eta)}\left((1+\eta) \lambda^2 \tau_k-\left(\sqrt{\tau_{k+1}}-\sqrt{4 k+4}\right)^2\right)\right\}.
\end{aligned}    
\end{equation*}

Theorem~\ref{thm: upper bound of information} now follows by first choosing the appropriate sequence $\tau_k(\delta)$ and proper $\eta$, and then verifying that the right-hand side of the above equation is below $\delta$. Setting $\eta=\lambda-1$ and $\tau_{k+1}=\lambda^4 \tau_k$, we easily ensure that $\tau_{k+1}$ grows in a linear rate,  thus it is much larger than $\sqrt{4k+4}$. Further, the exponent in the above equation will be a negative number, resulting in a rapid decrease of the right-hand side of the equation, and ensuring the convergence of the summation. In the following sections, we provide the missing proofs.

\subsection{Proof of Proposition \ref{prop: information term}}
\label{sec: proof of information term}
\begin{proof}[Proof]
Given $u \in \mathcal{S}^{d-1}$, taking $\mathcal{V}_k=\{\bm{V}_k\in\operatorname{Stief}(d,2k), \Phi\left(\bm{V}_k ; u\right) \leq \tau_k / d\}$ and $r=1+\eta$ in Proposition \ref{prop: Generic upper bound on likelihood ratios}, we know
\begin{equation}\label{eq: decomposition in proof}
    \mathbb{E}_{\mathbb{P}_0}\left[\left(\frac{\mathrm{d} \mathbb{P}_u\left(\mathrm{Z}_i \right)}{\mathrm{d} \mathbb{P}_0\left(\mathrm{Z}_i \right)}\right)^{1+\eta} \mathbb{I}\left(\bm{V}_k \in \mathcal{V}_k\right)\right] \leq \sup _{\tilde{V}_k \in \mathcal{V}_k} \prod_{i=1}^k g_i\left(\tilde{V}_{1: i}\right),
\end{equation}
and
\begin{equation}\label{eq: gvi}
\begin{aligned}
g_i\left(\tilde{V}_i\right) &=\mathbb{E}_{\mathbb{P}_0}\left[\left(\frac{\mathrm{d} \mathbb{P}_u\left(\mathrm{Z}_i \mid \mathrm{Z}_{i-1}\right)}{\mathrm{d} \mathbb{P}_0\left(\mathrm{Z}_i \mid \mathrm{Z}_{i-1}\right)}\right)^{1+\eta} \mid\bm{V}_i=\tilde{V}_i\right]\\
&=\mathbb{E}_{\mathbb{P}_0}\left[\left(\frac{\mathrm{d} \mathbb{P}_u\left(\bm{w}^{(i)} \mid \mathrm{Z}_{i-1}\right)}{\mathrm{d} \mathbb{P}_0\left(\bm{w}^{(i)} \mid \mathrm{Z}_{i-1}\right)}\right)^{1+\eta} \mid\bm{V}_i=\tilde{V}_i\right].
\end{aligned}
\end{equation}
Then we are going to bound $g_i\left(\tilde{V}_i\right)$. By Lemma~\ref{lemma: conditional likelihoods}, we have known the exact distribution 
$    \left(\begin{array}{c}
   \mathfrak{Re}( \bm{w}^{(i)} \mid \mathrm{Z}_{i-1})    \\
    \mathfrak{Im}(\bm{w}^{(i)} \mid \mathrm{Z}_{i-1})  
\end{array}\right)\sim \mathcal{N}\left(\left[\begin{array}{c}
  \lambda \mathfrak{Re}\left(u^{*} {\bm{v}}^{(i)}\right)  u    \\
   \lambda \mathfrak{Im}\left(u^{*} {\bm{v}}^{(i)}\right)  u 
\end{array}\right] , \frac{1}{d}\bm{I}_{2d} \right).$
Then by Lemma~\ref{lemma: ratio of gaussian}, we have
\begin{equation}\label{eq: computation of gi}
\begin{aligned}
    g_i\left(\tilde V_i\right) & \stackrel{(\ref{eq: gvi})}{=}\mathbb{E}_{\mathbb{P}_0}\left[\left(\frac{\mathrm{d} \mathbb{P}_u\left(\bm{w}^{(i)} \mid \mathrm{Z}_{i-1}\right)}{\mathrm{d} \mathbb{P}_0\left(\bm{w}^{(i)} \mid \mathrm{Z}_{i-1}\right)}\right)^{1+\eta} \mid\bm{V}_i=\tilde{V}_i\right]\\
    &=\exp \left(\frac{\eta(1+\eta) }{2} \left(\begin{array}{c}
   \lambda \mathfrak{Re}\left(u^{*} {\bm{v}}^{(i)}\right) u \\
   \lambda \mathfrak{Im}\left(u^{*} {\bm{v}}^{(i)}\right) u
\end{array}\right)^*\left(\frac{1}{d}\bm{I}_{2d}\right)^\dagger\left(\begin{array}{c}
   \lambda \mathfrak{Re}\left(u^{*} {\bm{v}}^{(i)}\right) u \\
   \lambda \mathfrak{Im}\left(u^{*} {\bm{v}}^{(i)}\right) u
\end{array}\right)\right) \\
& =\exp \left(\frac{\eta(1+\eta)}{2} \lambda^2 \cdot d\cdot  \left[\mathfrak{Re}^2\left(u^{*} {\bm{v}}^{(i)}\right)+\mathfrak{Im}^2\left(u^{*} {\bm{v}}^{(i)}\right)\right]\right)\\
& =\exp \left(\frac{\eta(1+\eta)}{2} \lambda^2 \cdot d\cdot|{\bm{v}}^{(i)*} u|^2\right).
\end{aligned}    
\end{equation}
Hence we get
\begin{equation*}
\begin{aligned}
\mathbb{E}_{\mathbb{P}_0}\left[\left(\frac{\mathrm{d} \mathbb{P}_u\left(\mathrm{Z}_i \right)}{\mathrm{d} \mathbb{P}_0\left(\mathrm{Z}_i \right)}\right)^{1+\eta} \mathbb{I}\left(\bm{V}_k \in \mathcal{V}_k\right)\right]\stackrel{(\ref{eq: decomposition in proof})}{\leq} & \sup _{\tilde{V}_k \in \mathcal{V}_k} \prod_{i=1}^k g_i\left(\tilde{V}_{1: i}\right)\\
\stackrel{(\ref{eq: computation of gi})}{=} & \sup _{\tilde{V}_k \in \mathcal{V}_k} \prod_{i=1}^k \exp \left(\frac{\eta(1+\eta)}{2}\lambda^2 \cdot d \cdot \left|\tilde{V}_k[i]^*u\right|^2\right) \\
= & \sup _{\tilde{V}_k \in \mathcal{V}_k} \exp \left( \frac{\eta(1+\eta)}{2} d\cdot\lambda^2 \cdot\Phi\left(\tilde{V}_k ; u\right)\right) \\
\leq & \exp \left(\frac{\eta(1+\eta)}{2} \lambda^2 \tau_k\right) ,
\end{aligned}    
\end{equation*}
where $\tilde{V}_k[i]$ denotes the $i$-th column of $\tilde{V}_k$.
\end{proof}

\subsection{Proof of Theorem~\ref{thm: upper bound of information}}
\label{sec: upper bound of information}
\begin{proof}[Proof] 
Taking $\tau_0(\delta)=64\lambda^{-2}(\lambda-1)^{-2}\left(\log\delta^{-1}+\log^{-1}\left(\lambda\right)\right)$, $\tau_{k}=\lambda^{4k} \tau_0(\delta)$ and $\eta=\lambda-1$, we have for $\delta \in(0,1 / e)$, 
\begin{equation}\label{eq: bound of tau0}
\tau_0(\delta)\ge \frac{64 }{\lambda^2(\lambda-1)^2}\left(1+\frac{1}{4 e \log \left(\lambda\right)}\right) .
\end{equation}
Then by Lemma~\ref{lemma: D.2}, we obtain
\begin{equation}\label{eq: sqrt tau_k+1}
\begin{aligned}
2 \sqrt{(4 k+4) / \tau_{k+1}} =4 \sqrt{\lambda^{-4 k}(k+1) / \tau_1} & \stackrel{(\ref{eq: bound of tau0})}{\leq} 4 \sqrt{\left(1+\frac{1}{4 e \log \left(\lambda\right)}\right) / \tau_1} \\
& = 4 \sqrt{\left(1+\frac{1}{4 e \log \left(\lambda\right)}\right) / \left(\lambda^4\tau_0\left(\delta\right)\right)} \\
& \stackrel{Lemma~\ref{lemma: D.2}}{\leq} 4 \sqrt{(\lambda-1)^2 / 64\lambda^2}=\frac{1}{2}(1-1 / \lambda) .
\end{aligned}    
\end{equation}
Subsequently, we obtain
\begin{equation}\label{eq: exp(tau)}
    \begin{aligned}
&\exp \left\{\frac{\eta}{2}\left(\lambda^2 \tau_k-\frac{\left(\sqrt{\tau_{k+1}}-\sqrt{4 k+4}\right)^2}{1+\eta}\right)\right\}\\
&=\exp \left\{\frac{\lambda-1}{2 \lambda}\left(\lambda^3 \tau_k-\left(\sqrt{\tau_{k+1}}-\sqrt{4 k+4}\right)^2\right)\right\} \\
& \le \exp \left\{\frac{\lambda-1}{2 \lambda}\left(\lambda^3 \tau_k-\tau_{k+1}+2 \sqrt{(4 k+4) / \tau_{k+1}}\cdot \tau_{k+1}\right)\right\} \\
& \stackrel{(\ref{eq: sqrt tau_k+1})}{\le} \exp \left\{\frac{\lambda-1}{2 \lambda}\left(\left(\lambda^3 \tau_k-\tau_{k+1}+\frac{1}{2}(1-1 / \lambda) \tau_{k+1}\right)\right\}\right. \\
& =\exp \left\{-\frac{\lambda-1}{2 \lambda}\left(\tau_{k+1}-\tau_{k+1} / \lambda+\frac{1}{2}(1-1 / \lambda) \tau_{k+1}\right\}\right. \\
& =\exp \left\{-\frac{(\lambda-1)^2}{4 \lambda^2} \tau_{k+1}\right\} =\exp \left\{-\frac{\tau_0(\delta) \lambda^{4 (k+1)}(\lambda-1)^2}{4 \lambda^2}\right\} .
\end{aligned}
\end{equation}
Then by Proposition~\ref{prop: mixture term} and Lemma~\ref{lemma: D.2}, we obtain
\begin{equation*}
\begin{aligned}
 &\mathbb{E}_{{\bm{u}}} \mathbb{P}_{u}\left[\exists k \geq 0: \Phi\left(\bm{V}_{k} ; u\right)>\tau_{k} / d\right] \\
 &\stackrel{(\ref{eq: decomposition of potential})}{\leq}  \sum_{k=1}^{\infty} \mathbb{E}_{{\bm{u}}} \mathbb{P}_{u}\left[\left\{\Phi\left(\bm{V}_k ; u\right) \leq \tau_k / d\right\} \cap\left\{\Phi\left(\bm{V}_{k+1} ; u\right)>\tau_{k+1} / d\right\}\right]\\
 &\stackrel{(\ref{eq: each term of decomposion})}{\le} \sum_{k=0}^{\infty} \exp \left\{\frac{\eta}{2(1+\eta)}\left((1+\eta) \lambda^2 \tau_k-\left(\sqrt{\tau_{k+1}}-\sqrt{4 k+4}\right)^2\right)\right\}\\
 &\stackrel{(\ref{eq: exp(tau)})}{\le} \sum_{k=0}^{\infty} \exp \left\{-\frac{\tau_0(\delta) \lambda^{4 (k+1)}(\lambda-1)^2}{4 \lambda^2}\right\}\\
 &= \sum_{k=0}^{\infty} \exp \left\{-16\log\delta^{-1}-\frac{16\lambda^{4 k}}{\log\left(\lambda\right)}\right\}\\
 &\stackrel{(a)}{\le} \sum_{k=0}^{\infty} \exp \left\{-16\log\delta^{-1}-64e{\log(k+1)}\right\}\\
 &\le \sum_{k=0}^{\infty} \exp \left\{-2\log\delta^{-1}-2{\log(k+1)}\right\}\\
 &= \sum_{k=0}^{\infty} \delta^{2}\frac{1}{(k+1)^2}=\frac{\delta^{2}\pi^2}{6}\le \delta.
\end{aligned}    
\end{equation*}
Inequality $(a)$ is because of Lemma~\ref{lemma: D.2}. We now complete the proof.
\end{proof}

\subsection{Proof of Theorem~\ref{thm: main theorem, to be prove}}\label{sec: proof of main result to be prove}
\begin{proof}[Proof]
    We decompose the probability as
    \begin{equation}\label{eq: P|v-hatv|}
    \begin{aligned}
        \mathbb{P}_{\operatorname{Alg}, \bm{G}, {\bm{u}}}\left(\left\|\hat{{\bm{v}}}-{\bm{v}}_1({\bm{M}}) \right\|_2\geq \frac{\sqrt{\operatorname{gap}}}{2}  \right) &\stackrel{(a)}{=}\mathbb{P}_{\operatorname{Alg}, \bm{G}, {\bm{u}}}\left(\left|\hat{\bm{v}}^* {\bm{u}} \right|_2\geq \frac{\sqrt{\operatorname{gap}}}{2} \right)+o_d(1)\\
        &\stackrel{(b)}{=}\mathbb{P}_{\operatorname{Alg}, \bm{G}, {\bm{u}}}\left(\left\|\bm{V}_{{T}+1}^{*} {\bm{u}} \right\|_2\geq \frac{\sqrt{\operatorname{gap}}}{2}\right)+o_d(1)\\
        &\le\mathbb{P}_{\operatorname{Alg}, \bm{G}, {\bm{u}}}\left(\exists k\ge 1: \left\|\bm{V}_{k}^{*} {\bm{u}} \right\|_2\geq \frac{\sqrt{\operatorname{gap}}}{2} \right)+o_d(1).
    \end{aligned}
    \end{equation}
     Equality $(a)$ is because of Corollary~\ref{corollary: reduction}.  Equality $(b)$ is because of Observation~\ref{obs: eigenvalue extend}.
    Then we apply Proposition~\ref{thm: upper bound of information} to bound the last probability. Solving the equation 
    \begin{equation*}
         \frac{64 \lambda^{4 k-4}\left(\log\delta^{-1}+\operatorname{gap}^{-1}\right)}{d\operatorname{gap}^2}=\frac{\operatorname{gap}}{4}
    \end{equation*}
    with $k={T}+1$, we obtain $\delta=\exp\left\{-\frac{1}{\operatorname{gap}}-\lambda^{-4{T}}\frac{d\operatorname{gap}^3}{256}\right\}.$ Then by Proposition \ref{thm: upper bound of information} we get
    \begin{equation}\label{eq: first term in final}
        \mathbb{P}_{\operatorname{Alg}, \bm{G}, {\bm{u}}}\left(\exists k\ge 1: \left\|\bm{V}_{k}^{*} {\bm{u}} \right\|_2\geq \frac{\sqrt{\operatorname{gap}}}{2} \right)\le \delta=\exp\left\{-\frac{1}{\operatorname{gap}}-\lambda^{-4{T}}\frac{d\operatorname{gap}^3}{256}\right\}.
    \end{equation}
    Combining with Equations~(\ref{eq: P|v-hatv|}), we obtain
    \begin{equation*}
        \mathbb{P}_{\operatorname{Alg}, \bm{G}, {\bm{u}}}\left(\left\|\hat{{\bm{v}}}-{\bm{v}}_1({\bm{M}}) \right\|_2\geq \frac{\sqrt{\operatorname{gap}}}{2}  \right) {\le} \exp\left\{-\frac{1}{\operatorname{gap}}-\lambda^{-4{T}}\frac{d\operatorname{gap}^3}{256}\right\}+o_{d}(1).
    \end{equation*}
\end{proof}

\subsection{Results for the Two-side Model}\label{sec: two side situation}

In this section, we give our results on the algorithm that accesses both ${\bm{M}}{\bm{v}}$ and ${\bm{M}}^*{\bm{v}}$ for the input vector ${\bm{v}}$ in each iteration, which is common in the asymmetric case. We present roadmaps and results similar to those in Section~\ref{section: application}. We first introduce the model and the random matrices that we consider.

\paragraph{Two-side Query Models} Here we still consider the query model, whose framework is similar to that in Section~\ref{section: application}. The difference is that $\operatorname{Alg}$ receives vector triples 
\begin{equation*}
    \left\{({\bm{v}}^{(t)}, \bm{w}^{(t)}, \bm{z}^{(t)})\colon  {\bm{v}}^{(t)}, \bm{w}^{(t)} = {\bm{M}}{\bm{v}}^{(t)}, \bm{z}^{(t)} = {\bm{M}}^*{\bm{v}}^{(t)}\right\}
\end{equation*}
at iteration $t$. We denote by $\tilde{\mathrm{Z}}_i:=\left({\bm{v}}^{(1)}, \bm{w}^{(1)}, \bm{z}^{(1)}, \ldots, {\bm{v}}^{(i)}, \bm{w}^{(i)}, \bm{z}^{(i)}\right)$  the history of the algorithm after iteration $i$. As we have mentioned in Section~\ref{section: application}, the original random matrix $\bm{G}+\lambda{\bm{u}}{\bm{u}}^*$ can be easily solved under this model. We will consider the asymmetric perturbation here.

\paragraph{Random matrix} The random matrix is defined as $\bm{G}+\lambda{\bm{u}}_l{\bm{u}}_r^*$, where $\bm{G}$, ${\bm{u}}_l$ and ${\bm{u}}_r$ are independent, $\bm{G}\sim \text{GinOE}(d)$, ${\bm{u}}_l, {\bm{u}}_r \stackrel{\text { unif }}{\sim} \mathcal{S}^{d-1}$, and $\lambda>1$ is the level of perturbation to be chosen. 

Our main purpose here is to detect the perturbation, i.e., we are going to find unit vectors $v_l$ and $v_r$ such that $|{\bm{u}}_l^*v_l|$ and $|{\bm{u}}_r^*v_r|$ are lower bounded. Following the similar roadmap of the one-side case, we provide a lower bound as below.

\begin{thm}\label{thm: main theorem, two side case} Let ${\bm{M}}=\bm{G}+\lambda {\bm{u}}_l{\bm{u}}_r^*$ where $\lambda>1, \bm{G}\sim \operatorname{GinOE}(d), {\bm{u}}_l, {\bm{u}}_r\stackrel{\text{unif}}{\sim}\mathcal{S}^{d-1}$, and $\bm{G}, {\bm{u}}_l, {\bm{u}}_r$ are independent. Then for any two-side query algorithm $\operatorname{Alg}$ if making ${T}\le \frac{\log d}{5(\lambda-1)}$ queries, with probability at least $1-o_{d}(1)$, $\operatorname{Alg}$ cannot identify unit vectors $\hat{\bm{v}}_l, \hat{\bm{v}}_r\in \mathbb{R}^d$ for which $\left|\langle\hat{\bm{v}}_l, {\bm{u}}_l\rangle\right|^2+\left|\langle\hat{\bm{v}}_r, {\bm{u}}_r\rangle\right|^2\le \frac{\lambda-1}{4}$.
\end{thm}

Here we provide the proof roadmap of Theorem~\ref{thm: main theorem, two side case}, which is similar to the reduction part of the one-side case. We first provide our observations on the queries in the two-side case. There are a few differences between the observations here and those in Section~\ref{section: application}.

\begin{obs} \label{abs:twoside2} Assume that the first $k$-queries are orthonormal for all $k\in\mathbb{N}^+$, namely, 
\begin{equation*}
    \bm{V}_k:=\left[{\bm{v}}^{(1)}\mid {\bm{v}}^{(2)}\mid \ldots \mid {\bm{v}}^{(k)}\right] \in \operatorname{Stief}(d, k).
\end{equation*}
\end{obs}
The first observation is valid for the same reason that we can always reconstruct the $k$-queries ${\bm{v}}^{(1)}, \ldots, {\bm{v}}^{(k)}$ from an associated orthonormal sequence obtained via the Gram-Schmidt procedure.

\begin{obs} \label{abs:twoside3}  Assume that the model receives 
\begin{equation*}
    \tilde{\bm{w}}^{(t)} = (\bm{I}-\bm{V}_{t-1} \bm{V}_{t-1}^{*}){\bm{M}}{\bm{v}}^{(t)}, \tilde{\bm{z}}^{(t)} = (\bm{I}-\bm{V}_{t-1} \bm{V}_{t-1}^{*}){\bm{M}}^*{\bm{v}}^{(t)}
\end{equation*}
at iteration $t$. Thus, the new vector is orthonormal to $\bm{V}_{t-1}$.
\end{obs}
The second observation is based on the fact that we can reconstruct the original $\bm{w}^{(k)}$ and $\bm{z}^{(k)}$ from $(\bm{I}-\bm{V}_{k-1}\bm{V}_{k-1}^*){\bm{M}}{\bm{v}}^{(k)}$ and $(\bm{I}-\bm{V}_{k-1}\bm{V}_{k-1}^*){\bm{M}}^*{\bm{v}}^{(k)}$. As we have mentioned in Section~\ref{section: application}, this observation is not necessary for the case where the oracle receives ${\bm{M}}{\bm{v}}$ only. But the observation is necessary here, otherwise the conditional distribution $(\mathrm{Z}_i \mid \mathrm{Z}_{i-1})$ will not be independent for different $i$.

The third and the fourth observation are similar with Observation~\ref{obs:deterministic} and Observation~\ref{obs: eigenvalue extend}, as well as the reason.
\begin{obs} We assume that for all $k \in[T+1]$, the query ${\bm{v}}^{(k+1)}$ is deterministic given the previous query-observation pairs $\left({\bm{v}}^{(i)}, \bm{w}^{(i)}, \bm{z}^{(i)}\right)_{1 \leq i \leq k}$.
\end{obs}

\begin{obs}\label{obs: eigenvalue extend, two-side} We assume without loss of generality that $\operatorname{Alg}$ makes two extra vectors ${\bm{v}}^{(T+1)}, {\bm{v}}^{(T+2)}$ after having outputted $\hat{
{\bm{v}}}$, and that $\lambda_{1}\left(\hat{{\bm{v}}}^* uu^* \hat{{\bm{v}}}\right) \leq \lambda_{1}\left(\bm{V}_{T+1}^{*} uu^{*} \bm{V}_{T+1}\right)$ for an arbitrary $u\in\mathbb{R}^d$.
\end{obs}

As a direct result of Observation~\ref{obs: eigenvalue extend, two-side}, we know that $u^*\bm{V}_k\bm{V}_k^*u$ is an upper bound for $u^*\hat{\bm{v}}\hat{\bm{v}}^*u$ for any $u$. Then we remain to lower bound $u_l^*\bm{V}_k\bm{V}_k^*u_l$ and $u_r^*\bm{V}_k\bm{V}_k^*u_r$. Our roadmap is the same as the above sections in this appendix, i.e., first decomposing the probability of overlapping using Equation~(\ref{eq: decomposition of potential}) and then applying the data-processing inequality Proposition~\ref{prop:proposition 4.2 in 18} and decomposition Proposition~\ref{prop: Generic upper bound on likelihood ratios}. In detail, we also denote ${\bm{u}}:=({\bm{u}}_l, {\bm{u}}_r), u:=(u_l, u_r), \Phi\left(\bm{V}_k ; u\right):={u_r^*\bm{V}_k\bm{V}_k^*u_r+u_l^*\bm{V}_k\bm{V}_k^*u_l}$ and consider the non-decreasing sequence $0\le\tau_0\le\tau_1\le\ldots$. Then we obtain

\begin{equation}\label{eq: complete decomposion, two-side}
\begin{aligned}
& \mathbb{E}_{{\bm{u}}} \mathbb{P}_{u}\left[\exists k \geq 0: \Phi\left(\bm{V}_{k+1} ; u\right)>\tau_{k+1} / d\right]\\
&\leq  \sum_{k=0}^{\infty} \mathbb{E}_{{\bm{u}}} \mathbb{P}_{u}\left[\left\{\Phi\left(\bm{V}_k ; u\right) \leq \tau_k / d\right\} \cap\left\{\Phi\left(\bm{V}_{k+1} ; u\right)>\tau_{k+1} / d\right\}\right]\\
&\stackrel{(a)}{\leq} \sum_{k=0}^{\infty}\left(\mathbb{E}_{{\bm{u}}} \mathbb{E}_{ \mathbb{P}_0}\left[\left(\frac{\mathrm{d} \mathbb{P}_{u}\left(\mathrm{Z}_k\right)}{\mathrm{d} \mathbb{P}_0\left(\mathrm{Z}_k\right)}\right)^{1+\eta} \mathbb{I}\left(\left\{\Phi\left(\bm{V}_k ; u\right)\leq \tau_k / d\right\}\right)\right]\right)^{\frac{1}{1+\eta}}\\
&\ \times \left(\sup _{V \in \mathcal{S}} \mathbb{P}_{{\bm{u}}}\left[\Phi\left(V ; u\right)>\tau_{k+1} / d\right]^\eta\right)^{\frac{1}{1+\eta}}\\
&\stackrel{(b)}{\leq} \sum_{k=0}^{\infty}\mathbb{E}_{{\bm{u}}}\left(\sup _{\Phi\left(\bm{V}_k ; u\right)\leq \tau_k / d}\prod_{i=1}^k\mathbb{E}_{\mathbb{P}_0}\left[\left(\frac{\mathrm{d} \mathbb{P}_{u}\left(\mathrm{Z}_i \mid \mathrm{Z}_{i-1}\right)}{\mathrm{d} \mathbb{P}_0\left(\mathrm{Z}_i \mid \mathrm{Z}_{i-1}\right)}\right)^{1+\eta} \mid\bm{V}_i=\tilde{V}_{1:i}\right]\right)\\
&\ \times \left(\sup _{V \in \mathcal{S}^{d-1}} \mathbb{P}_{{\bm{u}}}\left[\Phi\left(V ; u\right)>\tau_{k+1} / d\right]^\eta\right)^{\frac{1}{1+\eta}}.
\end{aligned}    
\end{equation}
We apply Proposition~\ref{prop:proposition 4.2 in 18} in Inequality $(a)$ and Proposition~\ref{prop: Generic upper bound on likelihood ratios} in Inequality $(b)$. Then we also need to deal with the information term and the entropy term. For the entropy term, we have
\begin{lemma}
For any fixed $V\in\operatorname{Stief}_\mathbb{C}(d,k)$, $\tau_{k}\ge 2k$ and independent ${\bm{u}}_l,{\bm{u}}_r\stackrel{\text{unif}}{\sim}\mathcal{S}^{d-1}$, we have 
    \begin{equation}\label{eq: entropy term, C}
       \mathbb{P}\left[{\bm{u}} _l^* VV^* {\bm{u}}_l+{\bm{u}} _r^* VV^* {\bm{u}}_r  \geq \tau_{k} / d\right] \leq \exp \left\{-\frac{1}{2}\left(\sqrt{\tau_{k}}-\sqrt{2k}\right)^2\right\}.
    \end{equation}
\end{lemma}
This lemma is similar with the entropy term in \cite{DBLP:journals/corr/abs-1804-01221} and can be directly derived based on the fact that $({\bm{u}}_l, {\bm{u}}_r)/\sqrt{2}\sim\stackrel{\text{unif}}{\sim}\mathcal{S}^{2d-1}$
To deal with the information term, we need to consider the conditional likelihoods $\mathrm{Z}_i \mid \mathrm{Z}_{i-1}$. Specifically, we have
\begin{equation*}
\begin{aligned}
\frac{\mathrm{d} \mathbb{P}_u\left(\mathrm{Z}_i \mid \mathrm{Z}_{i-1}\right)}{\mathrm{d} \mathbb{P}_0\left(\mathrm{Z}_i \mid \mathrm{Z}_{i-1}\right)}&\stackrel{(a)}{=}\frac{\mathrm{d} \mathbb{P}_u\left(\tilde{\bm{w}}^{(i)}, \tilde{\bm{z}}^{(i)} \mid \mathrm{Z}_{i-1}\right)}{\mathrm{d} \mathbb{P}_0\left(\tilde{\bm{w}}^{(i)}, \tilde{\bm{z}}^{(i)} \mid \mathrm{Z}_{i-1}\right)}=\frac{\mathrm{d} \mathbb{P}_u\left(\tilde{\bm{w}}^{(i)}+\tilde{\bm{z}}^{(i)}, \tilde{\bm{w}}^{(i)}-\tilde{\bm{z}}^{(i)} \mid \mathrm{Z}_{i-1}\right)}{\mathrm{d} \mathbb{P}_0\left(\tilde{\bm{w}}^{(i)}+\tilde{\bm{z}}^{(i)}, \tilde{\bm{w}}^{(i)}-\tilde{\bm{z}}^{(i)} \mid \mathrm{Z}_{i-1}\right)}\\
&\stackrel{(b)}{=}\frac{\mathrm{d} \mathbb{P}_u\left(\tilde{\bm{w}}^{(i)}+\tilde{\bm{z}}^{(i)}\mid \mathrm{Z}_{i-1}\right)}{\mathrm{d} \mathbb{P}_0\left(\tilde{\bm{w}}^{(i)}+\tilde{\bm{z}}^{(i)}\mid \mathrm{Z}_{i-1}\right)}\cdot\frac{\mathrm{d} \mathbb{P}_u\left(\tilde{\bm{w}}^{(i)}-\tilde{\bm{z}}^{(i)} \mid \mathrm{Z}_{i-1}\right)}{\mathrm{d} \mathbb{P}_0\left(\tilde{\bm{w}}^{(i)}-\tilde{\bm{z}}^{(i)} \mid \mathrm{Z}_{i-1}\right)}
\end{aligned}
\end{equation*}
Equality $(a)$ is because we assume the algorithm to be deterministic, and then the randomness comes from $\bm{G}$.  Equality $(b)$ is because vectors 
\begin{equation*}
\begin{aligned}
\tilde{\bm{w}}^{(i)}+\tilde{\bm{z}}^{(i)}=(\bm{I}-\bm{V}_{k-1}\bm{V}_{k-1}^*)(\bm{G}^*+\bm{G}+\lambda{\bm{u}}_l{\bm{u}}_r^*+\lambda{\bm{u}}_r{\bm{u}}_l^*){\bm{v}}^{(k)},\\ 
\tilde{\bm{w}}^{(i)}-\tilde{\bm{z}}^{(i)}=(\bm{I}-\bm{V}_{k-1}\bm{V}_{k-1}^*)(\bm{G}-\bm{G}^*+\lambda{\bm{u}}_l{\bm{u}}_r^*-\lambda{\bm{u}}_r{\bm{u}}_l^*){\bm{v}}^{(k)}
\end{aligned}    
\end{equation*}
are independent given $\mathrm{Z}_{i-1}$ and $\{{\bm{u}}_l=u_l, {\bm{u}}_r=u_r\}$. Furthermore, we note that $\bm{w}^{(i)}\pm\bm{z}^{(i)}$ are conditionally independent of $\bm{w}^{(1)}\pm\bm{z}^{(1)}, \bm{w}^{(2)}\pm\bm{z}^{(2)}, \ldots, \bm{w}^{(i-1)}\pm\bm{z}^{(i-1)}$ given ${\bm{v}}^{(1)}, \bm{w}^{(1)}, \bm{z}^{(1)}, {\bm{v}}^{(2)}, \ldots, {\bm{v}}^{(i-1)}, \bm{w}^{(i-1)}, \bm{z}^{(i-1)}, {\bm{v}}^{(i)}$, because the vectors $\left({\bm{v}}^{(j)}\right)_{j=1,2,\dots}$ are orthogonal. Consequently, the conditional distribution of ${\bm{v}}^{(i)}$ given $\mathrm{Z}_{i-1}$ can be computed as the queries ${\bm{v}}^{(1)}, {\bm{v}}^{(2)}, \ldots, {\bm{v}}^{(i-1)}$ were fixed in advance.

\begin{lemma}\label{lemma: conditional likelihoods for asymmetric}(Conditional likelihoods). Let $\mathrm{P}_i:=\bm{I}-\bm{V}_i \bm{V}_i^{*}$ denote the orthogonal projection onto the orthogonal complement of $\operatorname{span}\left({\bm{v}}^{(1)}, \ldots, {\bm{v}}^{(i)}\right)$. Under $\mathbb{P}_u$ \emph{(the joint law of ${\bm{M}}$ and $\mathrm{Z}_T$ on $\{{\bm{u}}_l=u_l, {\bm{u}}_r=u_r\}$)}, then $\bm{w}^{(i)}+\bm{z}^{(i)}$ and $\bm{w}^{(i)}-\bm{z}^{(i)}$ are independent and there distributions are
\begin{equation*}
\begin{aligned}
\bm{w}^{(i)}\pm\bm{z}^{(i)} \mid \mathrm{Z}_{i-1} \sim & \mathcal{N}\left(\lambda \mathrm{P}_{i-1}(u_l u_r^*\pm u_r u_l^*){\bm{v}}^{(k)}, \frac{2}{d}\mathrm{P}_{i-1}\left(\bm{I}\pm{\bm{v}}^{(i)}\left({\bm{v}}^{(i)}\right)^*\right)\mathrm{P}_{i-1}\right).
\end{aligned}  
\end{equation*}
\end{lemma}
\begin{proof}[Proof]
For the expectation is direct given $\mathrm{Z}_T$ on $\{{\bm{u}}_l=u_l, {\bm{u}}_r=u_r\}$, we have 
\begin{equation*}
    \mathbb{E}\left(\bm{w}^{(i)}\pm\bm{z}^{(i)}\right) =\mathbb{E}\mathrm{P}_{i-1}(\bm{G}^*\pm\bm{G}+\lambda u_l u_r^*\pm\lambda u_r u_l^*){\bm{v}}^{(k)}=\lambda \mathrm{P}_{i-1}(u_l u_r^*\pm u_r u_l^*){\bm{v}}^{(k)}.
\end{equation*}
For the variance, we have
\begin{equation*}
\begin{aligned}
    \operatorname{Var}\left(
      \bm{w}^{(i)}\pm\bm{z}^{(i)} \mid \mathrm{Z}_{i-1}\right)
      &=\frac{1}{d}
   \mathbb{E}_{\bm{G}}\left(
      \mathrm{P}_{i-1}(\bm{G}\pm\bm{G}^*){\bm{v}}^{(i)} {\bm{v}}^{(i)*}
      (\bm{G}\pm\bm{G}^*) \mathrm{P}_{i-1}  \right)\\
      &=\frac{2}{d}\mathrm{P}_{i-1}\left(\bm{I}\pm{\bm{v}}^{(i)}\left({\bm{v}}^{(i)}\right)^*\right)\mathrm{P}_{i-1}.
\end{aligned}
\end{equation*}
\end{proof}
We then deal with the information term in Equation~(\ref{eq: complete decomposion, two-side}). By Lemma~\ref{lemma: ratio of gaussian} and Lemma~\ref{lemma: conditional likelihoods for asymmetric}, let $\Sigma_{i}=\mathrm{P}_{i-1}\left(\bm{I}+{\bm{v}}^{(i)}\left({\bm{v}}^{(i)}\right)^*\right)\mathrm{P}_{i-1}$ and $\tilde\Sigma_{i}=\mathrm{P}_{i-1}\left(\bm{I}-{\bm{v}}^{(i)}\left({\bm{v}}^{(i)}\right)^*\right)\mathrm{P}_{i-1}$, we obtain that
\begin{equation}\label{eq: information term, C}
\begin{aligned}
    &\mathbb{E}_{\mathbb{P}_0}\left[\left(\frac{\mathrm{d} \mathbb{P}_u\left(\mathrm{Z}_i \mid \mathrm{Z}_{i-1}\right)}{\mathrm{d} \mathbb{P}_0\left(\mathrm{Z}_i \mid \mathrm{Z}_{i-1}\right)}\right)^{1+\eta} \mid\bm{V}_i=\tilde{V}_{1:i}\right]\\
    &=\mathbb{E}_{\mathbb{P}_0}\left[\left(\frac{\mathrm{d} \mathbb{P}_u\left(\tilde{\bm{w}}^{(i)}+\tilde{\bm{z}}^{(i)}\mid \mathrm{Z}_{i-1}\right)}{\mathrm{d} \mathbb{P}_0\left(\tilde{\bm{w}}^{(i)}+\tilde{\bm{z}}^{(i)}\mid \mathrm{Z}_{i-1}\right)}\cdot\frac{\mathrm{d} \mathbb{P}_u\left(\tilde{\bm{w}}^{(i)}-\tilde{\bm{z}}^{(i)} \mid \mathrm{Z}_{i-1}\right)}{\mathrm{d} \mathbb{P}_0\left(\tilde{\bm{w}}^{(i)}-\tilde{\bm{z}}^{(i)} \mid \mathrm{Z}_{i-1}\right)}\right)^{1+\eta} \mid\bm{V}_i=\tilde{V}_{1:i}\right]\\
    &=\exp\left\{\frac{\eta(1+\eta)}{2}\cdot \lambda^2 {\bm{v}}^{(i)*}(u_l u_r^*+ u_r u_l^*)\mathrm{P}_{i-1} \left(\frac{2}{d}\Sigma_i\right)^{\dagger}\mathrm{P}_{i-1}(u_l u_r^*+ u_r u_l^*){\bm{v}}^{(i)}\right\}\\
    &\times\exp\left\{\frac{\eta(1+\eta)}{2}\cdot \lambda^2 {\bm{v}}^{(i)*}(u_r u_l^*-u_l u_r^*)\mathrm{P}_{i-1} \left(\frac{2}{d}\tilde\Sigma_i\right)^{\dagger}\mathrm{P}_{i-1}(u_l u_r^*- u_r u_l^*){\bm{v}}^{(i)}\right\}\\
    &\stackrel{(a)}{\le}\exp\left\{\frac{\eta(1+\eta)}{4}\cdot d\lambda^2\left(\left\|(u_l u_r^*+ u_r u_l^*){\bm{v}}^{(i)}\right\|_2^2+\left\|(u_l u_r^*- u_r u_l^*){\bm{v}}^{(i)}\right\|_2^2\right)\right\}\\
    &=\exp\left\{\frac{d\lambda^2\eta(1+\eta)}{2}\cdot\left(|u_l^{*}{\bm{v}}^{(i)}|^2+|u_r^{*}{\bm{v}}^{(i)}|^2\right)\right\}.
\end{aligned}
\end{equation}
Equality.$(a)$ is because $\mathrm{P}_{i-1} \left(\frac{2}{d}\Sigma_i\right)^{\dagger}\mathrm{P}_{i-1}\le \frac{d}{2}\bm{I}$ and $\mathrm{P}_{i-1} \left(\frac{2}{d}\tilde\Sigma_i\right)^{\dagger}\mathrm{P}_{i-1}\le \frac{d}{2}\bm{I}$. Now we have obtained the distribution and inequalities that we need in Equation~(\ref{eq: complete decomposion, two-side}). Combining the propositions above, we get the probability of overlapping:
\begin{equation*}
\begin{aligned}
& \mathbb{E}_{{\bm{u}}} \mathbb{P}_{u}\left[\exists k \geq 0: \Phi\left(\bm{V}_{k+1} ; u\right)>\tau_{k+1} / d\right] \\
&\stackrel{(\ref{eq: complete decomposion, two-side})}{\le} \sum_{k=0}^{\infty}\mathbb{E}_{{\bm{u}}}\left(\sup _{\Phi\left(\bm{V}_k ; u\right)\leq \tau_k / d}\prod_{i=1}^k\mathbb{E}_{\mathbb{P}_0}\left[\left(\frac{\mathrm{d} \mathbb{P}_u\left(\mathrm{Z}_i \mid \mathrm{Z}_{i-1}\right)}{\mathrm{d} \mathbb{P}_0\left(\mathrm{Z}_i \mid \mathrm{Z}_{i-1}\right)}\right)^{1+\eta} \mid\bm{V}_i=\tilde{V}_{1:i}\right]\right)\\
&\ \times \left(\sup _{V \in \mathcal{S}^{d-1}} \mathbb{P}_{{\bm{u}}}\left[\Phi\left(V ; u\right)>\tau_{k+1} / d\right]^\eta\right)^{\frac{1}{1+\eta}}\\
&\stackrel{(\ref{eq: entropy term, C})(\ref{eq: information term, C})}{\le}\sum_{k=0}^{\infty}\mathbb{E}_{{\bm{u}}}\left(\sup _{\Phi\left(\bm{V}_k ; u\right)\leq \tau_k / d}\prod_{i=1}^k\exp\left\{\frac{d\lambda^2\eta(1+\eta)}{2}\cdot\left(|u_r^{*}{\bm{v}}^{(i)}|^2+|u_l^{*}{\bm{v}}^{(i)}|^2\right)\right\}\right)\\
&\ \times\exp \left\{-\frac{\eta}{2(1+\eta)}\left(\sqrt{\tau_{k+1}}-\sqrt{2k+2}\right)^2\right\}\\
&=\sum_{k=0}^{\infty}\mathbb{E}_{{\bm{u}}}\left(\sup _{\Phi\left(\bm{V}_k ; u\right)\leq \tau_k / d}\exp\left\{\frac{d\lambda^2\eta(1+\eta)}{2}\cdot\Phi\left(\bm{V}_k ; u\right)\right\}\right)\\
&\times\exp \left\{-\frac{\eta}{2(1+\eta)}\left(\sqrt{\tau_{k+1}}-\sqrt{2k+2}\right)^2\right\}\\
&\le\sum_{k=0}^{\infty}\mathbb{E}_{{\bm{u}}}\exp\left\{\frac{\lambda^2\eta(1+\eta)\tau_k}{2}\right\}\cdot\exp \left\{-\frac{\eta}{2(1+\eta)}\left(\sqrt{\tau_{k+1}}-\sqrt{k+1}\right)^2\right\}\\
&=\sum_{k=0}^{\infty}\exp \left\{\frac{\lambda^2\eta(1+\eta)\tau_k}{2}-\frac{\eta}{2(1+\eta)}\left(\sqrt{\tau_{k+1}}-\sqrt{k+1}\right)^2\right\}.
\end{aligned}    
\end{equation*}
By selecting appropriate parameters as done in Section~\ref{sec: upper bound of information}, the above probability will be lower than a constant $\delta$. Then we obtain the complexity lower bound of the overlapping problem $\left\{\Phi\left(\bm{V}_{k+1} ; u\right)>\tau_{k+1} / d\right\}$. In detail, we choose the same parameters: $\tau_0(\delta)=64\lambda^{-2}(\lambda-1)^{-2}\left(\log\delta^{-1}+\log^{-1}\left(\lambda\right)\right)$, $\tau_{k}=\lambda^{4k} \tau_0(\delta)$ and $\eta=\lambda-1$. Then the above expectation of probability can be bounded by a $\delta \in(0,1 / e)$(the computation is the same as in Section~\ref{sec: upper bound of information}).

Choosing $\delta=\exp\left\{-\frac{1}{\operatorname{gap}}-\lambda^{-4{T}}\frac{d\operatorname{gap}^3}{256}\right\}$, we complete the proof of Theorem~\ref{thm: main theorem, two side case} by following the same calculation as that in Section~\ref{sec: proof of main result to be prove}.

\section{Summary}\label{section: summary}
In this work, we have studied i.i.d.\ matrices with random perturbations in the aspects of eigenvalue outliers, the ESD, and the eigenvectors. Our work has extended the results of \cite{tao2014outliers} from deterministic to random perturbations. We have also developed tools such as the asymmetric and complex variant of the Hanson-Wright inequality along the proof, which makes the concentration of eigenvalue outliers possible. As an application of the results on the random matrix, we have presented a tight query complexity lower bound for approximating the eigenvector of an asymmetric diagonalizable matrix, which is a fundamental issue in numerical linear algebra. We have also used and extended the information-theoretic tools developed by \cite{DBLP:journals/corr/abs-1804-01221} from symmetric matrices to asymmetric matrices.

\appendix
\section{Technical Lemmas}
In this section, we provide the technical lemmas that we have mentioned in the paper for completeness.

\subsection{Variants of the Hanson-Wright Inequality}
\begin{lemma}\label{lemma: HW}(Stiefel Hanson-Wright, Proposition 6.4 in \cite{DBLP:journals/corr/abs-1804-01221}) For a  symmetric matrix $\boldsymbol{A}\in\mathbb{R}^{d\times d}$, it holds that
    \begin{equation*}
        \mathbb{P}_{\boldsymbol{U} \sim \operatorname{Stief}(d, r)}\left(\left\| \boldsymbol{U}^*\boldsymbol{A} \boldsymbol{U}-\bm{I}_r\cdot\operatorname{tr}(\boldsymbol{A}) / d \right\|_2 > \frac{8(t\|\bm{A}\|_2+\sqrt{t}\|\bm{A}\|_{F})}{d(1-2\sqrt{t/d})}\right) \leq 3e^{2.2r-t}.
    \end{equation*}
\end{lemma}
As the first corollary, we extend the above lemma to the case where $\bm{A}$ can be complex.
\begin{lemma}\label{lemma: Stiefel Hanson-Wright, complex matrix} (Stiefel Hanson-Wright, complex matrix) For a Hermitian matrix $\boldsymbol{A}\in\mathbb{C}^{d\times d}$, it holds that
    \begin{equation*}
        \mathbb{P}_{\boldsymbol{U} \sim \operatorname{Stief}(d, r)}\left(\left\| \boldsymbol{U}^*\boldsymbol{A} \boldsymbol{U}-\bm{I}_r\cdot\operatorname{tr}(\boldsymbol{A}) / d \right\|_2 > \frac{16(t\|\bm{A}\|_2+\sqrt{t}\|\bm{A}\|_{F})}{d(1-2\sqrt{t/d})}\right) \leq 6e^{2.2r-t}.
    \end{equation*}
\end{lemma}
\begin{proof}[Proof]
We decompose $\bm{A}$ into $\mathfrak{Re}(\bm{A})$ and $\mathfrak{Im}(\bm{A})$ and apply the above Lemma~\ref{lemma: HW}. In detail, we have
    \begin{equation*}
        \begin{aligned}
        &\mathbb{P}_{\boldsymbol{U} \sim \operatorname{Stief}(d, r)}\left(\left\| \boldsymbol{U}^*\boldsymbol{A} \boldsymbol{U}-\bm{I}_r\cdot\operatorname{tr}(\boldsymbol{A}) / d \right\|_2 > \frac{16(t\|\bm{A}\|_2+\sqrt{t}\|\bm{A}\|_{F})}{d(1-2\sqrt{t/d})}\right)\\
        \stackrel{(a)}{\le}&\mathbb{P}_{\boldsymbol{U}}\left(\left\| \boldsymbol{U}^*\boldsymbol{A} \boldsymbol{U}-\bm{I}_r\cdot\operatorname{tr}(\boldsymbol{A}) / d \right\|_2> \frac{8(t\|\mathfrak{Re}(\bm{A})\|_2+\sqrt{t}\|\mathfrak{Re}(\bm{A})\|_{F})}{d(1-2\sqrt{t/d})} \right.\\
        &\left. +\frac{8(t\|\mathfrak{Im}(\bm{A})\|_2+\sqrt{t}\|\mathfrak{Im}(\bm{A})\|_{F})}{d(1-2\sqrt{t/d})} \right)\\
        {\le}&\mathbb{P}_{\boldsymbol{U}}\left(\left\| \boldsymbol{U}^*\mathfrak{Re}(\bm{A}) \boldsymbol{U}-\bm{I}_r\cdot\frac{\operatorname{tr}(\boldsymbol{\mathfrak{Re}(\bm{A})})}{d} \right\|_2+\left\| \boldsymbol{U}^*\mathfrak{Im}(\bm{A}) \boldsymbol{U}-\bm{I}_r\cdot\frac{\operatorname{tr}(\boldsymbol{\mathfrak{Im}(\bm{A})})}{d} \right\|_2 \right.\\
        &\left. >  \frac{8(t\|\mathfrak{Re}(\bm{A})\|_2+\sqrt{t}\|\mathfrak{Re}(\bm{A})\|_{F})}{d(1-2\sqrt{t/d})}+\frac{8(t\|\mathfrak{Im}(\bm{A})\|_2+\sqrt{t}\|\mathfrak{Im}(\bm{A})\|_{F})}{d(1-2\sqrt{t/d})} \right)\\
        \leq& 6e^{2.2r-t}.
        \end{aligned}
    \end{equation*}
    The inequality $(a)$ is because $\|\bm{A}\|_*\ge\max\{\|\mathfrak{Re}(\bm{A})\|_*, \|\mathfrak{Im}(\bm{A})\|_*\}$ for $*\in\{2, F\}$. The last inequality is because one of the events
    \begin{equation*}
        \left\| \boldsymbol{U}^*\mathfrak{Re}(\bm{A}) \boldsymbol{U}-\bm{I}_r\cdot\frac{\operatorname{tr}(\boldsymbol{\mathfrak{Re}(\bm{A})})}{d} \right\|_2>\frac{8(t\|\mathfrak{Re}(\bm{A})\|_2+\sqrt{t}\|\mathfrak{Re}(\bm{A})\|_{F})}{d(1-2\sqrt{t/d})}
    \end{equation*}
    and
    \begin{equation*}
        \left\| \boldsymbol{U}^*\mathfrak{Im}(\bm{A}) \boldsymbol{U}-\bm{I}_r\cdot\frac{\operatorname{tr}(\boldsymbol{\mathfrak{Im}(\bm{A})})}{d} \right\|_2>\frac{8(t\|\mathfrak{Im}(\bm{A})\|_2+\sqrt{t}\|\mathfrak{Im}(\bm{A})\|_{F})}{d(1-2\sqrt{t/d})}
    \end{equation*}
    happens. Then we get the probability by the above Lemma~\ref{lemma: HW}. 
\end{proof}

We then extend the above lemma to the case where the perturbation is complex. We first consider the case $r=1$.

\begin{lemma}\label{lemma: Stiefel Hanson-Wright, complex perturbation, r=1} (Stiefel Hanson-Wright, complex perturbation, $r=1$) For a Hermitian matrix $\boldsymbol{A}\in\mathbb{C}^{d\times d}$, it holds that
\begin{equation*}
    \mathbb{P}_{\bm{u} \sim \operatorname{Stief}_\mathbb{C}(d, 1)}\left(\left| \bm{u}^*\bm{A} \bm{u}-\operatorname{tr}(\bm{A}) / d \right| > \frac{16(t\|\bm{A}\|_2+\sqrt{t}\|\bm{A}\|_{F})}{d(1-\sqrt{2t/d})}\right) \leq 6e^{2.2-t}.
\end{equation*}
\end{lemma}
\begin{proof}[Proof]
    Without loss of generality, we can assume $\bm{A}$ to be a real diagonal matrix. This is because we can apply the eigenvalue decomposition to $\bm{A}=\bm{C}^*\bm{D}\bm{C}$, where $\bm{C}$ is unitary and $\bm{D}$ is a real diagonal matrix. Further we have $\|\bm{A}\|_2=\|\bm{D}\|_2, \|\bm{A}\|_F=\|\bm{D}\|_F, \mathrm{tr}(\bm{A})=\mathrm{tr}(\bm{D})$. Meanwhile, $\bm{C}\bm{U}$ has the same distribution as $\bm{U}$, i.e., the uniform distribution over $\operatorname{Stief}_\mathbb{C}(d, r)$. Then by computation, we have
\begin{equation*}
    \begin{aligned}
        \left| \bm{u}^*\bm{A} \bm{u}-\operatorname{tr}(\bm{A}) / d \right|=&\left| \mathfrak{Re}^*(\bm{u})\bm{A} \mathfrak{Re}^*(\bm{u})+\mathfrak{Im}^*(\bm{u})\bm{A} \mathfrak{Im}^*(\bm{u})\right.\\
        &\left.+\mathbf{i}\mathfrak{Re}^*(\bm{u})\bm{A} \mathfrak{Im}^*(\bm{u})-\mathbf{i}\mathfrak{Im}^*(\bm{u})\bm{A} \mathfrak{Re}^*(\bm{u})-\operatorname{tr}(\bm{A}) / d \right|\\
        =&\left| \mathfrak{Re}^*(\bm{u})\bm{A} \mathfrak{Re}^*(\bm{u})+\mathfrak{Im}^*(\bm{u})\bm{A} \mathfrak{Im}^*(\bm{u})-\operatorname{tr}(\bm{A}) / d \right|\\
        =&\left| (\mathfrak{Re}(\bm{u})^T,\mathfrak{Im}(\bm{u})^T)\left(\begin{array}{cc}
           \bm{A}  & \mathbf{0} \\
            \mathbf{0} & \bm{A}
        \end{array}\right)\left(\begin{array}{c}
           \mathfrak{Re}(\bm{u})   \\
            \mathfrak{Im}(\bm{u}) 
        \end{array}\right) -\operatorname{tr}(2\bm{A}) / 2d \right|.
    \end{aligned}
\end{equation*}
Because $\bm{u} \sim \operatorname{Stief}_\mathbb{C}(d, 1)$, we know $\left(\begin{array}{c}
           \mathfrak{Re}(\bm{u})   \\
            \mathfrak{Im}(\bm{u}) 
        \end{array}\right)\sim \operatorname{Stief}(2d, 1)$. Then we get the probability by Lemma~\ref{lemma: Stiefel Hanson-Wright, complex matrix}.
\end{proof}

We then extend the above Lemma~\ref{lemma: Stiefel Hanson-Wright, complex perturbation, r=1} to the case of general $k$.
\begin{lemma}\label{lemma: Stiefel Hanson-Wright, complex perturbation, general $r$} (Stiefel Hanson-Wright, complex perturbation, general $r$) For a Hermitian matrix $\boldsymbol{A}\in\mathbb{C}^{d\times d}$, it holds that
\begin{equation*}
    \mathbb{P}_{\bm{U} \sim \operatorname{Stief}_\mathbb{C}(d, r)}\left(\left\| \bm{U}^*\bm{A} \bm{U}-\bm{I}_r\cdot\operatorname{tr}(\bm{A}) / d \right\|_2 > \frac{32(t\|\bm{A}\|_2+\sqrt{t}\|\bm{A}\|_{F})}{d(1-\sqrt{2t/d})}\right) \leq 6e^{4.4r+2.2-t}.
\end{equation*}
\end{lemma}

\begin{proof}[Proof]
Let $\bm{v}$ be any fixed vector in $\mathcal{S}^{r-1}$, then we know $\bm{U}\bm{v}\stackrel{unif}{\sim}\operatorname{Stief}_{\mathbb{C}}(d, 1)$. By Lemma~\ref{lemma: Stiefel Hanson-Wright, complex perturbation, r=1}, we know
\begin{equation*}
    \mathbb{P}_{\bm{U} \sim \operatorname{Stief}_\mathbb{C}(d, r)}\left(\left| \bm{v}^*\bm{U}^*\bm{A} \bm{U}\bm{v}-\operatorname{tr}(\bm{A}) / d \right|_2 > \frac{16(t\|\bm{A}\|_2+\sqrt{t}\|\bm{A}\|_{F})}{d(1-\sqrt{2t/d})}\right) \leq 6e^{2.2-t}.
\end{equation*}

Let $\mathcal{N}$ be a $1/4$-net of $\mathcal{S}_\mathbb{C}^{r-1}$, i.e., a $1/4$-net of $\mathcal{S}^{2r-1}$ (by taking the real and imaginary parts apart). Then by Exercise 4.4.3 in \cite{Vershynin2018hdp}, we know
\begin{equation*}
    \left\|\bm{U}^* \bm{A} \bm{U}-\bm{I}_r \cdot \operatorname{tr}(\bm{A}) / d\right\|_2\le 2\sup_{\bm{u}\in\mathcal{S}_\mathbb{C}^{d-1}}\bm{u}^*\left(\bm{U}^* \bm{A} \bm{U}-\bm{I}_r \cdot \operatorname{tr}(\bm{A}) / d\right)\bm{u},
\end{equation*}
because the Hermitian case is equivalent with the real symmetric case. By Corollary 4.2.13 in \cite{Vershynin2018hdp}, we know $|\mathcal{N}| \leq\left(1+\frac{2}{1 / 4}\right)^{2r}\leq \exp (4.4 r)$. Then we obtain the union-bound
\begin{equation*}
\begin{aligned}
    &\mathbb{P}_{\bm{U} \sim \operatorname{Stief}_\mathbb{C}(d, r)}\left(\left\| \bm{U}^*\bm{A} \bm{U}-\bm{I}_r\cdot\operatorname{tr}(\bm{A}) / d \right\|_2 > \frac{32(t\|\bm{A}\|_2+\sqrt{t}\|\bm{A}\|_{F})}{d(1-\sqrt{2t/d})}\right)\\
    \le & \mathbb{P}_{\bm{U} \sim \operatorname{Stief}_\mathbb{C}(d, r)}\left(\sup _{\bm{v} \in \mathcal{N}}\left| \bm{v}^*\bm{U}^*\bm{A} \bm{U}\bm{v}-\operatorname{tr}(\bm{A}) / d \right|_2 > \frac{16(t\|\bm{A}\|_2+\sqrt{t}\|\bm{A}\|_{F})}{d(1-\sqrt{2t/d})}\right)\\
    \le & 6|\mathcal{N}|e^{2.2-t}\le 6e^{4.4r+2.2-t}.
\end{aligned}    
\end{equation*}

\end{proof}

Finally, we extend the Hermitian case to the general case.
\begin{lemma}
\label{lemma: HW general}(Stiefel Hanson-Wright, general case) For a matrix $\boldsymbol{A}\in\mathbb{C}^{d\times d}$, it holds that
    \begin{equation*}
        \mathbb{P}_{\boldsymbol{U} \sim \operatorname{Stief}(d, r)}\left(\left\| \boldsymbol{U}^*\boldsymbol{A} \boldsymbol{U}-\bm{I}_r\cdot\operatorname{tr}(\boldsymbol{A}) / d \right\|_2 > \frac{16(t\|\bm{A}\|_2+\sqrt{t}\|\bm{A}\|_{F})}{d(1-2\sqrt{t/d})}\right) \leq 6e^{2.2r-t},
    \end{equation*}
    and
    \begin{equation*}
        \mathbb{P}_{\boldsymbol{U} \sim \operatorname{Stief}_\mathbb{C}(d, r)}\left(\left\| \boldsymbol{U}^*\boldsymbol{A} \boldsymbol{U}-\bm{I}_r\cdot\operatorname{tr}(\boldsymbol{A}) / d \right\|_2 > \frac{32(t\|\bm{A}\|_2+\sqrt{t}\|\bm{A}\|_{F})}{d(1-2\sqrt{t/d})}\right) \leq 6e^{4.4r+2.2-t}.
    \end{equation*}
\end{lemma}

\begin{proof}[Proof]
    The corollary is direct by the triangle inequalities. In detail, for the first inequality, we have
    \begin{equation*}
        \begin{aligned}
            & \mathbb{P}_{\boldsymbol{U} \sim \operatorname{Stief}(d, r)}\left(\left\| \boldsymbol{U}^*\boldsymbol{A} \boldsymbol{U}-\bm{I}_r\cdot\operatorname{tr}(\boldsymbol{A}) / d \right\|_2 > \frac{16(t\|\bm{A}\|_2+\sqrt{t}\|\bm{A}\|_{F})}{d(1-2\sqrt{t/d})}\right)\\
            \le & \mathbb{P}_{\boldsymbol{U}}\left(\left\| \boldsymbol{U}^*\boldsymbol{A} \boldsymbol{U}-\bm{I}_r\cdot\operatorname{tr}(\boldsymbol{A}) / d \right\|_2 > \frac{16(t\|(\bm{A}+\bm{A}^*)/2\|_2+\sqrt{t}\|(\bm{A}+\bm{A}^*)/2\|_{F})}{d(1-2\sqrt{t/d})}\right)\\
            \le & \mathbb{P}_{\boldsymbol{U}}\left(\left\| \boldsymbol{U}^*\left(\frac{\bm{A}+\bm{A}^*}{2}\right) \boldsymbol{U}-\bm{I}_r\cdot\operatorname{tr}\left(\frac{\bm{A}+\bm{A}^*}{2}\right) / d \right\|_2 > \frac{8(t\|\bm{A}+\bm{A}^*\|_2+\sqrt{t}\|\bm{A}+\bm{A}^*\|_{F})}{d(1-2\sqrt{t/d})}\right)\\
            \le & 6e^{2.2r-t}.
        \end{aligned}
    \end{equation*}
    The second inequality holds for the same reason.
\end{proof}

\subsection{Other Technical Lemmas}
\begin{thm}\label{thm: strong circular law}(Circular law, Theorem 1.2 in \cite{tao2008random})
    Assume that the $\left(W_{ij}\right)_{i,j=1,\ldots,d}$ are i.i.d.\ complex random variables with
zero mean and finite $(2+\eta)^{th}$ moment for some $\eta>0$, and strictly positive variance. Then the strong circular law holds for $\bm{W}=[W_{ij}]$. In particular, denote
\begin{equation*}
    \mu_d(s, t):=\frac{1}{d} \#\left\{k \leq d \mid \operatorname{Re}\left(\lambda_k\right) \leq s ; \operatorname{Im}\left(\lambda_k\right) \leq t\right\}.
\end{equation*}
Then with probability 1, it converges to the uniform distribution
\begin{equation*}
    \mu(s, t):=\frac{1}{\pi} \operatorname{mes}\left\{z\in \mathbb{C} \mid |z|\le 1; \operatorname{Re}\left(z\right) \leq s ; \operatorname{Im}\left(z\right) \leq t\right\},
\end{equation*}
over the unit disk as $d$ tends to infinity.
\end{thm}

\begin{thm}\label{thm: perbutation circular law}
    (Circular law for low rank perturbations of i.i.d.\ matrices, Corollary 1.17 in \cite{2010RANDOM}) Let the $\bm{W}_n$ be  i.i.d.\ random matrices, and for each $n$, let $\bm{C}_n$ be a deterministic matrix with rank $o(n)$ obeying the Frobenius norm bound
    \begin{equation*}
        \|\bm{C}_n\|_F=\mathcal{O}(n^{1/2}).
    \end{equation*}
    Then the empirical distribution of $\frac{1}{\sqrt{n}}\bm{W}_n+\bm{C}_n$ converges both in probability and in the almost sure sense to the circular measure $\mu$.
\end{thm}

\begin{thm}\label{thm: universality from a random base matrix, Theorem 1.17 in [TV07])}
(Theorem 1.17 in \cite{tao10random}). Let $x$ and $y$ be complex random variables with zero mean and unit variance. Let $X_n=\left(x_{i j}\right)_{1 \leq i, j \leq n}$ and $Y_n=\left(y_{i j}\right)_{1 \leq i, j \leq n}$ be $n \times n$ random matrices whose entries are i.i.d./ copies of $x$ and $y$, respectively. For each $n$, let $M_n$ be a random $n \times n$ matrix, independent of $X_n$ or $Y_n$, such that $\frac{1}{n^2}\left\|M_n\right\|_2^2$ is bounded in probability. Let $A_n:=M_n+X_n$ and $B_n:=M_n+Y_n$. Then the ESD difference $\mu_{\frac{1}{\sqrt{n}} A_n}-\mu_{\frac{1}{\sqrt{n}} B_n}$ converges in probability to zero. If we furthermore assume that $\frac{1}{n^2}\left\|M_n\right\|_2^2$ is almost surely bounded, and $\mu_{\left(\frac{1}{\sqrt{n}} M_n-z I\right)\left(\frac{1}{\sqrt{n}} M_n-z I\right)^*}$ converges almost surely to some limit for almost every $z$, then $\mu_{\frac{1}{\sqrt{n}} A_n}-\mu_{\frac{1}{\sqrt{n}} B_n}$ converges almost surely to zero.
\end{thm}

To utilize the above theorem conveniently, we consider the $M_n$ to be the low-rank perturbation $\bm{U}^*\Lambda\bm{U}$ we constructed above. We have the following lemma
\begin{lemma}\label{lemma: bound the perturbation}
    Let $\bm{M}_d\colon =\sqrt{d}\bm{U}^*\Lambda\bm{U}$, where $\bm{\Lambda}=\operatorname{diag}\left(\lambda_1, \lambda_2,\dots,\lambda_r\right)$ with $|\lambda_1|\ge|\lambda_2|\ge\cdots\ge|\lambda_r|\ge 0$ for fixed $r\in\mathbb{N}^+$, and $\bm{U}\stackrel{\text{unif}}{\sim}\operatorname{Stief}(d, r)$ or $\bm{U}\stackrel{\text{unif}}{\sim}\operatorname{Stief}_\mathbb{C}(d, r)$. Then $\frac{1}{d^2}\left\|\bm{M}_d\right\|_2^2$ is almost surely bounded, and $\mu_{\left(\frac{1}{\sqrt{d}} \bm{M}_d-z I\right)\left(\frac{1}{\sqrt{d}} \bm{M}_d-z I\right)^*}$ converges almost surely to some limit for almost every $z\in\mathbb{C}$.
\end{lemma}
\begin{proof}[Proof]
    Once $\bm{M}_d\colon =\sqrt{d}\bm{U}^*\Lambda\bm{U}$, the result $\frac{1}{d^2}\left\|\bm{M}_d\right\|_2^2$ is bounded is trivial. Further, we have
    \begin{equation*}
        \left(\frac{1}{\sqrt{d}} \bm{M}_d-z \bm{I}\right)\left(\frac{1}{\sqrt{d}} \bm{M}_d-z \bm{I}\right)^* = \bm{U}^*\left(\Lambda^2-(z+\Bar{z})\Lambda+|z|^2\bm{I}\right)\bm{U},
    \end{equation*}
    which has deterministic ESD for every $z$. Thus $\mu_{\left(\frac{1}{\sqrt{d}} \bm{M}_d-z I\right)\left(\frac{1}{\sqrt{d}} \bm{M}_d-z I\right)^*}$ converges to some limit  almost surely for almost every $z\in\mathbb{C}$.
\end{proof}

\begin{lemma}\label{lemma: theorem 5.17}
Let $\bm{W}$ be an ${d{\times} d}$ complex matrix of i.i.d.\ entries with mean zero and variance $\sigma^2$. If $\mathbb{E}\left|W_{11}^4\right|<\infty$, then
\[
\limsup _{d \rightarrow \infty}\left\|\left(\frac{1}{\sqrt{d}} \bm{W}\right)^k\right\|_2 \leq(1+k) \sigma^k \text {, a.s. }
\]
for all $k\in \mathbb{N}^+$. Further if $W_{11}$ has finite $n$-moments for all $n\in\mathbb{N}^+$, then for each $\frac{1}{2}>\epsilon>0$, there are constants $A(\epsilon), B(\epsilon)>0$, such that
\[\mathbb{P}\left(\left|\lambda_1\left(\frac{1}{\sqrt{d}} \bm{W}_d\right)\right| \leq 1+\epsilon \right)\ge 1-Ad^{-B}.\]
\end{lemma}
\begin{proof}[Proof]
    The first result is Theorem 5.17 in \cite{bai2010spectral}. The second result follows by choosing proper $m=\frac{1}{\epsilon^2}$ such that $(m+2)<(1+\epsilon)^m$. Then we have
    \begin{equation*}
    \begin{aligned}
        \mathbb{P}\left(\left|\lambda_1\left(\frac{1}{\sqrt{d}} \bm{W}_d\right)\right| \leq  1+\epsilon \right)&\le \mathbb{P}\left(\left|\lambda_1\left(\frac{1}{\sqrt{d}} \bm{W}_d\right)\right| \leq  (m+2)^{1/m} \right)\\
        &=\mathbb{P}\left(\left|\lambda_1^m\left(\frac{1}{\sqrt{d}} \bm{W}_d\right)\right| \leq m+2 \right)\\
        &\ge \mathbb{P}\left(\left\|\left(\frac{1}{\sqrt{d}} \bm{W}_d\right)^m\right\| \leq m+2 \right)\\
        &\stackrel{(a)}{\ge}  1-A(m)d^{-B(m)} = 1-A\left(\frac{1}{\epsilon^2}\right)d^{-B\left(\frac{1}{\epsilon^2}\right)}.
    \end{aligned}
    \end{equation*}
    The inequality $(a)$ is by Theorem 1.4 in \cite{tao2014outliers}.
\end{proof}

\begin{lemma}\label{lemma: largest eigenvalue of X}
    (Theorem 5.18 in \cite{bai2010spectral}) Let $\bm{W}$ be a ${d{\times} d}$ complex matrix of i.i.d.\ entries with mean zero and variance $\sigma^2$. If $\mathbb{E}\left|\bm{W}_{11}^4\right|<\infty$, then
\[
\limsup _{d \rightarrow \infty}\left|\lambda_{\max }\left(\frac{1}{\sqrt{d}} \bm{W}\right)\right| \leq \sigma, \text { a.s.. }
\]
\end{lemma}

\begin{lemma}\label{lemma: ratio of gaussian}
    (Lemma 4.6 in \cite{DBLP:journals/corr/abs-1804-01221}). Let $\mathbb{P}$ denote the distribution $\mathcal{N}\left(\mu_1, \Sigma\right)$ and $\mathbb{Q}$ denote $\mathcal{N}\left(\mu_2, \Sigma\right)$, where $\mu_1, \mu_2 \in$ $(\operatorname{ker} \Sigma)^{\perp}$. Then for $\eta>0$, we have
    \begin{equation*}
    \mathbb{E}_{\mathbb{Q}}\left[\left(\frac{\mathrm{d} \mathbb{P}}{\mathrm{d} \mathbb{Q}}\right)^{1+\eta}\right]=\exp \left(\frac{\eta(1+\eta)}{2}\left(\mu_1-\mu_2\right)^{\top} \Sigma^{\dagger}\left(\mu_1-\mu_2\right)\right).
    \end{equation*}
\end{lemma}

\begin{lemma} (Lemma D.2 in \cite{DBLP:journals/corr/abs-1804-01221})\label{lemma: D.2}
For $\lambda>0$, $k\in\mathbb{N}$, we have
\[
\max _{k \geq 0} \lambda^{-4 k}(k+1) \leq 1+\frac{1}{4 e \log \left(\lambda\right)} \quad \text { and } \quad \max _{k \geq 1} \lambda^{-4 k} \log (1+k) \leq \frac{1}{4 e \log \lambda}.
\]    
\end{lemma}

\bibliographystyle{abbrv}
\bibliography{reference}

\end{document}